\documentclass[11pt]{article}
\usepackage[T1]{fontenc}
\usepackage{subcaption}
\usepackage[colorinlistoftodos,textsize=footnotesize]{todonotes}
\usepackage{amsfonts,amsmath,amssymb}
\usepackage{hyperref}
\usepackage{mathtools}
\usepackage[utf8]{inputenc}
\usepackage[caption=true]{subfig}
\usepackage[ruled, vlined]{algorithm2e}
\usepackage{enumitem}
\usepackage{graphicx}
\usepackage{bbold}
\usepackage{bbm}
\usepackage{titlesec}
\usepackage{setspace} 
\usepackage[bottom]{footmisc}
\usepackage[scaled]{helvet}
\usepackage[round]{natbib}
\usepackage{float}

\usepackage{setspace}

\usepackage{caption}
\usepackage{subcaption}
\usepackage{setspace,graphicx,epstopdf,amsmath,amsfonts,amssymb,amsthm,versionPO}
\usepackage{marginnote,datetime,enumitem,subfig,rotating,fancyvrb}

\usdate 

\excludeversion{notes}		
\includeversion{links}          

\usepackage[margin=1in]{geometry}

\makeatletter\let\chapter\@undefined\makeatother 

\usepackage{indentfirst} 
\usepackage{endnotes}    
\usepackage[labelfont=bf,labelsep=period]{caption}   
\usepackage{float}

\mathtoolsset{showonlyrefs}
\allowdisplaybreaks

\newcommand{\E}{{\mathbb{E}}}
\newcommand{\F}{{\mathcal{F}}}
\newcommand{\A}{{\mathcal{A}}}

\newcommand{\D}{{\mathfrak{D}}}

\newtheorem{proposition}{\textbf{Proposition}}
\newtheorem{assume}{\textbf{Assumption}}
\renewenvironment{proof}[1][Proof]{\begin{trivlist}
\item[\hskip \labelsep {\bfseries #1}]}{\end{trivlist}}

\newcommand{\diff}{\mathrm{d}}
\newcommand{\dt}{{\diff t}}
\newcommand{\ds}{{\diff s}}
\newcommand{\du}{{\diff u}}

\titleformat{\subsection}    
       {\fontsize{11}{17}\bfseries\itshape}{\thesubsection}{1em}{}

\definecolor{myblue}{rgb}{0.8,0.8,1}
\definecolor{myred}{rgb}{1,0.8,0.8}
\definecolor{mygreen}{rgb}{0.8,1,0.8}
\definecolor{mygrey}{rgb}{220,220,220}

\DeclareMathAlphabet{\xcal}{OMS}{cmsy}{m}{n}

\definecolor{dbblue}{RGB}{10,65,155}				
\definecolor{dbred}{RGB}{215,0,50}					
\definecolor{blue}{RGB}{0,113.9850,188.9550} 			
\definecolor{red}{RGB}{216.7500,82.8750,24.9900} 		
\definecolor{green}{RGB}{118.8300,171.8700,47.9400} 	
\definecolor{grey}{RGB}{110,110,110}				
\definecolor{lgrey}{RGB}{210,210,210}				
\definecolor{c1}{RGB}{0,113.9850,188.9550}			
\definecolor{c2}{RGB}{216.7500,82.8750,24.9900}		
\definecolor{c3}{RGB}{236.8950,176.9700,31.8750}		
\definecolor{c4}{RGB}{125.9700,46.9200,141.7800}		
\definecolor{c5}{RGB}{118.8300,171.8700,47.9400}		
\definecolor{c6}{RGB}{76.7550,189.9750,237.9150}		
\definecolor{c7}{RGB}{161.9250,19.8900,46.9200}		
\definecolor{c14}{RGB}{0,102,102}

\hypersetup{colorlinks,urlcolor=blue,citebordercolor=blue,linkbordercolor=blue,filecolor=blue,filebordercolor=blue,citecolor=blue,linkcolor=blue,colorlinks=true}

\newtheorem{theorem}{\textbf{Theorem}}

\newtheorem{lemma}{\textbf{Lemma}}
\newtheorem{remark}{\textbf{Remark}}
\newtheorem{corollary}{\textbf{Corollary}}

\def\E{\mathbb{E}}

\makeatletter
\def\ps@pprintTitle{%
  \let\@oddhead\@empty
  \let\@evenhead\@empty
  \def\@oddfoot{\reset@font\hfil\thepage\hfil}
  \let\@evenfoot\@oddfoot
}
\makeatother

\usepackage{pgfplots}
\pgfplotsset{compat=1.18}
\usetikzlibrary{arrows.meta, positioning, shapes.multipart, decorations.pathreplacing}

\allowdisplaybreaks

\makeatletter
\providecommand\sf@counterlist{}
\makeatother



\begin{document}
\setlist{noitemsep} 
\onehalfspacing   

\author{Fayçal Drissi\thanks{\rm F. Drissi is with the Oxford-Man Institute, University of Oxford.  Email: \href{mailto:faycal.drissi@omi.ox.ac.uk}{faycal.drissi@omi.ox.ac.uk}. } 
\and Sebastian Jaimungal\thanks{\rm S. Jaimungal is with the Department of Statistical Sciences, University of Toronto and  the Oxford-Man Institute, University of Oxford.   Email:  \href{mailto:sebastian.jaimungal@utoronto.ca}{sebastian.jaimungal@utoronto.ca}.} \and Xuchen Wu\thanks{\rm X. Wu is with the Department of Mathematics, University of Toronto.   Email: \href{mailto:xuchen.wu@mail.utoronto.ca}{xuchen.wu@mail.utoronto.ca}.
\\
\textbf{Acknowledgments}: We are grateful to Álvaro Cartea, Patrick Chang, Zachary Feinstein, participants at the Research in Options 2025 and DeFi \& Crypto 2026 conferences  for insightful discussions. SJ would like to acknowledge support from the Natural Sciences and Engineering Research Council of Canada through grant RGPIN-2024-04317.}}

\title{\textbf{
Equilibrium Liquidity and Risk Offsetting in Decentralised Markets
}}



\maketitle

\thispagestyle{empty}

\bigskip

\begin{center}
\small
\begin{minipage}{\textwidth}
\centerline{\bf ABSTRACT}

We study the economic viability of liquidity provision in decentralised exchanges (DEXs) within a structural framework in which market outcomes are endogenous. We formulate strategic interactions as a sequential game: a risk-averse liquidity provider (LP) sets the supply of liquidity in the DEX and a costly dynamic replication strategy in a centralised exchange (CEX), price-sensitive traders determine trading volumes, and arbitrageurs align prices. We establish existence of equilibrium under general trading functions. We show that DEX liquidity depth is a central instrument for risk management, because the LP adjusts liquidity ex ante to manage exposure. In addition to the classical trade-off between liquidity demand and adverse selection, we identify two further determinants of the viability of liquidity provision: the ratio of risk aversion to replication costs and private information. The ratio governs the aggressiveness of replication: greater relative risk aversion reduces risk but also lowers equilibrium liquidity and its mean profitability. Private information has a non-monotonic effect. For moderate price movements, speculative benefits increase liquidity. For large price movements, anticipated adverse selection and replication costs lead to thinner markets.

\end{minipage}

\end{center}

\section{Introduction}

Blockchains provide the technology to operate  decentralised exchanges (DEXs), which have become a central component of decentralised finance,\footnote{See \cite{cong2019blockchain,john2023smart,harvey2024evolution}.} with monthly trading volumes reaching $\$420$ billion in 2025. The widespread adoption of DEXs hinges on whether liquidity provision remains viable in the current market structure, where DEXs operate alongside centralised exchanges (CEXs) that dominate price discovery and offer continuous trading.

The extant literature abstracts from key economic and practical mechanisms when evaluating the returns and risks of liquidity provision: the size of liquidity supply and trading conditions are typically treated as exogenous, and profitability is assessed under the assumption that the risk of DEX positions can be replicated without friction in a CEX. This paper departs from this approach. We develop a structural model in which DEX market outcomes are endogenised. Our analysis characterises how the size, rewards, and risks of DEX liquidity provision arise from the interplay between risk preferences, costly risk management, fundamental price volatility, access to private information, and strategic liquidity demand.

We formalise these mechanisms in a tractable partial equilibrium framework with three types of agents: a representative (monopolist) liquidity provider, price-sensitive noise liquidity takers (LTs), and arbitrageurs. Interactions proceed in three stages. In stage one, the LP chooses liquidity depth by trading off anticipated fee revenue against adverse selection costs, taking into account the expected proceeds and costs of replication in the CEX. In stage two, the LP determines a dynamic trading strategy in the CEX by solving a continuous-time control problem under inventory risk aversion, trading costs, and private information about future prices. Trading costs discourage replication, risk aversion encourages it, and private information incentivises directional speculation and informs about future inventory exposure. In stage three, noise LTs arrive stochastically and choose optimal trade sizes given prevailing trading costs in the DEX, while arbitrageurs align the  price in the DEX with fundamentals. The model is solved by backward induction.

Our first main finding is that the scale of liquidity is as important a risk management tool as dynamic replication. 

We find that the intensity of replication per unit of capital supplied is determined by the ratio of risk aversion to trading costs. Risk aversion encourages replication to reduce net exposure, whereas trading costs discourage it. For a fixed ratio, higher absolute levels of both risk aversion and trading costs reduce the optimal liquidity supply, because the LP anticipates the frictions associated with trading in the CEX. In equilibrium, the LP preserves the viability of DEX liquidity provision primarily by reducing liquidity supply and exposure ex ante.

When setting the scale of liquidity, the LP also anticipates that adverse selection costs increase with supplied capital, that revenue from price-sensitive demand increases with supplied capital, and that expected price movements can induce reserve adjustments, inventory exposure, or directional profits. Effects that increase profitability increase equilibrium liquidity, whereas effects that increase costs reduce it. Equilibrium liquidity therefore reflects the balance between these forces.

Our second contribution is to show how the determinants that govern equilibrium liquidity supply also shape the distribution of LP returns. 

Greater relative risk aversion induces aggressive replication, which reduces return variance by reducing price exposure, but also lowers mean returns because replication is costly. Two limiting cases arise. When risk aversion is sufficiently high relative to trading costs, the optimal strategy in the CEX approaches perfect replication of DEX reserves and equilibrium DEX liquidity falls to the lowest level consistent with nonnegative expected returns. At the other extreme, under risk neutrality, the LP allocates the entire available budget to liquidity provision provided  expected fee revenue from price-sensitive demand exceeds expected adverse selection costs. Otherwise, liquidity collapses and the DEX shuts down.

The role of liquidity demand in shaping the distribution of LP returns is described next. When liquidity demand is weak, LPs act defensively. This reduces both expected LP returns and their variance. When demand is strong, LPs supply more liquidity and accept greater exposure by trading less in the CEX, because higher fee revenue compensates for the additional risk on average. This increases both expected returns and their variance. Finally, consistent with the existing literature, higher fundamental volatility increases adverse selection costs, reduces expected returns, and increases their dispersion.

Our last finding concerns private information. Access to information about future price movements does not systematically increase liquidity or returns. Private signals generate speculative gains, but they also imply larger expected reserve adjustments and higher anticipated replication costs. For moderate expected price movements, speculative benefits dominate and liquidity increases. For sufficiently large expected drifts, however,  anticipated adverse selection and replication costs outweigh speculative gains, leading to thinner markets.

Our theoretical contribution is to develop a tractable structural model that endogenises the viability of liquidity provision and DEX market outcomes. In particular, we employ variational tools to reduce the LP's problem in the CEX to a coupled system of forward and backward stochastic differential equations. We show that the system admits a closed-form representation under general convex DEX trading functions and stochastic price signals. The optimal strategy decomposes into a replication component, governed by the ratio of risk aversion to trading costs, and a speculative component driven by private information.

Moreover, we establish existence of equilibrium DEX liquidity in stage one for general trading functions and we derive closed-form expressions in the case of constant product markets. In the latter setting, we further obtain an explicit decomposition of equilibrium liquidity into (i) a component that balances fee revenue, adverse selection, and price risk, and (ii) a speculative component that exploits private information.

Finally, we extend the model to a setting in which the LP's trading activity in the CEX is sufficiently large to affect prices and generate transient price impact. In this case, the LP's optimal replication problem reduces to a differential Riccati equation (DRE), whose solution exists, is unique, and can be computed efficiently. This characterisation ensures well-posedness of the replication problem and allows us to establish existence of equilibrium liquidity in the DEX.

\textbf{Literature review.} Our paper is part of the literature that uses quantitative methods to study the viability of DEXs. \cite{milionis2022automated} provide a foundational work introducing loss-versus-rebalancing (LVR) 
which describes the losses in DEXs when LPs hedge their price exposure in a continuous market.  \cite{cartea2023predictable,cartea2024decentralized} study strategic liquidity provision and discuss specificities of concentrated liquidity markets. \cite{milionis2024automated} consider the frictions of fees in arbitraging DEXs. \cite{bichuch2024defi} study replication of DEX liquidity positions. These works consider frictionless markets for LPs who manage the risk of their liquidity supply.

Numerous works explore the microstructure of DEXs.  
\cite{angeris2021replicating2, capponi2023decentralized,fabi2025economics} show that DEXs generate losses for LPs; see also \cite{agarwal2025optimal,echenim2024quantitative}.   
\cite{jaimungal2023optimal,miori2024clustering} study liquidity taking in DEXs.  \cite{lehar2021decentralized} describe competition between DEXs and order books. 
\cite{hasbrouck2022need} study  DEX fees.  \cite{bichuch2025axioms} formalise the axioms governing DEX design.  
\cite{klein2023price} examine the role of informed liquidity supply.  \cite{park2023conceptual} discuss the types of trading costs in DEXs. 
\cite{malinova2024learning} investigate the potential of DEXs to organise equity trading.  \cite{goyal2023finding,cartea2024strategic, he2024optimal} propose DEX designs to mitigate LP losses.\footnote{Recent works examine optimal behavior of LPs and optimal dynamic fee structures; see \cite{bergault2025optimal, baggiani2025optimal}.  \cite{campbell2025optimal} discuss the costs of replication in the CEX.   Finally, \cite{capponi2025longer, he2025arbitrage} characterise the microstructure of DEXs within the consensus protocol of blockchains.}

The strategy of the LP in stage two of our model relates to the literature on trading using stochastic control tools.\footnote{See \cite{cartea2015algorithmic}, \cite{gueant2016financial}, and \cite{donnelly2022optimal}.}  
Signals are introduced in \cite{cartea2016incorporating}. Latent models with signals are in \cite{casgrain2019trading}, while variational approaches to trading with multiple heterogeneous agents are in \cite{casgrain2018mean, casgrain2020mean, wu2024broker}.  Finally, inventory targets are analysed in \cite{cartea2016closed,bank2017hedging}.

The remainder of this paper proceeds as follows. Section~\ref{sec:DEXs} describes the economic trade-offs faced by LPs and introduces the general features of the model.  
Section~\ref{sec:LTs} solves for the trading volumes of noise LTs in stage three. 
Section~\ref{sec:hedging} solves the replication problem of the LP in stage two.  
Section~\ref{sec:liqprovuniswap} derives the optimal liquidity supply in stage one.  
Section~\ref{sec:numerics} examines the equilibrium in the case of a constant product market such as Uniswap and presents numerical experiments. Finally, Appendix \ref{apx:proofs} collects the proofs, and Appendix \ref{sec:transient} extends our model to transient impact.

\section{General features of the model}\label{sec:DEXs}

We consider a decentralised exchange (DEX) that operates as a liquidity pool on a blockchain: a liquidity provider (LP) deposits reserves into a liquidity pool, and execution prices for liquidity takers (LTs) are determined algorithmically as a function of these reserves according to a pricing rule. The pricing rule maps reserves into execution prices. Therefore, it governs both trading costs for LTs and risks and returns for LPs. This section describes the mechanics of price and liquidity in DEXs, and introduces the general features of our model.

The mathematical foundations of DEXs have been extensively studied in the literature. We refer for instance to \cite{bichuch2025axioms}, and references therein, who provide a general axiomatic framework for DEXs and formalise their properties. Rather than revisiting these fundamental results, we introduce below the specific notation and assumptions required for our model to characterise equilibrium market outcomes in DEXs in a tractable manner.

Consider a DEX for a pair of assets $\{X, Y\}$, where $X$ is a reference asset used by  agents to value their wealth, and $Y$ is a risky asset.  A representative LP deposits initial reserves $X_0$ and $Y_0$ of assets $X$ and $Y$ into the pool at time $t=0.$ The LP commits to maintaining this liquidity position until a terminal horizon $T$, i.e., the LP neither adds reserves to the pool nor withdraws reserves from the pool over the interval $[0,T]$.

We motivate this assumption as follows. In contrast to traditional markets with continuous trading, where market makers can adjust quoted spreads at negligible cost, modifying liquidity in a DEX requires submitting on-chain transactions that incur gas fees. These fees make frequent rebalancing of liquidity positions economically costly. In addition, clearing on blockchains is discrete, so liquidity adjustments can generally occur once per block.\footnote{Liquidity adjustments may also be executed strategically through bot activity aimed at extracting miner extractable value (MEV). Blockchain designers are actively developing architectures that mitigate or eliminate MEV, so our theoretical model abstracts from these strategic frictions. We refer the reader to \cite{wan2022just, capponi2023paradox} and \cite{daian2020flash, oz2024wins} for foundational analyses of these topics.} As a consequence, liquidity positions are typically maintained for several consecutive blocks; see \cite{cartea2025decentralised} for empirical evidence on the duration of liquidity provision.\footnote{In Ethereum, the most widely used blockchain for DEXs, a block is produced approximately every $12$ seconds.}

As trading unfolds over the interval $[0,T]$, the reserves in the DEX act as counterparty to trades by liquidity takers (LTs). Each transaction modifies the composition of reserves, so the quantities of both assets adjust over time. We denote by $(X_t)_{t \ge 0}$ and $(Y_t)_{t \ge 0}$ the processes describing the evolution of reserves in assets $X$ and $Y$, respectively.

\subsection{DEX price and liquidity.}

The mechanics of price and liquidity in DEXs are defined by \textit{iso-liquidity curves}.  Once the LP establishes the pool, the reserves satisfy, for all $t \in [0, T]$,
\begin{equation}\label{trading_condi}
    f(X_t, Y_t) = \kappa^2 = f(X_0, Y_0)\,,
\end{equation}
where $\kappa > 0$ is the \emph{depth} of liquidity, and 
$f : (0, \infty)^2 \to (0, \infty)$ is the DEX's \emph{trading function}.

The function $f$ characterises the feasible combinations of reserves in assets $X$ and $Y$ that preserve the depth $\kappa$ of liquidity in the DEX. Thus, movements along the level set $f(X,Y)=\kappa^2$ correspond to trades that alter the composition of reserves without changing the scale of liquidity provision $\kappa$. For the analysis that follows, we impose the following assumptions.\footnote{$\partial_i f(\cdot,\cdot)$ denotes the first-order partial derivative with respect to the $i$th argument of $f$ for $i\in{1,2}$, and $\partial_{ii} f(\cdot,\cdot)$ denotes the corresponding second-order partial derivative.}

\begin{assume}\label{assume:1}
\begin{enumerate}[label=\textnormal{(\roman*)}, ref=\theassume(\roman*)]
    \item\label{assume:1:i}
    $f\in C^3((0,\infty)^2)$ and has strictly positive partial derivatives.
    
    \item\label{assume:1:ii}
    For each $y>0$, $f(\cdot,y):(0,\infty)\to (0,\infty)$ is surjective. Thus, for each $\kappa > 0$, the level set
    $f(x, y) = \kappa^2$ admits a unique solution
    $x = \varphi(y, \kappa)\,.$
    
    \item\label{assume:1:iii}
    $R := \dfrac{\partial_2 f}{\partial_1 f}$ satisfies
    $R\,\partial_1 R - \partial_2 R > 0$ everywhere, and is decreasing in $\kappa.$ 
    
    \item\label{assume:1:iv}
    $\partial_1\varphi$ satisfies the limits $\lim_{y\downarrow 0}\partial_1\varphi(y,\kappa)=-\infty\quad$ and $\quad\lim_{y\uparrow \infty}\partial_1\varphi(y,\kappa)=0\,.$
\end{enumerate}
\end{assume}

Assumption~\ref{assume:1:i} ensures that the  depth of liquidity $\kappa$ increases with the level of reserves deposited in the DEX.  We refer to $\varphi$ in Assumption~\ref{assume:1:ii} as the \emph{level function}.  
By the implicit function theorem, and since $f$ has strictly positive partial derivatives under Assumption~\ref{assume:1:i}, the mapping $\varphi$ is $C^3$ on $(0,\infty)^2$. Using the iso-liquidity condition \eqref{trading_condi}, and in the absence of additional liquidity injections or withdrawals, reserves in the reference asset \(X\) can be expressed as a function of reserves in the risky asset \(Y\) and the liquidity depth \(\kappa\) as
\begin{equation}\label{eq:levelfunc}
    X_t=\varphi(Y_t,\kappa).
\end{equation}

In a DEX, when an LT purchases a quantity~$\Delta y$ of the risky asset,  
the iso-liquidity condition \eqref{trading_condi}, or equivalently the identity \eqref{eq:levelfunc}, determines the payment~$\Delta x$ in the reference asset paid to the DEX. Following the trade, reserves must satisfy $X_t + \Delta x = \varphi(Y_t - \Delta y, \kappa)$. Accordingly, the execution price, i.e., the amount of the reference asset paid per unit of the risky asset purchased, is
\begin{equation}\label{eq:askprice}
    \frac{\Delta x}{\Delta y} 
    = \frac{\varphi(Y_t - \Delta y, \kappa) - X_t}{\Delta y} 
    = \frac{\varphi(Y_t - \Delta y, \kappa) - \varphi(Y_t, \kappa)}{\Delta y}.
\end{equation}
Similarly, if an LT sells a quantity $\Delta y$ of the risky asset, the execution price, i.e., the amount of the reference asset received per unit of the risky asset sold, is
\begin{equation}\label{eq:bidprice}
    \frac{\Delta x}{\Delta y} 
    = \frac{\varphi(Y_t, \kappa) - \varphi(Y_t + \Delta y, \kappa)}{\Delta y}.
\end{equation}

We illustrate the mechanism of iso-liquidity curves in Figure \ref{fig:iso}. As the traded quantity $\Delta y$ tends to zero, the execution prices in \eqref{eq:askprice}-\eqref{eq:bidprice} converge to the infinitesimal price $-\partial_1 \varphi(Y_t, \kappa),$ which we refer to as the \emph{marginal price}. The marginal price plays a role analogous to the midprice in limit order books and serves as the benchmark around which trading costs are measured.

\begin{figure}[H]
    \centering
    \includegraphics[width=.7\linewidth]{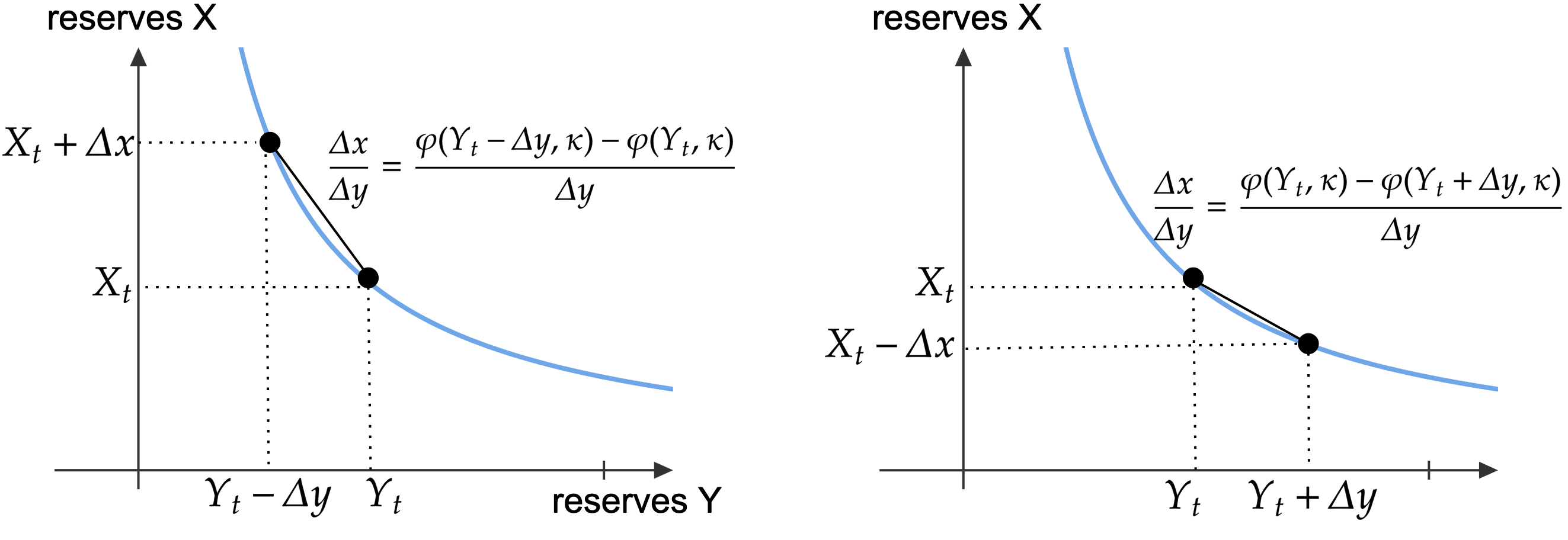}
    \caption{Illustration of how iso-liquidity curves map reserve levels into execution prices.}
    \label{fig:iso}
\end{figure}

The difference between the marginal price and the execution prices in \eqref{eq:askprice}-\eqref{eq:bidprice} captures the trading cost associated with executing a finite quantity in the DEX. These costs are given by
\begin{equation}\label{eq:costs}
    \frac{\varphi(Y_t - \Delta y, \kappa) - \varphi(Y_t, \kappa)}{\Delta y} 
    + \partial_1 \varphi(Y_t, \kappa)
    \qquad \text{and} \qquad
    \frac{\varphi(Y_t, \kappa) - \varphi(Y_t + \Delta y, \kappa)}{\Delta y}
    + \partial_1 \varphi(Y_t, \kappa)\,.
\end{equation}
These trading costs are strictly positive when $\varphi$ is convex in the reserves $Y_t$, which is ensured by Assumption~\ref{assume:1:iii}. This assumption further implies that the marginal price $-\partial_1 \varphi$ is strictly decreasing in reserves. Indeed,
\[
    \partial_1 \varphi(y, \kappa)
    = -R(\varphi(y, \kappa), y) \quad \text{and}\quad \partial_{11} \varphi(y, \kappa)
    = \partial_1 R(\varphi(y, \kappa), y) \, R(\varphi(y, \kappa), y)
       - \partial_2 R(\varphi(y, \kappa), y).
\]

Under our assumptions, the latter derivative is negative, which implies that the marginal price decreases as reserves in asset $Y$ increase. Economically, when an LT sells (resp. buys) the asset $Y$ to the pool, reserves in $Y$ increase (resp.~decrease), and the marginal price correspondingly decreases (resp.~increases).  Moreover, the convexity of the level function implies that the trading costs in~\eqref{eq:costs} increase with the traded quantity~$\Delta y$. Executing larger orders therefore requires moving further along the iso-liquidity curve, generating higher marginal costs. This effect is analogous to walking the book in limit order markets.

For economically valid specifications of the DEX trading function, the convexity term $\partial_{11}\varphi(y,\kappa)$ decreases with liquidity depth $\kappa$.  Higher reserves flatten the iso-liquidity curve and reduce execution costs, whereas lower reserves steepen the curve and make trading more expensive for LTs. This property of DEXs constitutes the  economic channel through which liquidity taking activity interacts with the LP's capacity choice in our model. Higher liquidity depth reduces trading costs, which attracts greater trading flow and fee revenue. However, as we show below, higher liquidity also increases exposure to arbitrage and inventory risk.

\subsection{The model} 

The aim of our model is to characterise equilibrium outcomes in the DEX, including the size and viability of liquidity supply and the resulting trading volumes. The setup is motivated by the current structure for many digital assets, where centralised platforms such as Binance account for the majority of trading volume and price discovery, while DEXs function primarily as secondary venues.\footnote{A model of joint price discovery across centralised and decentralised venues is beyond the scope of this paper. Such a framework would need to address differences in continuous versus discrete clearing and the role of priority fees. We refer the reader to \cite{malinova2024learning,aoyagi2025coexisting,capponi2025longer} for in-depth analyses of market structure differences and price discovery on blockchains.} Accordingly, we consider a DEX operating alongside a CEX with continuous trading where LPs manage inventory risk. In contrast to existing literature, and consistent with practice, risk is not managed frictionlessly but through costly dynamic trading and according to risk preferences and potentially private information.

This section introduces our partial equilibrium model. Strategic interactions in the DEX are formulated as a sequential competitive game.  There are three types of agents: a representative LP, arbitrageurs, and noise LTs. The LP deposits initial reserves $X_0$ and $Y_0$ of assets $X$ and $Y$ into the pool at time $t=0$. Over the investment horizon $[0, T]$, (i) the LP trades dynamically in the CEX to manage inventory exposure, and (ii) the two types of LTs interact with the DEX: arbitrageurs and noise LTs. Arbitrageurs align the pool's marginal price $-\partial_{1}\varphi(Y_t,\kappa)$ with the fundamental value of asset $Y$. Noise LTs with price-sensitive demand trade against the pool, generating fee revenue.

These interactions are formulated as a three-stage optimization problem and solved recursively by backward induction. In stage three, noise LTs take the liquidity depth in the DEX as given and solve a utility maximisation problem to determine their trading volumes. These volumes generate fee revenue for the LP. In stage two, the LP determines a dynamic strategy in the CEX, taking the liquidity level in the DEX as given. In stage one, anticipating both the response of liquidity demand and the costs of replication in the CEX, the LP chooses the initial liquidity depth. We now describe our assumptions and the economic tradeoffs at play in each stage in more detail.

\subsubsection{Stage one} In stage one, the LP chooses the liquidity depth $\kappa$, or equivalently the level of reserves deposited in the DEX. This choice determines the scale of inventory exposure. The LP anticipates two opposing forces: fee revenue generated by price-sensitive liquidity demand, and adverse selection costs arising from arbitrage-driven price adjustments.

First, we describe the dynamics of the LP's reserves in the DEX and the adverse selection costs due to arbitrage.\footnote{In the remainder of this paper, we work on a filtered probability space $(\Omega, \mathcal{F}, \mathbb{F} = (\mathcal{F}_t)_{t \in [0, T]}, \mathbb{P})$ satisfying the usual conditions.} Let $F_t$ denote the fundamental price of the risky asset in units of the reference asset $X$. We assume that arbitrageurs align the marginal price $-\partial_1 \varphi(Y_t,\kappa)$ in the DEX with the fundamental value $F_t$. 
Accordingly, $F_t = -\partial_1 \varphi(Y_t, \kappa)$. Under Assumption~\ref{assume:1:iv}, the mapping $-\partial_1 \varphi(\cdot,\kappa)$ is a $C^2$-diffeomorphism from $(0,\infty)$ to $(0,\infty)$. Thus, it admits a $C^2$ inverse $h(\cdot,\kappa)$:
\begin{equation}\label{eq:arbCEXDEX}
    F_t = -\partial_1 \varphi(Y_t, \kappa)
    \iff 
    Y_t = h(F_t, \kappa).  
\end{equation}

By It\^o's formula, the dynamics of the DEX reserves in units of the reference asset $X$ are
\begin{align}\label{eq:PL}
\diff(X_t + Y_t\,F_t)   & = \diff\left(\varphi(Y_t,\kappa) - Y_t\,\partial_1\varphi(Y_t,\kappa)\right) 
 = Y_t\,\diff F_t - \underbrace{\tfrac12 \partial_{11} \varphi(h(F_t,\kappa),\kappa)\,(\partial_{1} h(F_t,\kappa))^2\,\diff  [F]_t}_{\text{LVR, convexity cost}}\,.
\end{align}

Because arbitrageurs  align the DEX price continuously with fundamentals, fluctuations in $F_t$ translate directly into fluctuations in DEX reserves. The first term in~\eqref{eq:PL}, $Y_t\, \diff F_t$, is the exposure of the liquidity position to fundamental price movements. The second term is a predictable adverse selection component arising from the convexity of the iso-liquidity curve. This term is referred to as the loss-versus-rebalancing (LVR) (see \cite{milionis2022automated}), and it reflects the mechanical losses incurred as reserves adjust along a convex pricing schedule in response to arbitrageurs.

Crucially, the LVR component scales with DEX liquidity depth $\kappa$ and  price volatility $\sigma$. Deeper liquidity $\kappa$ also increases the magnitude of reserves $Y_t$ and therefore amplifies exposure to price fluctuations $\diff F_t$. When inventory risk management through dynamic trading is costly, greater exposure implies higher expected future replication costs. Consequently, both the LVR component and the inventory exposure term $Y_t\, \diff F_t$ create incentives to reduce liquidity provision $\kappa$.

In equilibrium, this incentive to reduce liquidity is balanced against expected fee revenue. We now describe how fee revenue is generated in DEXs. In addition to the trading costs induced by the convexity of the level function $\varphi$, liquidity takers pay a proportional fee $\pi \in (0,1)$ when transacting in the DEX. For a purchase of $\Delta y$ units of asset $Y$, an additional amount $\pi\,\Delta y\,F_t$ in the reference asset is paid to the pool. For a sale of $\Delta y$, a fraction $\pi\,\Delta y\,F_t$ of the proceeds is retained by the LP.

Fee revenue is therefore proportional to trading volume and depends directly on the depth of liquidity  in the DEX through the effect of $\kappa$ on trading costs. Because liquidity demand is price-sensitive, higher liquidity reduces execution costs and increases trading flow  and fee income. 

Thus, in stage one, the LP anticipates that liquidity depth simultaneously (i) amplifies inventory exposure, risk-offsetting costs, and adverse selection costs, and (ii) increases fee revenue. The equilibrium level of $\kappa$ at this stage reflects the balance between these opposing forces.

\subsubsection{Stage two}  

In stage two of the sequential optimization problem, the LP takes the liquidity depth $\kappa$ as given and determines an optimal trading strategy in the CEX based on risk preferences, trading costs, and private information. The LP maximises terminal wealth across both the CEX and the DEX, subject to convex inventory  costs.

Let $Q_t$ denote the LP's inventory in asset $Y$ held in the CEX at time $t$. The net risky exposure of the LP across venues is given by $Q_t + Y_t$. To model aversion to net inventory exposure, we introduce a penalty in the LP's objective on deviations from perfect replication. This penalty is scaled by a parameter governing the  degree of risk aversion of the LP. Higher risk aversion induces more aggressive replication, while lower aversion tolerates exposure to price risk.

A perfect replication strategy would satisfy $Q_t = -Y_t$ for all $t$, which eliminates net inventory exposure in \eqref{eq:PL}. However, trading in the CEX incurs convex execution costs that reflect finite market liquidity. The LP therefore faces a fundamental trade-off: risk aversion penalises deviations from zero net exposure, while trading costs penalise aggressively achieving this target. The optimal trading strategy balances these opposing forces, and feeds back into stage one through its impact on the profitability and viability of liquidity provision.

Beyond risk preferences, we allow the LP to possess private information about the drift of fundamentals. Private information introduces a speculative component into the optimal trading strategy. Importantly, however, private information also allows the LP to anticipate large future reserve adjustments and the associated costs of adverse selection  and offsetting of non-zero net exposure. Anticipated price movements can therefore lead to thinner DEX markets. This mechanism is analogous to traditional limit order markets, where expected price moves induce market makers to quote wider spreads in order to mitigate adverse selection risk and inventory exposure.

\subsubsection{Stage three} 
In stage three, we solve the problem of price-sensitive  noise LTs  who  trade against the pool. We model the arrival of noise LTs over the interval $[0,T]$ as stochastic. Each arriving LT has an exogenous liquidity need, represented by a utility for buying or selling the asset that is unrelated to fundamentals.\footnote{As in classical microstructure models, these traders may represent institutional asset managers rebalancing portfolios or uninformed demand disconnected from price signals.}  Upon arrival, an LT observes the prevailing trading costs, determined by liquidity depth $\kappa$, the current fundamental price $F_t$ (see \eqref{eq:costs}), and the proportional fee $\pi$. The LT solves a static optimisation problem to determine the optimal trading quantity.

The cumulative trading volume over the LP's trading horizon generates fee revenue. In stage one, the LP forms rational expectations over the behavior of noise LTs and anticipates how liquidity depth affects trading costs, volumes, and ultimately fee income.

We solve the model by backward induction. Section~\ref{sec:LTs} solves stage three, where LTs take the liquidity depth~$\kappa$ as given and determine their optimal trading volumes. Section~\ref{sec:hedging} solves stage two, where the LP takes the liquidity depth~$\kappa$ as given and determines the optimal  trading strategy in the CEX. Finally, Section~\ref{sec:liqprovuniswap} solves stage one, where the LP anticipates the effects of trading in the CEX and the activity of both arbitrageurs and noise LTs, to determine the level of DEX reserves.

\section{Stage three: trading volumes and fee revenue}\label{sec:LTs}

\subsection{Assumptions}

The timing of stage three corresponds to the LP's investment horizon $[0,T]$. Throughout this window, arbitrageurs align prices, and noise LTs with price-sensitive demand trade against the pool.  

Consider a noise LT arriving at time $t$ and who demands to buy a quantity $\delta>0$ of asset $Y$. The execution cost consists of (i) the price impact implied by liquidity depth $\kappa$ through \eqref{eq:askprice}, and (ii) the proportional fee $\pi \in (0,1)$. The effective execution price per unit of the risky asset is therefore
\begin{equation}\label{eq:askprice with fee}
    \frac{\varphi(Y_t - \delta, \kappa) - \varphi(Y_t, \kappa) + \pi\,\delta\,F_t}{\delta}.
\end{equation}

In our model, individual noise LTs' liquidity needs are assumed to be small relative to the size of reserves. Thus, noise LTs use a second-order approximation of the execution price. Specifically,
\begin{equation}\label{eq:askpricefee}
\frac{\varphi(Y_{t}-\delta,\kappa)-\varphi(Y_{t},\kappa)+\pi\delta F_{t}}{\delta}
\approx
\frac{-\delta\,\partial_{1}\varphi(Y_{t},\kappa)
+\tfrac{1}{2}\delta^{2}\partial_{11}\varphi(Y_{t},\kappa)}{\delta}+\pi F_{t} = F_{t} + \pi\,F_{t} + \tfrac{1}{2}\delta\,\partial_{11}\varphi(Y_{t},\kappa)\,.
\end{equation}

As shown in \cite{cartea2025decentralised,drissi2023models}, this approximation is accurate in practice. Importantly, it preserves the key economic mechanism: execution prices worsen as liquidity depth $\kappa$ decreases. 
{Specifically, by Assumption~\ref{assume:1:iii},} the convexity term $\partial_{11}\varphi$ is decreasing in $\kappa$, so lower liquidity increases trading costs. Similarly, if the noise LT wishes to sell the quantity $\delta > 0$ of asset~$Y$,  the effective execution price is 
\begin{equation}\label{eq:bidpricefee}
    \frac{\varphi(Y_t, \kappa) - \varphi(Y_t + \delta, \kappa) - \pi\,\delta\,F_t}{\delta} \approx F_{t}-\pi\,F_t-\dfrac{1}{2}\delta\,\partial_{11}\varphi\left(Y_{t},\kappa\right)\,.
\end{equation}

\subsection{Liquidity needs and trading volumes}

The trading volumes of noise LTs solve an optimisation problem. We model liquidity needs as random over the interval $[0,T]$. When an LT arrives at time $t \in [0,T]$, the LT has a utility $F_t(1+V)$ for holding one unit of the asset, where $V$ is the realization of a random variable that is independent of all other processes and it is symmetrically distributed around zero. 
The symmetry of $V$ implies that the expected number of buyers equals the expected number of sellers.

We assume that the distribution of $|V|$ is supported on the compact interval $[\pi,1]$, and we denote
$v = \mathbb{E}[|V|].$ An LT with private utility $V \ge \pi$ (resp. $V \le -\pi$) demands to buy (resp. sell) the asset. The lower bound $\pi$ ensures that the private utility exceeds the proportional trading fee $\pi$, which guarantees nonnegative optimal trading volumes. 

Noise LTs take the liquidity depth $\kappa$ in the DEX, determined by the LP in stage~one, as given. An LT arriving  at time $t$  determines the optimal trading volume $\delta_t^\star$ by balancing  private utility for the asset against execution costs \eqref{eq:askpricefee}--\eqref{eq:bidpricefee}.  The LT's objective, when buying or selling $\delta>0$, is
\begin{align}\label{eqn:criterionLT}
\delta\,\big(|V|-\pi\big)\,F_{t}
-\frac{1}{2}\,\delta^{2}\,\partial_{11}\varphi\left(Y_{t},\kappa\right),
\end{align}
which is maximised with
\begin{equation}\label{eq:LTtradingvolume}
\delta^{\star}_t
= F_t\,\frac{|V| - \pi}{\partial_{11}\varphi\left(Y_{t},\kappa\right)}\,.
\end{equation}
Using the identity \eqref{eq:arbCEXDEX} enforced by arbitrageurs, the trading volume $\delta^\star_t$ in \eqref{eq:LTtradingvolume} can be written  as
\begin{equation}\label{eq:deltastar}
\delta_t^{\star}
= \delta^{\star}(F_t,\kappa)
= \frac{|V|-\pi}{\partial_{11}\varphi\left(h(F_t,\kappa)\,\kappa\right)}\,F_{t}\,.
\end{equation}

\subsection{Fee revenue}
We assume that, over $[0,T]$, arrivals of noise LTs follow a Poisson process $(N_t)_{t \in [0,T]}$ with constant intensity $\lambda$. Each arrival generates fee revenue proportional to the traded volume. The cumulative fee revenue is therefore stochastic. The LP's expected fee revenue, anticipated in stage one, is
\begin{equation}\label{eq:Pi}
\mathbb{E}\left[\int_{0}^{T}\pi\,\delta_{t}^{\star}\,F_{t}\,dN_{t}\right]
= \lambda\,\pi\,(v-\pi)\,
\mathbb{E}\left[\int_{0}^{T}\frac{F_{t}^{2}}{\partial_{11}\varphi\left(h(F_{t},\kappa),\kappa\right)}\,dt\right]=
\mathbb{E}\left[\int_{0}^{T}\Pi(F_{t},\kappa)\,dt\right].
\end{equation}
where we refer to $\Pi_t := \Pi(F_{t},\kappa)$ as the fee revenue rate, expressed in units of the reference asset $X$.

The economic mechanism is immediate: higher depth $\kappa$ reduces convexity $\partial_{11}\varphi$, increases optimal trade sizes $\delta_t^\star$, and therefore raises fee revenue \eqref{eq:Pi}. Moreover,  larger mean private utility $v$ increases trade size. In stage one, the LP anticipates the positive effect of liquidity on fee income. However, deeper liquidity also amplifies inventory exposure and losses to arbitrageurs. The equilibrium liquidity of stage one  balances these opposing effects.

\section{Stage two: risk offsetting}\label{sec:hedging}

In stage two, the LP takes the liquidity depth $\kappa$ chosen in stage one as given. As shown in \eqref{eq:PL}, the DEX position is exposed to inventory risk and losses from arbitrage.  Thus, the LP  manages risk through trading in the CEX.  This section derives the LP's optimal dynamic strategy. The key departure from existing models is that replication is costly, so risk cannot be eliminated without incurring trading costs. The optimal strategy therefore reflects a trade-off between reducing exposure and limiting execution costs.

\subsection{Assumptions} 

Our model is formulated in partial equilibrium: the DEX operates as a secondary market alongside the CEX where price discovery takes place. Accordingly, the LP's trading activity in either venue does not affect liquidity conditions in the CEX.

\paragraph{DEX holdings.} The LP deposits reserves $(X_0, Y_0)$ at time $0$ into the DEX and withdraws reserves $(X_T, Y_T)$ at the terminal time $T>0$. As discussed earlier, the LP remains passive in the DEX over the interval $[0,T]$, so liquidity depth $\kappa$ is fixed throughout.  During the trading window $[0,T]$, the LP holds reserves $X_t$ in the reference asset and $Y_t$ in the risky asset, and arbitrageurs align the pool's marginal price $-\partial_{1}\varphi(Y_t,\kappa)$ with the fundamental value $F_t$, that is, they enforce the equality~\eqref{eq:arbCEXDEX}. 

The LP may also possess private information about the evolution of fundamentals. We assume that the fundamental price $F_t$ follows the stochastic differential equation
\begin{equation}\label{eq: dyn F}
\diff F_t = A_t \,F_t \, \diff t + \sigma \,F_t \, \diff W_t,
\end{equation}
where $F_0 > 0$ is known, $W$ is a Brownian motion, $\sigma > 0$ is a volatility parameter.\footnote{Formally, we consider processes $A = (A_t)_{t \in [0, T]}$ that are progressively measurable and that satisfy $\E\left[\int_0^T|A_t|^p\,\dt\right]<\infty$ for some $p>2$.}
The process $A_t$ represents the LP's stochastic private signal about the drift of fundamentals.\footnote{The fundamental price $F$ satisfies the SDE \eqref{eq: dyn F}, whose solution is
$F_t=F_0\,\exp\left\{\int_0^t(A_s-\tfrac{\sigma^2}{2})\,\diff s+\sigma\,W_t\right\}$,
so the identity \eqref{eq:arbCEXDEX} is well defined.}

\begin{remark}
Our results are agnostic to the information structure of the LP's signal $A_t$. The signal may be fully observable, partially observable, or latent. Examples include signals inferred from filtered order flow or obtained from external price oracles. In the numerical illustrations, we focus on either zero or constant non-zero signals in order to isolate and highlight the implications of private information on equilibrium outcomes.
\end{remark}

Define $G_t := \partial_1 h(F_t,\kappa)\,A_t
+ \dfrac{\sigma^2}{2}\,\partial_{11}h(F_t,\kappa)\,F_t .$ The reserves $X_t$ evolve according to \eqref{eq:levelfunc}, and the reserves $Y_t$ evolve as
\begin{equation}\label{eqn:dyn Y}
\diff Y_t
= G_t\,F_t\,\dt
+ \sigma\,\partial_1 h(F_t,\kappa)\,F_t\,\diff W_t .
\end{equation}


\begin{remark}
In practice, aligning the marginal price of the DEX with its fundamental value involves costs in both markets in which arbitrageurs operate; see \cite{milionis2024automated} for a detailed exposition of these frictions. Moreover, LPs may earn some revenue from arbitrage activity. In our model, we abstract from these institutional details to isolate the central economic mechanism: passive liquidity provision in the DEX exposes the LP to predictable losses when fundamentals move and arbitrageurs rebalance the reserves. Regardless of the specific costs of arbitrage, arbitrageurs only trade when profitable, implying that these trades are, in expectation, loss-making for the LP. We refer  to \cite{bichuch2025price} for a detailed discussion on the fees generated by arbitrageurs.
\end{remark}

We denote the LP's wealth  in the DEX, in units of the reference asset $X$, by $L_t$, and define it as
\begin{equation}\label{eq:wealth DEX}
L_t
:= \int_{0}^{t} \Pi(F_u,\kappa)\,du + X_t + Y_t\,F_t.
\end{equation}
The first term on the right-hand side of~\eqref{eq:wealth DEX} represents the  cumulative fee revenue generated by noise LTs, as determined by their  stage-three best response~\eqref{eq:Pi} to the liquidity level $\kappa$.  The  second and third terms correspond to the marked-to-market value of the LP's liquidity position, evaluated at the fundamental price $F_t$.

\paragraph{CEX holdings.} The LP also trades the risky asset   on the CEX. The LP begins with an initial inventory $Q_0 = -Y_0$ in the CEX and initial cash holdings $\tilde X_0 = 0$ in the reference asset $X$.

\begin{remark}
For simplicity, we assume that the LP begins with a CEX position $Q_0 = -Y_0 = -h(F_0,\kappa)$, which offsets the initial DEX inventory. This normalisation simplifies comparative statics of performance and risk across different model primitives by eliminating initial net exposure. Our results  extend to alternative initial inventory or cash  without altering the main economic mechanisms. In particular, if the LP starts with a CEX position that deviates  from $-Y_0$, then replicating the DEX position requires building a position close to $-Y_0$ in the CEX first, which incurs additional trading costs. This strengthens the incentive of LPs to reduce liquidity supply in anticipation of higher costs, and reinforces the results described below.
\end{remark}

In contrast to the discrete liquidity adjustments in the DEX at times $0$ and $T$, trading in the CEX occurs continuously. Let $\nu_t$ denote the LP's trading rate in the risky asset at time $t$.\footnote{Formally, we consider trading strategies $\nu$ from the admissible set $\mathcal A_2$ of $\mathbb F$-progressively measurable processes satisfying $\mathbb E\left[\int_0^T |\nu_t|^2, dt\right] < \infty$.} The resulting inventory process $Q_t^\nu$ evolves according to\footnote{The superscript $\nu$ indicates that the process $Q_t^\nu$ depends on the LP's trading activity.}
\begin{equation}
Q_t^\nu = Q_0 + \int_0^t \nu_s \,\ds .
\label{eqn: dyn Q}
\end{equation}

In contrast to the extant literature, which assumes costless replication when characterising LP returns, we assume that the CEX has finite liquidity. As a result, trades are not executed at the fundamental price $F_t$. Instead, trading incurs an execution cost that increases with trading intensity. Specifically, the LP faces a linear price impact in the trading rate $\nu_t$. This cost captures the expense of walking the book in a limit order market or, the deterioration of quoted prices when trading large volumes in over-the-counter markets. In our model, the execution price at time $t$ is given by $F_t + \eta\,\nu_t$, where $\eta > 0$ is a cost parameter that measures CEX liquidity. Accordingly, the LP's cash position $\tilde X_t^\nu$ in the reference asset evolves as
$\diff \tilde X_t^\nu = -\left(F_t+\eta\,\nu_t\right)\,\nu_t\,dt\,.$

\begin{remark}
Our model assumes that the LP's trading activity in the CEX incurs instantaneous execution costs but does not affect the midprice. This corresponds to a setting in which the LP's trades are small relative to overall CEX market activity, so that liquidity replenishes instantaneously and prices are not persistently impacted. In practice, if the DEX were sufficiently large, LP trades could generate persistent price impact. 
Appendix~\ref{sec:transient} extends the model to incorporate transient price impact and analyse how this modifies the risk-offsetting problem.
\end{remark}

\subsection{The performance criterion}

At the terminal time $T$, the LP holds inventory $Q_T^\nu$ in the CEX, valued at the terminal price $F_T$, and cash holdings $\tilde X_T^\nu$ in the reference asset. The LP's terminal wealth in the DEX is $L_T$ (see \eqref{eq:wealth DEX}), which consists of: (i) terminal reserves $Y_T$ in the risky asset, valued at $F_T$, (ii) terminal reserves $X_T$ in the reference asset, and (iii) cumulative fee revenue generated by noise LTs.

To capture aversion to inventory risk, we penalise deviations from a perfect replication strategy satisfying $Q_t^\nu = -Y_t$, which would eliminate net exposure. The LP therefore chooses a trading strategy $\nu$ to maximise expected terminal wealth subject to a running deviation penalty. The performance criterion, when employing an admissible trading strategy $\nu$, is
\begin{align}\label{eqn:criterion}
\mathbb{E}\Bigg[
\underbrace{L_T}_{\text{DEX wealth}}
+ \underbrace{Q_T^\nu\,F_T + \tilde X_T^\nu}_{\text{CEX wealth}}
- \underbrace{\dfrac{\phi}{2}\int_0^T (Q_t^\nu + Y_t)^2\,\diff t}_{\text{deviation penalty}}
\Bigg]\,.
\end{align}

Note that $L_T$ does not depend on the strategy $\nu$, so the problem is equivalently written as
\begin{equation}\label{eqn:criterion2}\tag{P}
\E\bigg[\underbrace{Q_{T}^{\nu}\,F_{T}}_\text{CEX position}-\underbrace{\int_{0}^{T}\left(F_t+\eta\,\nu_{t}\right)\,\nu_{t}\,\dt}_\text{risk offsetting}-\underbrace{\dfrac{\phi}{2}\int_{0}^{T}\left(Q_{t}^{\nu}+Y_{t}\right)^{2}\,\dt\bigg]}_\text{deviation penalty}\,.
\end{equation}

The first term in \eqref{eqn:criterion2} is the terminal value of the LP's holdings in the CEX. It creates an incentive to exploit private information. To manage exposure and respond to trading signals, the LP implements the strategy $\nu$ in the CEX. The second term in \eqref{eqn:criterion2} captures the cumulative trading proceeds and costs incurred when executing this strategy. The third term in \eqref{eqn:criterion2} is the key component of the objective. It imposes a running penalty on deviations from perfect replication. The parameter $\phi>0$ scales this penalty: larger values correspond to greater aversion to holding non-zero net exposure $Q_t^\nu+Y_t$. In the limit as $\phi \to \infty$, the optimal strategy converges to perfect replication of the DEX reserves, and as $\phi \to 0$, the LP is risk-neutral.

\subsection{Risk-offsetting strategy}

To solve the LP's optimisation problem in stage two, we employ convex analysis tools for variational problems.\footnote{See Section $5$ of \cite{ekeland1999convex} for a detailed exposition of the tools used in this work.} First, we impose the following mild technical assumptions, which ensure that the optimisation problem is well posed, and we show that the objective~\eqref{eqn:criterion2}  is a real-valued linear-quadratic functional.\footnote{We work on the  admissible set $\mathcal A_2$ of $\mathbb F$-progressively measurable processes, which is exactly the real Hilbert space
\(L^2\!\left(\Omega\times[0,T],\mathcal P,\diff\mathbb P\otimes\dt\right),\) equipped with the inner product $\langle \nu,\zeta\rangle \coloneqq \mathbb E\!\left[\int_0^T \nu_t\,\zeta_t\,\dt\right]$ and the associated norm $\|\nu\| \coloneqq \langle\nu,\nu\rangle^{1/2}$. $\mathcal P$ denotes the progressive $\sigma$-algebra.}

\begin{assume}\label{assume:A and h}
\begin{enumerate}[label=\textnormal{(\roman*)}, ref=\theassume(\roman*)]
    \item\label{assume:A-h:i}
    The private signal satisfies $\E\left[\exp\left(r\int_0^T |A_s|\,\ds\right)\right] < \infty$ for all $r \in \mathbb{R}$.
    \item\label{assume:A-h:ii}
    For each $\kappa > 0$, there exist real numbers
    $C_\kappa, q_\kappa, p_\kappa$ such that, for all $x > 0$, 
    \begin{equation}
        |h(x,\kappa)|
        + |\partial_1 h(x,\kappa)|
        + |\partial_{11} h(x,\kappa)|
        \le C_\kappa \left(x^{q_\kappa} + x^{p_\kappa}\right)\,.
    \end{equation}
\end{enumerate}
\end{assume}

\begin{remark}
Signals that satisfy Assumption~\ref{assume:A-h:i} include all continuous Gaussian processes. Constant product markets such as Uniswap are an example of a market that satisfies Assumption~\ref{assume:A-h:ii}.     
\end{remark}

\begin{lemma}\label{lem:perfcrit no transient}
The performance criterion \eqref{eqn:criterion2} can be written as $J[\nu]+H\,$, where $H=Q_0\,F_0+\E\!\Big[\int_{0}^{T}\big\{Q_0\,A_t\,F_t-\tfrac{\phi}{2}\,(Y_t+Q_0)^2\big\}\, \dt\Big]$ is a well-defined real number that does not depend on $\nu$, 
and $J$ is the bounded linear-quadratic functional
\begin{equation}\label{eq:functional J}
    J[\nu]=-\tfrac{1}{2}\,\langle\Lambda\nu,\nu\rangle+\langle b,\nu\rangle\,,
\end{equation}
where $\Lambda$ is a symmetric bounded linear operator on $\A_2$ and $b$ is an element of $\A_2$, defined by\footnote{$\mathfrak Q$ is well defined for all admissible $\nu$ as $
\E\big[\int_0^T\left|\int_0^t\nu_s\,\ds\right|^2\,\dt\big]\leq\E\big[\int_0^Tt\int_0^t|\nu_s|^2\,\ds\,\dt\big]\leq \frac{T^2}{2}\E\big[\int_0^T|\nu_t|^2\,\dt\big]<\infty.$ $\mathfrak{Q}^\top$ denotes the \emph{transpose} of $\mathfrak{Q}$, defined to be the unique bounded linear operator on $\mathcal{A}_2$ such that $\langle\mathfrak{Q}^\top\nu,\zeta\rangle=\langle\nu,\mathfrak{Q}\zeta\rangle$ for all $\nu,\zeta\in\mathcal{A}_2$.} 
\begin{equation}
    \Lambda=2\,\eta+\phi\,\mathfrak{Q}^\top\mathfrak{Q}\,, \qquad   b=\mathfrak{Q}^\top(A\,F-\phi\,(Y+Q_0))\,, \qquad \text{and} \qquad (\mathfrak{Q}\nu)_t=\int_0^t\nu_s\,\ds\,.\label{defeq:Lambda no transient}
\end{equation}
\end{lemma}

\begin{proof}
See Appendix \ref{proof:lem:perfcrit no transient}.
\end{proof} 

We solve the LP’s optimisation problem explicitly for general convex trading functions and stochastic price signals. The analysis proceeds in three steps. First, we show that the LP’s objective is differentiable in the sense of the Gâteaux directional derivative with respect to admissible stochastic trading strategies, compute its directional derivative, and establish strict concavity (Proposition~\ref{proposition strict concavity no transient}). Strict concavity implies that an admissible strategy for which the derivative vanishes is the unique maximiser. Second, we characterise the first-order condition by identifying where the derivative vanishes. This yields a representation of the critical points of the objective as a coupled system of forward-backward stochastic differential equations (FBSDE) (Theorem~\ref{theorem optimal speed no transient}). Finally, we solve this FBSDE system explicitly and obtain a closed-form expression for the LP's optimal  trading strategy in the CEX  (Proposition~\ref{prop:notrasient}).

\begin{proposition}\label{proposition strict concavity no transient} 
The objective $J$ defined in Lemma~\ref{lem:perfcrit no transient} is strictly concave. Moreover, $J$ is G\^ateaux differentiable, and its G\^ateaux derivative $\D J[\nu]$ at $\nu\in\A_2$ is an element of $\A_2$, given by
\begin{equation}\label{eq:gateauxder no transient}
    \D J[\nu]_t=-2\,\eta\,\nu_{t}+\E\!\Big[\int_{t}^{T}\left(A_{s}\,F_s-\phi\left(Y_{s}+Q_{s}^{\nu}\right)\right)\,\ds\,\mid\,\F_t\Big]\,.
\end{equation}
\end{proposition}

\begin{theorem} \label{theorem optimal speed no transient}
The Gâteaux derivative \eqref{eq:gateauxder no transient} vanishes at $\nu^\star \in \mathcal A_2$ if and only if $\nu^\star$ solves the FBSDE
\begin{align}\label{eq:FBSDE no transient}
\begin{cases}
2\,\eta\,\diff\nu^\star_{t}  & =\left(-A_t\,F_t+\phi\,\left(Y_t+Q_t\right)\right)\,\dt+\diff M_{t},  \qquad \nu^\star_{T}  =0\,, 
\\
\diff Q_{t}  & =\nu^\star_t\,\dt\,, 
\end{cases}
\end{align}
for some martingale $M$ such that $M_T\in L^2(\Omega)$.
\end{theorem}

\begin{proposition}\label{prop:notrasient} 
The optimal risk offsetting strategy in the CEX is 
\begin{equation}\label{eq:nugeneral2}
\nu_{t}^\star=P(t)\,Q_t+\ell_t=\underbrace{P(t)\,Q_{t}-\psi^2\,\E\!\Big[\int_{t}^{T}\tilde{P}(t,s)\,Y_{s}\,\ds\,\mid\,\F_{t}\Big]}_{\text{replication component}}+\underbrace{\tfrac{1}{2\,\eta}\E\!\Big[\int_{t}^{T}\tilde{P}(t,s)\,A_{s}\,F_{s}\,\ds\,\mid\,\F_{t}\Big]}_{\text{speculative component}}.
\end{equation}

\begin{equation}\label{eq:solP no transient}
\text{where}\qquad P(t)=\psi\tanh\left(\psi(t-T)\right) \, \text{,} \qquad  \tilde P(s,t) = \frac{\cosh\left(\psi(t-T)\right)}{\cosh\left(\psi(s-T)\right)}\,,\qquad \text{and}\qquad \psi=\sqrt{\tfrac{\phi}{2\,\eta}}\,.
\end{equation}
\end{proposition}

\begin{proof}
See Appendix   \ref{proof:proposition strict concavity no transient} for the proof of Proposition \ref{proposition strict concavity no transient}, Appendix \ref{proof:theorem optimal speed no transient} for the proof of Theorem \ref{theorem optimal speed no transient}, and Appendix~\ref{proof:prop:notrasient} for the proof of Proposition \ref{prop:notrasient} .
\end{proof}

\paragraph{Replication component.} The first two terms in \eqref{eq:nugeneral2} constitute the replication component. They drive the strategy towards reducing net exposure between the DEX and the CEX.  To see this, note that the kernel $\tilde P(t,s)$ defined in \eqref{eq:solP no transient} is positive and increasing in $s$ for fixed $t$, while $P(t)$ is negative and increasing in $t$. Moreover, recall that DEX reserves $Y_t$ are strictly positive, whereas the LP's initial CEX inventory is negative. Thus, for fixed DEX reserves path $\{Y_s\}_{s\ge t}$, a larger absolute value of the CEX inventory $Q_t$ increases the magnitude of the term $P(t)\,Q_t$, and induces the LP to trade more aggressively to reduce net exposure. Conversely, for fixed CEX inventory $Q_t$, higher expected DEX reserves $\{Y_s\}_{s\ge t}$ lead the LP to sell more in anticipation of future positive exposure.

Because future DEX reserves depend on the evolution of the fundamental price, driven by the signal process $A$, and on the DEX trading function $\varphi$, the optimal strategy in the CEX is inherently forward-looking. The LP therefore targets a weighted average of expected future reserve levels
\begin{equation}\label{eq:hatY}
\hat Y_t := \E\Big[\int_{t}^{T}\tilde{P}(t,s)\,Y_{s}\,\ds\,\mid\,\F_{t}\Big]= \E\Big[\int_{t}^{T}\tilde{P}(t,s)\,h(F_t, \kappa)\,\ds\,\mid\,\F_{t}\Big]\,,
\end{equation}
where the kernel $\tilde P(t,s)$ assigns greater weight to reserves closer to the terminal time $T$. This weighting reflects the fact that net CEX-DEX exposure closer to maturity is more costly to unwind due to convex costs. We make this mechanism explicit below in the constant product case.

Both $\phi$ and $\eta$ represent sources of disutility for the LP:  trading costs in the CEX discourage active replication, whereas aversion to net exposure encourages it. The replication component in \eqref{eq:nugeneral2} depends only on the ratio of the CEX trading cost $\eta$ to the LP’s risk-aversion parameter $\phi$, and not on their absolute levels separately.\footnote{That is, conditional on a given level of inventory $Q_t$.} 
It is therefore the ratio $\psi=\sqrt{\phi/2\,\eta}$ that governs the aggressiveness of risk management. As we discuss below, when risk aversion $\phi$ is small relative to trading costs $\eta$, or equivalently when trading costs are high relative to risk aversion, replication becomes less aggressive and the LP tolerates greater net exposure. 

When the risk-aversion parameter $\phi$ tends to zero, or when trading in the CEX becomes prohibitively costly ($\eta \to \infty$), the replication component vanishes. In particular, as $\phi \to 0$, the LP becomes risk neutral, and trading in the CEX is purely speculative and independent of the DEX position. In this limit, the optimal trading rate converges to
$
\nu_t\ \underset{\phi\rightarrow0}{\longrightarrow} \tfrac{1}{2\,\eta}\E\left[\left.\int_{t}^{T}A_{s}\,F_{s}\,\ds\,\right|\,\F_{t}\right]\,.
$

Finally, the LP's private signal has an indirect but economically significant effect on the replication component. In particular, the expected weighted future level of reserves $\hat Y_t$ in \eqref{eq:hatY} 
depends on the signal process driving fundamentals. When strong expected price drifts are present, anticipated future reserves may increase substantially in one direction. This mechanically amplifies the replication component, as the LP must trade more aggressively to manage the expected future exposure. In stage one, the LP anticipates this effect: large expected reserve movements raise anticipated risk-offsetting costs, which can lead to lower endogenous liquidity. Private information therefore affects liquidity not only through speculation, but also through anticipated replication costs.

\paragraph{Speculative component.}  The speculative component of the optimal strategy \eqref{eq:nugeneral2} is driven by expected future values of the trading signal, weighted through time by the kernel $\tilde P$. Positive (resp.~negative) expected signals induce the LP to buy (resp.~sell) the risky asset in anticipation. The strength of this speculative motive is dampened by trading costs: as $\eta$ increases, trading becomes more expensive and speculation declines. When anticipated in stage one, this component incentivises greater endogenous liquidity.

\section{Stage one: liquidity supply}\label{sec:liqprovuniswap}

In the previous section, we derived the optimal stage-two trading strategy $\nu_t^{\star}$ in the CEX for an arbitrary liquidity depth $\kappa$, corresponding to initial DEX reserves $Y_0 = h(F_0,\kappa)$ and $X_0 = \varphi(Y_0,\kappa)$.

To determine the optimal liquidity depth $\kappa^{\star}$, the LP anticipates three forces. First, the stage-two risk-offsetting strategy is executed in the CEX at a cost. In particular, perfect replication minimises risk but maximises costs. Second, trading volumes of price-sensitive noise LTs follow the  best response \eqref{eq:LTtradingvolume} to the  liquidity supplied. Third, adverse selection costs,  characterised in \eqref{eq:PL}, increase with the liquidity depth $\kappa$.

In this section, we consider general price signals $A$ satisfying Assumption~\ref{assume:A-h:i} and general DEX trading functions satisfying Assumption~\ref{assume:A-h:ii}. We establish the existence of an equilibrium level of liquidity provision under this broad specification. Section~\ref{sec:numerics} then derives the equilibrium explicitly in the case of constant product markets.

Let $Q_t^\star$ denote the  inventory in the CEX resulting from the optimal stage-two strategy $\nu_t^\star$ in \eqref{eq:nugeneral2}. In stage one, the LP maximises the profit and loss (PnL) from providing liquidity. The initial value of the LP's combined holdings in the CEX and the DEX, expressed in units of the reference asset $X$, is
$
X_0 +Y_0\,F_0 + \tilde X_0 + Q_0\,F_0 = X_0\,, 
$
where the equality follows from the normalisation $Q_0 = -Y_0$ and $\tilde X_0=0$. The optimisation problem of the LP in stage one is therefore
\begin{align}\label{eqn:valuefn no transient}\tag{K}
\sup_{\kappa\in[\underline{\kappa},\overline\kappa]}\mathbb{E}\!
\Bigg[
\underbrace{\int_0^T\Pi(F_t,\kappa)\,\dt+X_T -X_0+ Y_T\, F_T}_\text{PnL in the DEX} 
\ + \ \underbrace{Q_T^{\star}\,F_T
- \int_0^T \left(F_t + \eta\, \nu^\star_t\right)\, \nu^\star_t \, \diff t}_\text{PnL in the CEX}
\Bigg]\,,
\end{align}
where $\Pi$ is defined in \eqref{eq:Pi}, and $\underline{\kappa}$ (resp. $\overline{\kappa}$) denotes the minimum (resp. maximum) admissible liquidity depth implied by the LP's budget constraint.\footnote{The LP may choose $\underline{\kappa}$ arbitrarily small. We do not set it equal to zero in order to avoid the singularity at $\kappa = 0$ in Assumptions~\ref{assume:1} and to simplify the analysis.}

The next result shows that the LP's objective is well defined and establishes mild conditions under which equilibria exist. These conditions are satisfied by popular trading functions including the constant product and the geometric mean market makers.

\begin{proposition}\label{prop:valuefndfd no transient}
The objective in  \eqref{eqn:valuefn no transient} 
is well-defined for all $\kappa>0$. 
Moreover, suppose there exist $\mathfrak{p},\mathfrak{q}\in\mathbb{R}$ and a continuous function $\mathfrak{C}:(0,\infty)\to(0,\infty)$ such that, for all $x,\kappa,\kappa^\prime>0$,
\begin{equation}
    |h(x,\kappa)-h(x,\kappa^\prime)|+|\partial_1h(x,\kappa)-\partial_1h(x,\kappa^\prime)|+|\partial_{11}h(x,\kappa)-\partial_{11}h(x,\kappa^\prime)|\leq \left(x^{\mathfrak{p}}+x^{\mathfrak{q}}\right)\,|\mathfrak{C}(\kappa)-\mathfrak{C}(\kappa^\prime)|\,.
\end{equation}
Then the LP's objective \eqref{eqn:valuefn no transient} is continuous in $\kappa$ and attains its maximum over the compact set $[\underline{\kappa},\overline{\kappa}]$.
\end{proposition}
\begin{proof}
    See Appendix~\ref{proof:prop:valuefndfd no transient}.
\end{proof}

\section{Constant product markets}\label{sec:numerics}

To illustrate the implications of our model, we examine equilibrium outcomes in constant product markets (CPMs), such as Uniswap. This specification allows for explicit characterisation of equilibrium liquidity and for comparative statics. In a CPM, the level function is
$\varphi(Y,\kappa) = \frac{\kappa^{2}}{Y},$ 
and the corresponding fundamental price and reserves satisfy
\begin{equation}\label{eq:uniswap2}
F_t = -\partial_{1}\varphi(Y_t,\kappa)
= \frac{\kappa^{2}}{Y_t^{2}}
\qquad\text{and}\qquad
Y_t = h(F_t,\kappa)
= \frac{\kappa}{\sqrt{F_t}}\,.
\end{equation}
Under these assumptions, and for a given depth $\kappa$, the fee revenue rate \eqref{eq:Pi} simplifies to
\begin{equation}\label{eq:Pi CPM}
\Pi\left(F_{t},\kappa\right)
= \gamma\,\kappa
\sqrt{F_{t}}\,.
\end{equation}

Next, we derive the equilibrium liquidity in the CPM. 

\begin{proposition}\label{prop:liq no hedge}
Assume~\eqref{eq:uniswap2} holds and let $P$ and $\tilde P$ be defined as in \eqref{eq:solP no transient}. Define  the processes 
\begin{align}\label{eq:def Ct Dt}
C_{t}^{\ell} = -\psi^2\,\,
\E\Big[\int_{t}^{T}\tilde{P}(t,s)\,F_{s}^{-1/2}\,\ds\,\mid\,\F_{t}\Big],
\qquad
D_{t}^{\ell} = \tfrac{1}{2\,\eta}\,
\E\Big[\int_{t}^{T}\tilde{P}(t,s)\,A_{s}\,F_{s}\,\ds\,\mid\,\F_{t}\Big], \\
C_{t}^{Q} = - \sqrt{F_{0}}\,\tilde{P}(0,t)+\int_{0}^{t}\tilde{P}(s,t)\,C_{s}^{\ell}\,\ds, \quad 
 D_{t}^{Q}  = \int_{0}^{t}\tilde{P}(s,t)\,D_{s}^{\ell}\,\ds,
\quad 
C_{t}^{\nu}  = P(t)\,C_{t}^{Q} + C_{t}^{\ell},
\end{align}
and $D_{t}^{\nu}  = P(t)\,D_{t}^{Q} + D_{t}^{\ell}$. The equilibrium liquidity is
$\max{\left(\underline{\kappa},\min{\left(\overline{\kappa},\kappa^\star\right)}\right)}\,,$
where
\begin{align}\label{eq:kappahedgeA}
    \kappa^{\star}&=\underbrace{\left(\gamma-\sigma^{2}/4\right)\frac{\int_{0}^{T}\E\!\left[F_{t}^{1/2}\right]\dt}{2\,\eta\,\int_{0}^{T}\E\!\left[\left(C_{t}^{\nu}\right)^{2}\right]\,\dt}}_{\text{Fee revenue / adverse selection / risk}}+\underbrace{\frac{\E\!\left[\int_{0}^{T}\left(A_{t}\,F_{t}^{1/2}+C_{t}^{Q}\,A_{t}\,F_{t}-2\,\eta\,C_{t}^{\nu}\,D_{t}^{\nu}\right)\,\dt\right]}{2\,\eta\,\int_{0}^{T}\E\!\left[\left(C_{t}^{\nu}\right)^{2}\right]\,\dt}}_{\text{speculative component}}\,,
\end{align}
and $\gamma$ is called the profitability parameter and is defined as
\begin{equation}\label{eq:gamma def}
\gamma := \frac{\lambda\,\pi\,(v-\pi)}
{2}\,.
\end{equation}
\end{proposition}
\begin{proof}
See Appendix~\ref{proof:prop:liq no hedge}.
\end{proof}

The first component of the equilibrium liquidity~\eqref{eq:kappahedgeA} is proportional to the net profitability of price-sensitive liquidity demand relative to adverse selection costs, given by $\gamma - \frac{1}{4}\sigma^2$.  The LP anticipates that the ratio $\psi$ of risk aversion $\phi$ to trading costs $\eta$ governs the aggressiveness of replication and, consequently, the expected cost of executing the optimal strategy \eqref{eq:nugeneral2}. The denominator of this component therefore captures the expected CEX trading costs incurred when implementing the strategy. In equilibrium, liquidity supply decreases when anticipated replication costs are high, either because CEX liquidity is scarce (high $\eta$) or because risk aversion is strong (high $\phi$).

The speculative component in~\eqref{eq:kappahedgeA} is driven by private information and operates through two distinct economic channels: a return motive and a cost motive. 

First, trading signals generate expected capital gains from executing the speculative component of the optimal strategy in the CEX, which takes directional positions in anticipation of future price movements. These gains can, in equilibrium, support a larger liquidity supply.

Second,  price changes also increase anticipated net inventory exposure across the CEX and the DEX. Larger expected reserve movements require more aggressive replication in the CEX, which increases expected trading costs. This cost channel reduces equilibrium liquidity.

As shown in the comparative statics below, the net effect of private information on equilibrium liquidity depends on the relative strength of these two forces. For moderate positive expected drifts, speculative gains dominate and equilibrium liquidity increases. However, when expected price movements are  sufficiently large, higher anticipated replication costs  outweigh speculative benefits, leading to a contraction in liquidity supply. 

In traditional electronic markets, market makers who anticipate adverse price movements can widen spreads or skew quotes to avoid accumulating inventory and mitigate adverse selection. By contrast, liquidity provision in DEXs is passive and constrained by the slow and costly infrastructure of blockchains.  As a result, the LP's primary strategic response to anticipated adverse price movements is to reduce the level of liquidity supplied.

Next, Section~\ref{subsec:no info} analyses how CEX market conditions affect the viability of DEX liquidity provision when the LP does not rely on private information. Section~\ref{subsec:info} then examines how private information modifies equilibrium liquidity and performance.

\subsection{Comparative statics: aversion, costs, volatility, and fees}\label{subsec:no info}

Assume the liquidity provider does not exploit private information and that the fundamental price follows \eqref{eq: dyn F} with $A\equiv 0$. In this case the speculative component in the equilibrium liquidity expression \eqref{eq:kappahedgeA} vanishes. The corollary below characterises the equilibrium liquidity supply under these assumptions; it is a special case of Proposition~\ref{prop:liq no hedge}, so we omit the proof.

\begin{corollary}\label{cor:Azero}
Assume~\eqref{eq:uniswap2} holds and $A\equiv0.$ Let $\psi=\sqrt{\tfrac{\phi}{2\,\eta}}$ and $\Sigma=\frac38\sigma^{2}$. If $\gamma>\frac{1}{4}\sigma^2$, then the equilibrium liquidity supply is
\begin{equation}\label{eq:kappahedgeCPM}
\kappa^{\star}=\frac{4\,\sqrt{F_{0}}}{\sigma^{2}}\left(\gamma-\frac{\sigma^{2}}{4}\right)\frac{1-e^{-\frac18\sigma^{2}\,T}}{\eta\,\int_{0}^{T}\mathbb{E}\left[\left(C_{t}^{\nu}\right)^{2}\right]\dt}\,,
\end{equation}
where $C_t^\nu$ is defined in \eqref{eq:def Ct Dt} and $C_t^\ell$ becomes
$C_{t}^{\ell}=-\frac{\psi^{2}\,e^{\Sigma\,(T-t)}}{2\,\sqrt{F_{t}}\,\cosh\!\left(\psi\,(t-T)\right)}\left(\frac{1-e^{\left(\Sigma+\psi\right)(t-T)}}{\Sigma+\psi}+\frac{1-e^{\left(\Sigma-\psi\right)(t-T)}}{\Sigma-\psi}\right)\,.$
\end{corollary}

We now examine how the model primitives shape equilibrium outcomes in CPMs. In particular, we study the effects of CEX trading costs $\eta$, risk aversion $\phi$, fundamental volatility $\sigma$, and the profitability parameter $\gamma$.

\vspace{.3cm}
\noindent
\textbf{Liquidity supply.} Liquidity provision is viable only when the profitability parameter $\gamma$ exceeds the adverse selection threshold $\frac{1}{4}\sigma^2$. When $\gamma > \frac{1}{4}\sigma^2$, liquidity is strictly positive and decreasing in fundamental volatility $\sigma$, consistent with the existing literature; see \cite{milionis2022automated} and \cite{cartea2023predictable}. Beyond this threshold, liquidity depends critically on both the ratio $\psi$ of risk aversion $\phi$ to replication costs $\eta$, and the absolute level of trading costs $\eta$. 

\begin{figure}[!h]
    \centering
    \includegraphics[width=.7\linewidth]{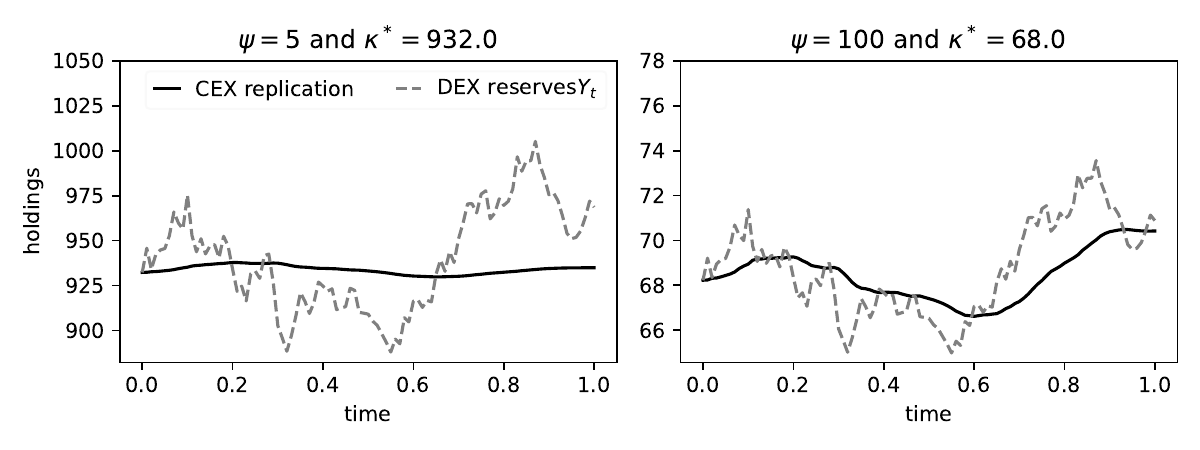}
    \caption{Sample path of  the LP's reserves $Y_t$ held in the DEX and the inventory $Q_t$ held in the CEX. The left panel corresponds to a ratio of risk aversion to trading costs $\psi = 5$, while the right panel corresponds to $\psi = 100$. Other default parameter values are trading costs $\eta=0.01$, profitability $\gamma = 0.05$, fundamental volatility $\sigma = 0.2$, and investment horizon $T = 1$.}
    \label{fig:paths}
\end{figure}

Figure~\ref{fig:paths} illustrates this mechanism in the constant product setting. For fixed trading costs $\eta$, when $\psi$ is low (left panel), the LP adjusts the CEX position only gradually in response to changes in DEX reserves. As a result, expected trading costs remain moderate and equilibrium liquidity is relatively high. When $\psi$ increases (right panel), the LP tracks fluctuations in DEX reserves more closely. Anticipating the associated increase in CEX trading costs, the LP optimally reduces liquidity provision ex ante to reduce price exposure. 

Thus, the viability of DEX liquidity is preserved  through a reduction in DEX liquidity depth, and not by replication itself. Liquidity itself becomes a key risk-management tool which controls  exposure, and the associated expected risk and returns.

Figure~\ref{fig:liquiditysupply} illustrates equilibrium liquidity~\eqref{eq:kappahedgeCPM} as a function of market conditions and risk aversion. The first panel plots liquidity against the ratio $\psi$ of risk aversion to CEX trading costs. As $\psi$ increases, disutility from net exposure becomes more significant relative to trading costs, leading to more aggressive replication of the DEX position (see Figure~\ref{fig:paths}). Anticipating the higher replication costs associated with this behaviour, the LP optimally reduces liquidity depth as $\psi$ increases.

As $\psi \to 0$, the LP is neutral to risk. In this limit, if liquidity provision is profitable in expectation, that is, if $\gamma > \frac{1}{4}\sigma^2$, the LP allocates the entire available budget $\overline{\kappa}$ to liquidity provision. Conversely, if $\gamma \leq \frac{1}{4}\sigma^2$, liquidity collapses and the DEX shuts down.

\begin{figure}[!h]
    \centering
    \includegraphics[width=.8\linewidth]{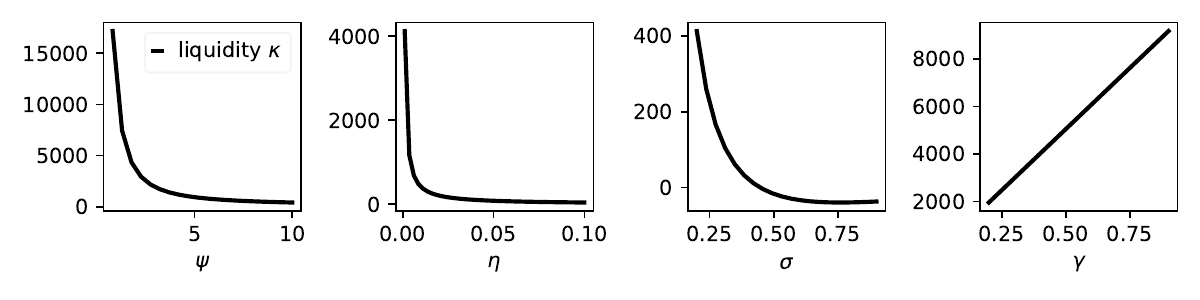}
    \caption{Equilibrium liquidity supply $\kappa^{\star}$ in~\eqref{eq:kappahedgeCPM} plotted as functions of the model primitives. Default parameter values are: fundamental volatility $\sigma = 0.2$, investment horizon $T = 1$, private signal $A = 0$, CEX trading cost $\eta = 10^{-2}$, ratio $\psi = \sqrt{\phi/2\eta} = 10$, and profitability $\gamma = 0.05$.
    }
    \label{fig:liquiditysupply}
\end{figure}

The second panel of Figure~\ref{fig:liquiditysupply} shows that, for a fixed ratio $\psi$, higher trading costs $\eta$ reduce equilibrium DEX liquidity. The economic mechanism is as follows. Even if the intensity of replication is governed by $\psi$,\footnote{As shown in  Figure~\ref{fig:paths} which shows that the intensity of replication in the CEX is governed solely by the ratio $\psi$.} the absolute cost of replication increases with $\eta$. As a result, liquidity provision becomes less attractive and equilibrium $\kappa$ declines. In the limit as $\eta \to 0$, risk management becomes effectively costless and equilibrium liquidity reaches the budget constraint $\overline\kappa$.

The third panel of Figure~\ref{fig:liquiditysupply} shows that higher fundamental volatility $\sigma$ reduces equilibrium liquidity. The final panel shows that greater profitability of price-sensitive liquidity demand increases liquidity. In the model, profitability increases with the fee rate $\pi$, the arrival intensity of noise LTs $\lambda$, and the average absolute liquidity need $v$.

\vspace{.3cm}
\noindent
\textbf{Viability of DEXs.} We next examine the equilibrium profitability of liquidity provision in a CPM as a function of market conditions and risk aversion. Specifically, we compute the LP's profit and loss when risk is optimally offset in the CEX, as given by~\eqref{eqn:valuefn no transient}. We also report returns both in absolute terms, expressed in units of the reference asset $X$, and in relative terms. For the latter, we normalise the PnL by the initial combined DEX-CEX position $X_{0}+C_0+(Y_0+Q_0)F_0 = X_0 = \kappa^{\star}\sqrt{F_{0}},$ and express profitability as a percentage of this initial capital.

The top panels of Figure~\ref{fig:perfs} show that absolute profits from liquidity provision mirrors equilibrium liquidity in Figure~\ref{fig:liquiditysupply}. This is mechanical: absolute profits scale with the level of liquidity supplied. Liquidity declines when risk management is more aggressive (high $\psi$), when CEX trading is costly (high $\eta$), or when adverse selection risk is elevated (high $\sigma$). In each case, lower liquidity reduces fee revenue and therefore absolute returns. Conversely, stronger price-sensitive liquidity demand expands the liquidity base, increasing both absolute profits and the dispersion of outcomes.

The bottom panels of Figure~\ref{fig:perfs} report returns relative to the LP's initial investment, $X_0 = \kappa^\star \sqrt{F_0}$, which scales with the liquidity deposit. Normalising by initial capital removes scale effects and isolates the risk-return profile per unit of deployed capital.

\begin{figure}[!h]
    \centering
    \includegraphics[width=.8\linewidth]{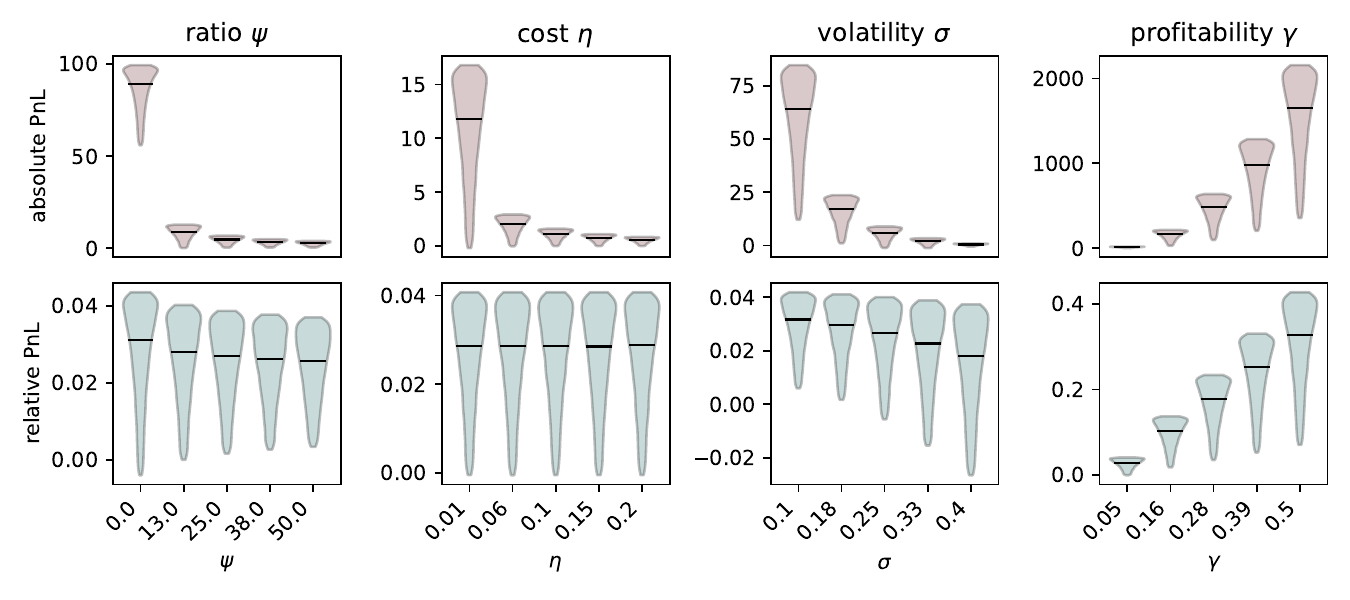}
    \caption{Violin plots of the empirical distributions of the equilibrium profits of liquidity provision expressed in units of the reference asset $X$ (top panels) and in percentage term (bottom panels).  The distribution is obtained from $100$,$000$ market simulations, with the time interval discretised into $100$ steps.  Default parameter values are $\sigma = 0.2$, $T = 1$, $A = 0$, $\eta =  10^{-2}$, $\psi = 10$, and $\gamma = 0.05$.}
\label{fig:perfs}
\end{figure}

A key structural implication emerges once liquidity is endogenised. The viability of liquidity provision is not determined solely by the trade-off between fee revenue and adverse selection; rather, it also depends critically on replication costs and risk aversion. Specifically, the ratio $\psi$  influences the risk-return trade-off per unit of deployed capital. A lower ratio increases expected returns towards the risk-neutral benchmark when $\gamma > \frac{1}{4}\sigma^2$, but at the cost of greater return variance. Conversely, higher values of $\psi$ reduce equilibrium  variance of returns but at the cost of lower profitability.

The second bottom panel of Figure~\ref{fig:perfs} shows that although $\eta$ determines the cost of trading in the CEX, equilibrium liquidity adjusts such that the distribution of returns per unit of capital remains invariant as $\eta$ varies with all else held fixed.  Trading costs therefore affect the scale of liquidity provision, but not the per-unit risk-return trade-off in equilibrium.

As emphasised in the existing literature, the viability of DEX liquidity also critically depends on fundamental volatility and the profitability of liquidity demand. The last two bottom panels of Figure~\ref{fig:perfs} show that higher volatility substantially increases the dispersion of relative returns and lowers their mean, reflecting greater adverse selection risk per unit of capital. In contrast, greater profitability of price-sensitive liquidity demand increases both expected relative returns and risk: as $\gamma$ rises, the LP optimally supplies more liquidity (see Figure~\ref{fig:liquiditysupply}) and increases exposure to price fluctuations; however, higher fee revenue compensates for this additional risk on average.

\subsection{Comparative statics: private information}\label{subsec:info}

Assume that the LP observes a constant private signal $A$ driving the drift of the fundamental price of the risky asset $Y$. As discussed earlier, private information affects equilibrium outcomes through two distinct economic channels. First, it creates a speculative motive. Second, anticipated price changes imply larger inventory exposure  requiring more intensive risk management.

\begin{figure}[!h]
    \centering
    \includegraphics[width=.65\linewidth]{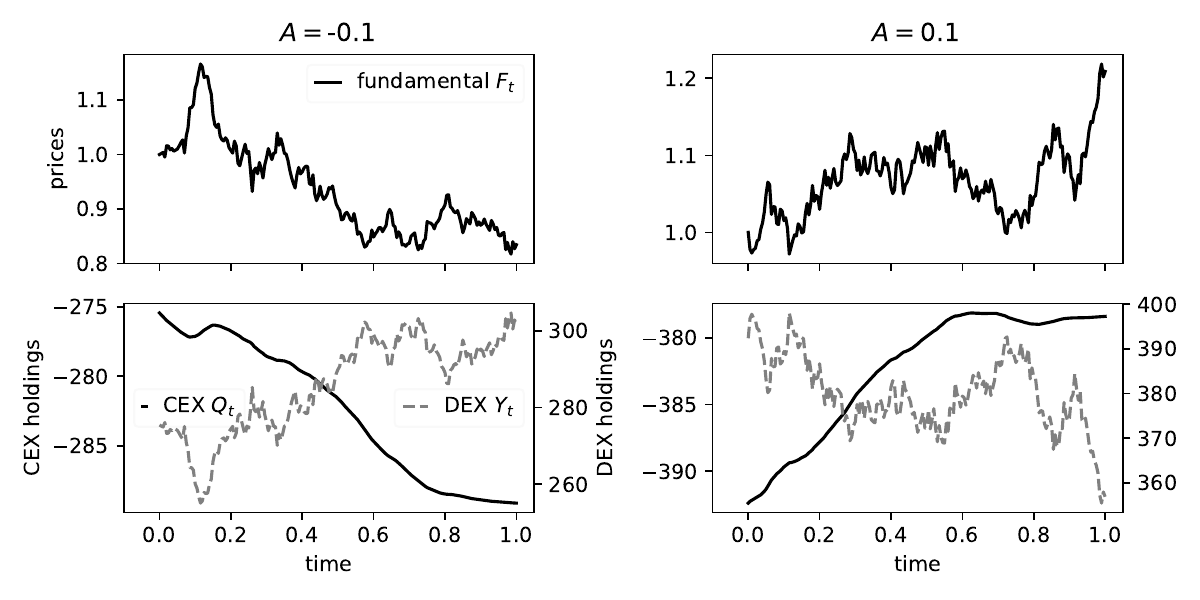}
    \caption{Sample paths of the fundamental price $F_t$ (top panels), and the LP's reserves $Y_t$ held in the DEX and the inventory $Q_t$ held in the CEX (bottom panels). The left panels correspond to $A=-0.2$ and the right panels to $A=0.2.$ Other default parameter values are ratio $\psi=0.1,$ profitability $\gamma = 0.05$, fundamental volatility $\sigma = 0.2$, and investment horizon $T = 0.3$.}
    \label{fig:pathsA}
\end{figure}

The first channel is illustrated in Figure~\ref{fig:pathsA}, where the LP has low relative risk aversion ($\psi=0.1$). In this case, the  position $Q_t$ in the CEX is tilted towards exploiting the expected price movement rather than closely tracking the DEX reserves, which move in the opposite direction to price changes. Speculative positioning therefore dominates replication.

The positive effect of speculative incentives on liquidity provision is, however, generally weaker than the adverse effects arising from increased adverse selection costs and higher replication costs.

Figure~\ref{fig:distrA} illustrates how private information affects equilibrium liquidity and its viability. Private information does not systematically lead to higher performance or deeper markets. Rather than serving purely as a tool to enhance liquidity (see \cite{klein2023price}), private information in the context of passive DEX liquidity provision also operates as a defensive mechanism: it alerts the LP to anticipated adverse selection and future replication costs.

\begin{figure}[!h]
    \centering
    \includegraphics[width=0.65\linewidth]{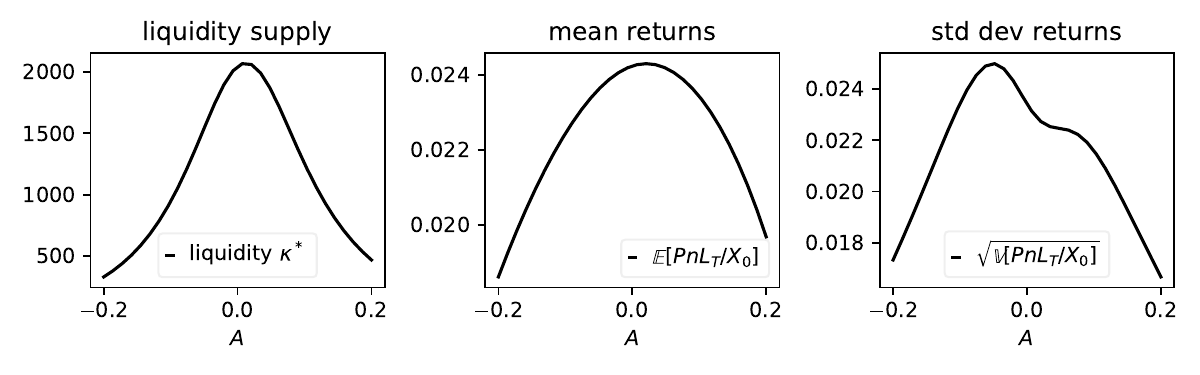}
    \caption{Left panel: equilibrium DEX liquidity $\kappa^\star$. Middle panel: average equilibrium profits of liquidity provision expressed in percentage terms. Right panel: standard deviation of equilibrium profits of liquidity provision expressed in percentage terms. The distribution is obtained from $100$,$000$ market simulations, with the time interval discretised into $100$ steps.     
    Default parameter values are $\sigma = 0.2$, $T = 1$, $\eta =  10^{-2}$, $\psi = 10$, and $\gamma = 0.05$.}
\label{fig:distrA}
\end{figure}

For moderate positive values of the drift signal, the LP anticipates that expected capital gains from speculative trading in the CEX will complement fee revenue, and liquidity supply increases. However, when the signal becomes sufficiently large in absolute value, anticipated reserve adjustments increase substantially. Replicating the DEX position then requires more intensive trading in the CEX, raising expected trading costs. Anticipating these higher costs, the LP optimally reduces liquidity, leading to thinner markets.

Finally, an asymmetry between moderate positive and negative signals is visible in Figure~\ref{fig:distrA}. This asymmetry arises from the nonlinearity of the reserve mapping $Y = h(F_t,\kappa)$. In the constant product case, $\partial_{FF}Y_t > 0$, so $Y(F)$ is convex. As a result, downward movements in $F_t$ increase reserves more rapidly than comparable upward movements reduce them. Consequently, negative price signals generate larger reserve expansions and require more intensive replication in the CEX. 

\section{Conclusions}

This paper builds a structural model of liquidity provision in DEXs in which arbitrageurs align DEX prices with fundamentals, which generates adverse selection costs for LPs, while noise price-sensitive liquidity demand generates fee revenue. We show that, once trading volumes and liquidity supply are endogenised, the losses and risks borne by LPs are not summarised by any single measure. Instead, they depend on (i) market conditions, including CEX liquidity depth, fundamental volatility, and noise trading activity, and on (ii) the LP's risk aversion, which ultimately shapes the distribution of returns from DEX liquidity provision.

\appendix

\allowdisplaybreaks

\section{Proofs}\phantomsection\label{apx:proofs}

\paragraph{Preliminary result.} To prove the results below, we first a prove a useful Lemma.

\begin{lemma}\label{lem:FhA_2}
    The following inequalities hold:
\begin{itemize}
    \item $\E\Big[\sup_{t\leq T}F_t^q\Big]<\infty\,,\quad\E\Big[\int_0^TF_t^q\,\dt\Big]<\infty\,,\quad\forall q\in\mathbb R.$
    \item $\E\Big[\sup_{t\leq T}Y_t^q\Big]<\infty\,,\quad\E\Big[\int_0^tY_t^q\,\dt\Big]<\infty\,,\quad \forall q\in[1,\infty).$
    \item $\E\Big[\int_0^T|G_t|^q\,\dt\Big]<\infty\,,\quad \forall q\in[1,p)\,.$
\end{itemize}
Then,
    $F^q\in\A_2$,\  $\forall q\in\mathbb{R}$.        Moreover, $\forall \kappa>0$ and $q\geq 1$, $h(F,\kappa)^q$, $(\partial_1h(F,\kappa))^q$, $(\partial_{11}h(F,\kappa))^q \in \A_2$.
\end{lemma}
\begin{proof}
    \noindent
    For each $q\in\mathbb{R}$, consider the exponential martingale $M(q)_t:=e^{q\,\sigma\,W_t-\tfrac{1}{2}\,q^2\,\sigma^2\,t},$ and write $F_t^q=F_0^q\,e^{\tfrac{1}{2}\,(q^2-q)\,\sigma^2\,t}\;e^{q\int_0^tA_s\,\ds}\,M(q)_t\,.$ By Cauchy-Schwarz inequality, Doob's inequality, and Assumption~\ref{assume:A-h:i}, we obtain
\begin{align}
    &\E\Big[\sup_{t\leq T}F_t^q\Big]\leq F_0^q\,e^{\tfrac{1}{2}\,|q^2-q|\,\sigma^2\,T}\,\E\big[e^{|q|\int_0^T|A_s|\,\ds}\;\sup_{t\leq T}M(q)_t\big]
    \\
    &\leq F_0^q\,e^{\tfrac{1}{2}\,|q^2-q|\,\sigma^2\,T}\,\E\big[e^{2\,|q|\int_0^T|A_s|\,\ds}\big]^{\tfrac12}\,\E\big[\sup_{t\leq T}(M(q)_t)^{2}\big]^{\tfrac12}
    \\
    &\leq 2\,F_0^q\,e^{\tfrac{1}{2}\,|q^2-q|\,\sigma^2\,T}\,\E\big[e^{2\,|q|\int_0^T|A_s|\,\ds}\big]^{\tfrac12}\,\E\big[(M(q)_T)^2\big]^{\tfrac12}
    \leq 2\,F_0^q\,e^{\tfrac{1}{2}\,(|q^2-q|+q^2)\,\sigma^2\,T}\,\E\big[e^{2\,|q|\int_0^T|A_s|\,\ds}\big]^{\tfrac12}<\infty
\end{align}
and $\E\Big[\int_0^TF_t^q\,\dt\Big]\leq T\,\E\Big[\sup_{t\leq T}F_t^q\Big]<\infty\,.$ By Assumption~\ref{assume:A-h:ii},  and using $Y_t=h(F_t,\kappa)$, we obtain,
\begin{align}
    \forall q\in[1,\infty),\quad \E\Big[\sup_{t\leq T}Y_t^q\Big]=\E\Big[\sup_{t\leq T}h(F_t,\kappa)^q\Big]&\leq C_\kappa^q\,\E\Big[\sup_{t\leq T}\left(F_t^{q_\kappa}+F_t^{p_\kappa}\right)^q\Big]
    \\
    &\leq C_\kappa^q\,2^{q-1}\,\left(\E\Big[\sup_{t\leq T}F_t^{q_\kappa\,q}\Big]+\E\Big[\sup_{t\leq T}F_t^{p_\kappa\,q}\Big]\right)
    <\infty\,.
\end{align}
Moreover, we also obtain $\E\Big[\int_0^tY_t^q\,\dt\Big]\leq T\,\E\Big[\sup_{t\leq T}Y_t^q\Big]<\infty\,,$ and $\forall q\in[1,p)$, we have
\begin{align}
    \E&\Big[\int_0^T|G_t|^q\,\dt\Big]
    =\E\Big[\int_0^T\left|\partial_1 h(F_t,\kappa)\,A_t+\tfrac{\sigma^2}{2}\,\partial_{11}h(F_t,\kappa)\,F_t\right|^q\,\dt\Big]
    \\
    &\lesssim\E\Big[\int_0^T\left(F_t^{q_\kappa\,q}+F_t^{p_\kappa\,q}\right)\,|A_t|^q\,\dt+\int_0^T\left(F_t^{q_\kappa\,q+q}+F_t^{p_\kappa\,q+q}\right)\,\dt\Big]
    \\
    &\lesssim\E\Big[\int_0^T\left(F_t^{\tfrac{p\,q_\kappa\,q}{p-q}}+F_t^{\tfrac{p\,p_\kappa\,q}{p-q}}\right)\,\dt\Big]^{\tfrac{p-q}{p}}\,\E\Big[\int_0^T|A_t|^p\,\dt\Big]^{\tfrac{q}{p}}
     +\E\Big[\int_0^T\left(F_t^{q_\kappa\,q+q}+F_t^{p_\kappa\,q+q}\right)\,\dt\Big]
    <\infty\,.
\end{align}

\noindent
The inequalities above and Assumption~\ref{assume:2} immediately imply the result of the lemma. \qed
\end{proof}

\subsection{Proof of Lemma \ref{lem:perfcrit no transient} }\phantomsection\label{proof:lem:perfcrit no transient}

\noindent
Lemma \ref{lem:FhA_2} shows that the processes involved in the objective are well defined. The rest of the proof proceeds in four steps. First, we show that the performance criterion \eqref{eqn:criterion2} is well-defined and continuous. Next, we show that the functional $J$ in \eqref{eq:functional J} is also well-defined and continuous. Next, we show that the performance criterion  \eqref{eqn:criterion2} and $J$ in \eqref{eq:functional J} agree up to a constant on bounded processes. Finally, we conclude.

\noindent
\textbf{Step 1.} First, we show that the performance criterion \eqref{eqn:criterion2} is well-defined and continuous. Note that $Q^\nu=Q_0+\mathfrak{Q}\nu$. Take $\nu\in\A_2$. Then
\begin{equation}
    \E\!\Big[\left|Q^\nu_T\right|^2\Big]=\E\!\Big[\left|Q_0+\int_0^T\nu_t\,\dt\right|^2\Big]\leq 2\,\left(|Q_0|^2+T\,\E\!\Big[\int_0^T|\nu_t|^2\,\dt\Big]\right)<\infty\label{ineq:Q_t}
\end{equation}
This together with Lemma~\ref{lem:FhA_2} and Cauchy-Schwarz inequality implies $\E\!\Big[Q^\nu_T\,F_T\Big]$ is well-defined. Because \eqref{eqn:criterion2} can be written as
\begin{equation}
    \E\!\Big[Q^\nu_T\,F_T\Big]-\eta\,\Vert\nu\Vert^2-\frac{\phi}{2}\,\Vert\mathfrak{Q}\nu\Vert^2-\langle F,\nu\rangle-\phi\,\langle Y+Q_0,\mathfrak{Q}\nu\rangle-\frac{\phi}{2}\,\Vert Y+Q_0\Vert^2\,,\label{P lemma 1}
\end{equation}
where $Y\in\A_2$ by Lemma~\ref{lem:FhA_2} and $\mathfrak{Q}$ is bounded linear operator on $\A_2$, it is well-defined.

\vspace{.3cm}\noindent
Because $\nu\mapsto-\eta\,\Vert\nu\Vert^2-\frac{\phi}{2}\,\Vert\mathfrak{Q}\nu\Vert^2-\langle F,\nu\rangle-\phi\,\langle Y+Q_0,\mathfrak{Q}\nu\rangle$ is a linear-quadratic form on $\A_2$, it is continuous, it remains to show $\E\![Q^\nu_T\,F_T]$ is continuos in $\nu$. Indeed, if $\nu^{\prime}\in\A_2$, then
\begin{align}
    \left|\E\Big[F_T\,\left(Q^{\nu^{\prime}}_T-Q^\nu_T\right)\Big]\right|&\leq\E\Big[|F_T|^2\Big]^{\tfrac12}\,\E\Big[\left|Q^{\nu^{\prime}}_T-Q^\nu_T\right|^2\Big]^{\tfrac12}
    \leq \sqrt{2\,T}\,\E\Big[|F_T|\Big]^{\tfrac12}\,\E\Big[\int_0^T\left|\nu^{\prime}_t-\nu_t\right|^2\,\dt\Big]^{\tfrac12}\,.
\end{align}
This implies $\E\!\Big[Q^\nu_T\,F_T\Big]$ is continuous in $\nu$, as desired.

\vspace{.3cm}\noindent
\textbf{Step 2.} Next, we show that $J$ is well-defined and continuous. Since $\mathfrak{Q}$ is well-defined, $\Lambda$ is also well-defined. Because we know $Y\in\A_2$ by Lemma~\ref{lem:FhA_2}, it remains to show the product $A\,F$ is in $\A_2$, which implies $b\in\A_2$. Indeed, by Lemma~\ref{lem:FhA_2}, $\E\!\Big[\int_0^T\left|A_{t}\,F_t\right|^2\,\dt\Big]\leq\E\!\Big[\int_0^T\left|A_t\right|^p\,\dt\Big]^{\tfrac{2}{p}}\,\E\!\Big[\int_0^TF_t^{\tfrac{2\,p}{p-2}}\,\dt\Big]^{\tfrac{p-2}{p}}<\infty\,.$

\vspace{.3cm}\noindent
\textbf{Step 3.} Next, we show that the performance criterion  \eqref{eqn:criterion2} and $J$ in \eqref{eq:functional J} agree up to a constant on bounded processes. Take $\nu\in\A_2$ such that $|\nu|\leq N$ for some constant $N$. Then $\left|Q^\nu_t\right|=\left|Q_0+\int_0^t\nu_s\,\ds\right|\leq |Q_0|+T\,N\,.$
By \eqref{eq: dyn F} and \eqref{eqn: dyn Q}, we write $Q_{T}^{\nu}\,F_T$ as 
\begin{align}
Q_0\,F_0+\int_0^T Q^\nu_t\diff F_t+\int_0^TF_t\,\diff Q^\nu_t
=Q_0\,F_0+\int_0^TQ^\nu_t\,A_t\,F_t\dt+\int_0^TF_t\,\nu_t\dt+\sigma\int_0^TQ^\nu_t\,F_t\diff W_t\,.
\end{align}
Because $\E\Big[\int_0^TF_t^2\,\left|Q^\nu_t\right|^2\,\dt\Big]
    \leq(|Q_0|+T\,N)^2\,\E\Big[\int_0^TF_t^2\,\dt\Big]<\infty\,,$ the process
$\int_0^tQ^\nu_s\,F_s\,\diff W_s$
is a martingale, so $\E\Big[\int_0^TQ^\nu_t\,F_t\,\diff W_t\Big]=0\,.$ Thus we may further write \eqref{P lemma 1} as

\begin{align}
    &\E\!\Big[Q^\nu_T\,F_T\Big]-\eta\,\Vert\nu\Vert^2-\frac{\phi}{2}\,\Vert\mathfrak{Q}\nu\Vert^2-\langle F,\nu\rangle-\phi\,\langle Y+Q_0,\mathfrak{Q}\nu\rangle-\frac{\phi}{2}\int_{0}^{T}(Y_t+Q_0)^2\,\dt - Q_0\,F_0
    \\
    &=\E\!\Big[\int_{0}^{T}\big\{Q_0A_tF_t-\tfrac{\phi}{2}(Y_t+Q_0)^2\big\}\, \dt\Big]+\langle A\,F,\mathfrak{Q}\nu\rangle-\eta\,\langle\nu,\nu\rangle-\tfrac{\phi}{2}\,\langle\mathfrak{Q}\nu,\mathfrak{Q}\nu\rangle-\phi\langle Y+Q_0,\mathfrak{Q}\nu\rangle
    \\
    &=\E\!\Big[\int_{0}^{T}\big\{Q_0\,A_t\,F_t-\tfrac{\phi}{2}\,(Y_t+Q_0)^2\big\}\, \dt\Big]-\frac{1}{2}\,\langle\Lambda\nu,\nu\rangle+\langle b,\nu\rangle\,,
\end{align}
that is, we may rewrite the performance criterion \eqref{eqn:criterion2} as
\begin{align}
Q_0\,F_0+\E\!\Big[\int_{0}^{T}\left\{Q_0\,A_t\,F_t-\tfrac{\phi}{2}\,(Y_t+Q_0)^2\right\}\, \dt\Big]+J[\nu]\,.\label{eqn:agree}
\end{align}

\noindent
\textbf{Step 4.} Because bounded processes are dense in $\A_2$, by continuity, \eqref{eqn:agree} holds for all $\nu\in\A_2$.
\qed

\subsection{Proof of Proposition \ref{proposition strict concavity no transient}}\phantomsection\label{proof:proposition strict concavity no transient}

\vspace{.3cm}\noindent
First, we show $J$ is strictly concave. Take $\nu,\zeta\in\A_2$ and $\rho\in(0,1)$. By Lemma~\ref{lem:perfcrit no transient},
\begin{align}
    J&[\rho\,\nu+(1-\rho)\,\zeta]=-\frac{1}{2}\,\langle\Lambda(\rho\,\nu+(1-\rho)\,\zeta),\rho\,\nu+(1-\rho)\,\zeta\rangle+\langle b,\rho\,\nu+(1-\rho)\,\zeta\rangle
    \\
    &=-\frac{1}{2}\left(\rho^2\,\langle\Lambda\nu,\nu\rangle+2\,\rho\,(1-\rho)\,\langle\Lambda\nu,\zeta\rangle+(1-\rho)^2\,\langle\Lambda\zeta,\zeta\rangle\right) +\rho\,\langle b,\nu\rangle+(1-\rho)\,\langle b,\zeta\rangle
    \\
    &=\frac{1}{2}\,\rho\,(1-\rho)\,\left(\langle\Lambda\nu,\nu\rangle-2\,\langle\Lambda\nu,\zeta\rangle+\langle\Lambda\zeta,\zeta\rangle\right)+\rho\,J[\nu]+(1-\rho)\,J[\zeta]
    \\
    &=\frac{1}{2}\,\rho\,(1-\rho)\,\left(\eta\,\Vert\nu-\zeta\Vert^2+\frac{\phi}{2}\,\Vert\mathfrak{Q}(\nu-\zeta)\Vert^2\right)+\rho\,J[\nu]+(1-\rho)\,J[\zeta]
    \;\geq\rho\,J[\nu]+(1-\rho)\,J[\zeta]\,,
\end{align}
with equality if and only if $\nu=\zeta$. Hence, $J$ is strictly concave. Next, we show differentiability of $J$. Take $\nu,\delta\in\mathcal{A}_2$ and $\epsilon>0$. Then
\begin{align}
    \frac{1}{\epsilon}\,(J[\nu+\epsilon\,\delta]-J[\nu])&=\frac{1}{\epsilon}\,\left(-\frac{1}{2}\,\langle\Lambda(\nu+\epsilon\,\delta),\nu+\epsilon\,\delta\rangle+\langle b,\nu+\epsilon\,\delta\rangle+\frac{1}{2}\langle\Lambda\nu,\nu\rangle-\langle b,\nu\rangle\right)
    \\
    &=\frac{1}{\epsilon}\,\left(-\frac{1}{2}\,\langle\Lambda\nu,\nu\rangle-\epsilon\,\langle\Lambda\nu,\delta\rangle-\frac{\epsilon^2}{2}\,\langle\Lambda\delta,\delta\rangle+\langle b,\epsilon\,\delta\rangle+\frac{1}{2}\langle\Lambda\nu,\nu\rangle\right)
    \\
    &=-\langle\Lambda\nu,\delta\rangle-\frac{\epsilon}{2}\,\langle\Lambda\delta,\delta\rangle+\langle b,\delta\rangle
\end{align}
It follows that the G\^ateaux derivative $\D J[\nu](\delta)$ at $\nu\in\A_2$ is $\lim_{\epsilon\downarrow0}\dfrac{J[\nu+\epsilon\delta]-J[\nu]}{\epsilon}=\langle-\Lambda\nu+b,\delta\rangle\,.$
We identify $\D J[\nu]$ with $-\Lambda\nu+b$. From \eqref{defeq:Lambda no transient} and \eqref{defeq:Lambda no transient}, we get
\begin{equation}
    \D J[\nu]=-\Lambda\nu+b=-2\,\eta\,\nu+\mathfrak{Q}^\top(A\,F-\phi\,(Y+Q_0+\mathfrak{Q}\nu))\,.
    \label{eqn:DJHilbertelement no transient}
\end{equation}
To obtain an expression for $\mathfrak{Q}^\top$, write
\begin{align}\label{derivation Q transp}
    \begin{split}\langle\mathfrak{Q}\nu,\zeta\rangle=\E\!\Big[\int_0^T\int_0^t\nu_s\,\ds&\,\zeta_t\,\dt\Big]=\E\!\Big[\int_0^T\nu_s\int_s^T\zeta_t\,\dt\,\ds\Big]=\int_0^T\E\!\Big[\nu_s\int_s^T\zeta_t\,\dt\Big]\,\ds
    \\&=\int_0^T\E\!\Big[\E\!\Big[\left.\nu_s\int_s^T\zeta_t\,\dt\,\right|\,\F_s\Big]\Big]\,\ds=\int_0^T\E\!\Big[\nu_s\,\E\!\Big[\left.\int_s^T\zeta_t\,\dt\,\right|\,\F_s\Big]\Big]\,\ds.
    \end{split}
\end{align}
Thus, $\langle\mathfrak{Q}\nu,\zeta\rangle=\E\!\Big[\int_0^T\nu_s\,\E\!\Big[\left.\int_s^T\zeta_t\,\dt\,\right|\,\F_s\Big]\,\ds\Big]\,.$ 

\noindent 
Conclude that  $\mathfrak{Q}^\top$ is given by $(\mathfrak{Q}^\top\zeta)_t=\E\!\Big[\left.\int_t^T\zeta_s\,\ds\,\right|\,\F_t\Big]=\E\!\Big[\left.\int_0^T\zeta_s\,\ds\,\right|\,\F_t\Big]-\int_0^t\zeta_s\,\ds,$ 
where the in the last expression, the martingale term is c\`adl\`ag, so the entire process is c\`adl\`ag and thus progressively measurable.    It follows from \eqref{eqn:DJHilbertelement no transient} that $\D J[\nu]_t$ is given by \eqref{eq:gateauxder no transient} as desired.
\qed

\subsection{Proof of Theorem \ref{theorem optimal speed no transient}}\phantomsection\label{proof:theorem optimal speed no transient}

\vspace{.3cm}\noindent
Suppose $\D J[\nu^\star]=0$ for some $\nu^\star\in\A_2$. Then by Proposition~\ref{proposition strict concavity no transient} we have
\begin{align}
    2\,\eta\,\nu^\star_{t}&=\E\!\Big[\left.\int_{t}^{T}\left(A_{s}\,F_s-\phi\left(Y_{s}+Q_{s}^{\nu^\star}\right)\right)\,\ds\,\right|\,\F_t\Big]
        \\
        &=\E\!\Big[\left.\int_{0}^{T}\left(A_{s}\,F_s-\phi\left(Y_{s}+Q_{s}^{\nu^\star}\right)\right)\,\ds\,\right|\,\F_t\Big]-\int_{0}^{t}\left(A_{s}\,F_s-\phi\left(Y_{s}+Q_{s}^{\nu^\star}\right)\right)\,\ds\,.
\end{align}
Note that the process $M$, defined by
    \begin{equation}
        M_t=\E\!\Big[\left.\int_0^T\left(A_s\,F_s-\phi\,\left(Y_s+Q^{\nu^\star}_s\right)\right)\,\diff s\,\right|\,\F_t\Big]\,,\label{def mart m}
    \end{equation}
    is a martingale with $M_T\in L^2(\Omega)$. It follows that
    \begin{align}
        2\,\eta\,\nu^\star_s&=M_t-\int_0^t\big(A_s\,F_s-\phi\,\big(Y_s+Q^{\nu^\star}_s\big)\big)\,\diff s
        \\
        &=M_t-\int_0^T\big(A_s\,F_s-\phi\,\big(Y_s+Q^{\nu^\star}_s\big)\big)\,\diff s+\int_t^T\big(A_s\,F_s-\phi\,\big(Y_s+Q^{\nu^\star}_s\big)\big)\,\diff s
        \\
        &=-\int_t^T\big(-A_s\,F_s+\phi\,\big(Y_s+Q^{\nu^\star}_s\big)\big)\,\diff s-(M_T-M_t)\,.\label{2 eta nu no transient}
    \end{align}
Thus $\nu^\star$ satisfies the FBSDE \eqref{eq:FBSDE no transient}.

\vspace{.3cm}\noindent
Conversely, assume $\nu^\star\in\A_2$ satisfies the FBSDE \eqref{eq:FBSDE no transient} for some martingale $M$ such that $M_T\in L^2(\Omega)$. By the dynamics of $\nu^\star$ and $Q$, we may write
\begin{equation}
    2\,\eta\,\nu^\star_t=\int_t^T\left(A_s\,F_s-\phi\,(Y_s+Q^{\nu^\star}_s)\right)\,\ds-M_T+M_t\,.\label{2 eta nu no transient 2}
\end{equation}
Taking conditional expectation on above equation gives $\D J[\nu^\star]_t=0$ because

\begin{equation}
    2\,\eta\,\nu^\star_t=\E\!\Big[\left.\int_t^T\left(A_s\,F_s-\phi\,\left(Y_s+Q^{\nu^\star}_s\right)\right)\,\ds\,\right|\,\F_t\Big]\,.\label{2 eta nu no transient 3}
\end{equation} 

\subsection{Proof of Proposition \ref{prop:notrasient}}\phantomsection\label{proof:prop:notrasient}

\vspace{.3cm}\noindent
The LP's optimisation problem reduces to solving the  FBSDE \eqref{eq:FBSDE no transient}. The ansatz $\nu_t=P(t)\,Q_t+\ell_t$ gives the equations
\begin{align}
P^\prime(t)&=-P(t)^2+\tfrac{\phi}{2\,\eta},\qquad P(T)=0\label{eqn:dp} \\
\diff\ell_t&=\left(-P(t)\,\ell_t+\tfrac{-A_t\,F_t+\phi\,Y_t}{2\,\eta}\right)\dt+\tfrac{1}{2\,\eta}\,\diff M_t,\quad \ell_T=0.\label{eqn:dell}
\end{align}
The solution of \eqref{eqn:dp} is $P$ in \eqref{eq:solP no transient}. To solve \eqref{eqn:dell}, we define $\tilde P(s,t)$ as in \eqref{eq:solP no transient}, 
and use generalised It\^o's formula to write
\begin{align*}
    \tilde{P}(0,t)\,\ell_t&=-\int_t^T\tilde{P}(0,s)\,P(s)\,\ell_s\,\ds-\int_t^T\tilde{P}(0,s)\,\diff\ell_s
    \\
    &=-\int_t^T\tilde{P}(0,s)\,P(s)\,\ell_s\,\ds-\int_t^T\tilde{P}(0,s)\,\left(-P(s)\,\ell_s+\tfrac{-A_s\,F_s+\phi\,Y_s}{2\,\eta}\right)\,\ds -\tfrac{1}{2\,\eta}\int_t^T\tilde{P}(0,s)\,\diff M_s
    \\
    &=\tfrac{1}{2\,\eta}\int_t^T\tilde{P}(0,s)\,(A_s\,F_s-\phi\,Y_s)\,\ds-\tfrac{1}{2\,\eta}\int_t^T\tilde{P}(0,s)\,\diff M_s,
\end{align*}
therefore, $\ell_t=\tfrac{1}{2\,\eta}\,\E\Big[\left.\int_t^T\tilde{P}(t,s)\,(A_s\,F_s-\phi\,Y_s)\,\ds\,\right|\,\F_t\Big].$ Similarly, $Q$ is obtained by solving the equation $\diff Q_t=(P(t)\,Q_t+\ell_t)\,\dt,$ whose solution is
\begin{align}\label{eq:Qtgeneral}
    Q_t=Q_0\,\tilde{P}(0,t)+\int_0^t\tilde{P}(s,t)\,\ell_s\,\ds.
\end{align}
Finally, $\nu_t=P(t)\,Q_t+\ell_t=P(t)\,\left(Q_0\,\tilde{P}(0,t)+\int_0^t\tilde{P}(s,t)\,\ell_s\,\ds\right)+\ell_t.$
\qed

\subsection{Proof of Proposition \ref{prop:valuefndfd no transient}}\phantomsection\label{proof:prop:valuefndfd no transient}

\vspace{.3cm}\noindent
Let $\kappa>0$. Applying It\^o's formula to \eqref{eqn:dyn Y} and \eqref{eq: dyn F}, we write $Y_T\,F_T$ as  
\begin{equation}\label{ito YF}
Y_0\,F_0+\int_0^T\big\{\left(G_t+\sigma^2\,\partial_1h(F_t,\kappa)\right)\,F_t^2+A_t\,F_t\,Y_t\big\}\dt+\sigma\int_0^TY_t\,F_t\diff W_t+\sigma\int_0^T\partial_1h(F_t,\kappa)F_t^2\diff W_t\,.
\end{equation}
Recall that $h(\cdot,\kappa)$ is the inverse of $-\partial_1\varphi(\cdot,\kappa)$, so
\begin{equation}
    -\partial_{11}\varphi(h(x,\kappa),\kappa)=\partial_1h(x,\kappa)^{-1}\,,\quad\forall x>0\,.\label{eqn:d_11varphi}
\end{equation}
and using $X_T=\varphi(Y_T,\kappa)$:
\begin{align}
    X_T-X_0&=\int_0^T\partial_1\varphi(Y_t,\kappa)\,\diff Y_t+\frac{1}{2}\int_0^T\partial_{11}\varphi(Y_t,\kappa)\,\diff\langle Y\rangle_t
    =-\int_0^TF_t\,\diff Y_t-\frac{1}{2}\int_0^T\frac{1}{\partial_1h(F_t,\kappa)}\,\diff\langle Y\rangle_t
    \\
    &=-\int_0^T\big(G_t+\tfrac{\sigma^2}{2}\,\partial_1h(F_t,\kappa)\big)\,F_t^2\,\dt-\sigma\int_0^T\partial_1h(F_t,\kappa)\,F_t^2\,\diff W_t\,.
\end{align}
Moreover, \eqref{eqn:d_11varphi} implies $\Pi(F_{t},\kappa)=\dfrac{\lambda\,\pi\,(v-\pi)\,F_{t}^{2}}{\partial_{11}\varphi\left(h(F_{t},\kappa),\kappa\right)}=\lambda\,\pi\,(\pi-v)\,\partial_1h(F_t,\kappa)\,F_{t}^{2}\,.$ Combining these:
\begin{align*}
   & \int_0^T\Pi(F_t,\kappa)\,\dt+X_T -X_0+ Y_T\, F_T\\ & =Y_0\,F_0+\int_0^T\left\{\left(\tfrac{\sigma^2}{2}+\lambda\,\pi\,(\pi-v)\right)\,\partial_1h(F_t,\kappa)\,F_t^2+A_t\,F_t\,Y_t\right\}\,\dt+\sigma\int_0^TY_t\,F_t\,\diff W_t
\end{align*}
Because the processes $\partial_1h(F,\kappa)$, $F^2$, $A\,F$, and $Y^2$ are all in $\A_2$
we obtain that
\begin{equation}
    \E\!\Big[\int_0^T\big\{\big(\frac{\sigma^2}{2}+\lambda\,\pi\,(\pi-v)\big)\,\partial_1h(F_t,\kappa)\,F_t^2+A_t\,F_t\,Y_t\big\}\dt\Big]=\big(\tfrac{\sigma^2}{2}+\lambda\pi(\pi-v)\big)\langle\partial_1h(F,\kappa),F^2\rangle+\langle A\,F,Y\rangle
\end{equation}
is well-defined and $\E\!\Big[\int_0^TY_t^2\,F_t^2\,\dt\Big]\leq\left\Vert Y^2\right\Vert\,\left\Vert F^2\right\Vert<\infty\,.$ Thus, $\E\!\Big[\int_0^TY_t\,F_t\,\diff W_t\Big]=0\,.$ These and Lemma~\ref{lem:perfcrit no transient} imply \eqref{eqn:valuefn no transient} is well-defined and can be written as $J[\nu^\star]+\frac{\phi}{2}\,\Vert\mathfrak{Q}\nu^\star+Y-Y_0\Vert^2+\hat{H}\,,$ where
\begin{equation}
    \hat{H}=\E\!\Big[\int_{0}^{T}\left\{\left(\frac{\sigma^2}{2}+\lambda\,\pi\,(\pi-v)\right)\,\partial_1h(F_t,\kappa)\,F_t^2+A_t\,F_t\,(Y_t-Y_0)-\frac{\phi}{2}\,(Y_t-Y_0)^2\right\}\,\dt\Big]\,.
\end{equation}

\noindent
Next, we show $\hat{H}$ and $J[\nu^\star]+\frac{\phi}{2}\,\Vert\mathfrak{Q}\nu^\star+Y-Y_0\Vert^2$ are both continuous in $\kappa$. To that end, fix $\kappa_n\to\kappa$. Because $|Y_t(\kappa_n)-Y_t(\kappa)|=|h(F_t,\kappa_n)-h(F_t,\kappa)|\leq\left(F_t^{\mathfrak{p}}+F_t^{\mathfrak{q}}\right)\,|\mathfrak{C}(\kappa_n)-\mathfrak{C}(\kappa)|\,,$ we have
\begin{align}
    \Vert Y(\kappa_n)-Y(\kappa)\Vert=\E\!\Big[\int_0^T|Y_t(\kappa_n)-Y_t(\kappa)|^2\,\dt\Big]^{1/2}&\leq|\mathfrak{C}(\kappa_n)-\mathfrak{C}(\kappa)|\,\E\!\Big[\int_0^T\left(F_t^{\mathfrak{p}}+F_t^{\mathfrak{q}}\right)^2\,\dt\Big]^{1/2}
    \\
    &\leq|\mathfrak{C}(\kappa_n)-\mathfrak{C}(\kappa)|\,\left(\left\Vert F^\mathfrak{p}\right\Vert+\left\Vert F^{\mathfrak{q}}\right\Vert\right)
\end{align}
and $\Vert Y_0(\kappa_n)-Y_0(\kappa)\Vert\leq|\mathfrak{C}(\kappa_n)-\mathfrak{C}(\kappa)|\,\left(\left\Vert F_0^\mathfrak{p}\right\Vert+\left\Vert F_0^{\mathfrak{q}}\right\Vert\right).$ Thus, the map $\kappa\mapsto Y(\kappa)-Y_0(\kappa)$ from $(0,\infty)$ to $\A_2$ is continuous. It follows that $\kappa\mapsto\hat{H}(\kappa)$ is continuous as the composition of
\begin{equation}
    \zeta\mapsto\left(\frac{\sigma^2}{2}+\lambda\,\pi\,(\pi-v)\right)\,\left\langle\partial_1h(F,\kappa),F^2\right\rangle+\langle A\,F,\zeta\rangle-\frac{\phi}{2}\,\Vert\zeta\Vert^2
\end{equation}
with $\kappa\mapsto Y(\kappa)-Y_0(\kappa)$. Next, we consider $J[\nu^\star]+\frac{\phi}{2}\,\Vert\mathfrak{Q}\nu^\star+Y-Y_0\Vert^2$. Recall the bounded linear operator $\Lambda=2\,\eta+\phi\,\mathfrak{Q}^\top\mathfrak{Q}$ in Lemma~\ref{lem:perfcrit no transient}. Use $\langle\Lambda\nu,\nu\rangle=2\,\eta\Vert\nu\Vert^2+\phi\Vert\mathfrak{Q}\nu\Vert^2\geq2\,\eta\Vert\nu\Vert^2\,,$ to obtain that  $\Lambda$ is coercive. By Lax-Milgram lemma, $\Lambda$ has an inverse $\Lambda^{-1}$, which is a bounded linear operator on $\A_2$, and $\nu^\star$ is given by $\nu^\star=\Lambda^{-1}b$, so
\begin{align}
    J[\nu^\star]+\frac{\phi}{2}\,\Vert\mathfrak{Q}\nu^\star+Y-Y_0\Vert^2&=-\frac{1}{2}\,\left\langle\Lambda\Lambda^{-1}b,\Lambda^{-1}b\right\rangle+\left\langle b,\Lambda^{-1}b\right\rangle+\frac{\phi}{2}\,\Vert\mathfrak{Q}\Lambda^{-1}b+Y-Y_0\Vert^2
    \\
    &=\frac{1}{2}\,\left\langle b,\Lambda^{-1}b\right\rangle+\frac{\phi}{2}\,\Vert\mathfrak{Q}\Lambda^{-1}b+Y-Y_0\Vert^2\,,
\end{align}
where $b=\mathfrak{Q}^\top(A\,F-\phi\,(Y-Y_0))\,.$ 
Because $\kappa\mapsto b(\kappa)=\mathfrak{Q}^\top(A\,F-\phi\,(Y(\kappa)-Y_0(\kappa)))$ is continuous, then $\kappa\mapsto J[\nu^\star](\kappa)$ is continuous. \qed

\subsection{Proof of Proposition \ref{prop:liq no hedge}}\phantomsection\label{proof:prop:liq no hedge}
\noindent
Recall that the stage-three trading volumes generate fee revenue \eqref{eq:Pi}. In the case of a CPM, these write $\Pi(F_t,\kappa)=\gamma\, \kappa\,\sqrt{F_{t}}\,,$ where we define $\gamma$ as in \eqref{eq:gamma def}. The solutions in~\eqref{eq:nugeneral2}--\eqref{eq:Qtgeneral} become
$Q^\star_t=\kappa\, C^Q_t+D^Q_t$ and $\nu^\star_t=\kappa\,C^\nu_t+D^\nu_t.$
In the case of CPMs, the profitability function is \eqref{eq:Pi CPM}, and the initial DEX reserves in asset $X$ for a given liquidity supply $\kappa$ are $X_0 = \kappa\,F_0^{1/2},$ and the terminal holdings in the DEX are $X_T + Y_T\,F_T = 2\,\kappa\,F_T^{1/2}$. Thus, when $\nu=\nu^\star$,  the value function \eqref{eqn:valuefn no transient} is
\begin{align*}
    &\E\!\Big[\int_0^T\gamma\,F_t^{1/2}\,\dt+2\,F_T^{1/2}-F_0^{1/2}\Big]\kappa+\E\!\Big[Q^\star_T\,F_T-\int_{0}^{T}\left(F_t+\eta\,\nu^\star_{t}\right)\,\nu^\star_{t}\,\dt\Big]
    \\
    &=\E\!\Big[\int_0^T\gamma\,F_t^{1/2}\,\dt+2\,F_T^{1/2}-F_0^{1/2}\Big]\kappa
    +\E\!\Big[\left(C^Q_T\,\kappa+D^Q_T\right)\,F_T
    \\
    &\qquad\ -\int_{0}^{T}\left(F_t+\eta\,\left(C^\nu_t\,\kappa+D^\nu_t\right)\right)\left(C^\nu_t\,\kappa+D^\nu_t\right)\,\dt\Big]
    \\
    &=-\E\!\Big[\int_0^T\eta\,\left(C^\nu_t\right)^2\,\dt\Big]\,\kappa^2+\E\!\Big[2\,F_T^{1/2}-F_0^{1/2}+C^Q_T\,F_T+\int_0^T\left(\gamma\,F_t^{1/2}-C^\nu_t\,F_t-2\,\eta\,C^\nu_t\,D^\nu_t\right)\,\dt\Big]\kappa
    \\
    &\quad\ +\E\!\Big[D^Q_T\,F_T-\int_0^T\left(F_t\,D^\nu_t+\eta\,\left(D^\nu_t\right)^2\right)\,\dt\Big]\,.
\end{align*}
In this case the optimal level of liquidity $\kappa$ is
\begin{align}\label{eq:kappahedge}
    \kappa^{\star}=\frac{\E\!\Big[2\,F_T^{1/2}-F_0^{1/2}+C^Q_T\,F_T+\int_0^T\left(\gamma\,F_t^{1/2}-C^\nu_t\,F_t-2\,\eta\,C^\nu_t\,D^\nu_t\right)\,\dt\Big]}{2\,\eta\,\E\!\Big[\int_0^T\left(C^\nu_t\right)^2\,\dt\Big]}
\end{align}
By It\^o's formula, $2\,F_T^{1/2}=2\,F_0^{1/2}+\int_0^T(A_t-\frac{\sigma^2}{4})\,F_t^{1/2}\,\dt+\sigma\int_0^TF_t^{1/2}\,\diff W_t,$ and $C^Q_T\,F_T=-F_0^{1/2}+\int_0^T(C^\nu_t\,F_t+C^Q_t\,A_t\,F_t)\,\dt+\sigma\int_0^TC^Q_t\,F_t\,\diff W_t\,.$ Therefore, $\kappa^\star$ is as in \eqref{eq:kappahedgeA} where the terms $\E[\int_0^TF_t^{1/2}\,\diff W_t]$ and $\E[\int_0^TC^Q_t\,F_t\,\diff W_t]$ vanish because $\E[\int_0^TF_t\,\dt]<\infty$ and
\begin{align}
\E&\Big[\int_0^T\big|C^Q_t\big|^2\,F_t^2\,\dt\Big]\leq\E\Big[\int_0^T\big|C^Q_t\big|^4\,\dt\Big]^{1/2}\,\E\Big[\int_0^TF_t^4\,\dt\Big]^{1/2}\\
&=\E\Big[\int_0^T\big|\int_0^t\tilde{P}(s,t)\,C^\ell_s\,\ds-F_0^{-1/2}\,\tilde{P}(0,t)\big|^4\,\dt\Big]^{1/2}\,\E\Big[\int_0^TF_t^4\,\dt\Big]^{1/2}\\
&\lesssim\Big(\E\Big[\int_0^T\big|C^\ell_t\big|^4\,\dt\Big]+F_0^{-2}\Big)^{1/2}\,\E\Big[\int_0^TF_t^4\,\dt\Big]^{1/2}\\
&\leq\Big(\int_0^T\E\Big[\big|\int_t^T\tilde{P}(t,s)\,F_s^{-1/2}\,\ds\big|^4\Big]\,\dt+F_0^{-2}\Big)^{1/2}\,\E\Big[\int_0^TF_t^4\,\dt\Big]^{1/2}
\\&\lesssim\Big(\E\Big[\int_0^TF_t^{-2}\,\dt\Big]+F_0^{-2}\Big)^{1/2}\,\E\Big[\int_0^TF_t^4\,\dt\Big]^{1/2}\ \ <\infty\,.
\end{align}
\qed

\section{DEX liquidity with CEX transient price impact}\label{sec:transient}

In the main model, we assumed that the LP's trading activity in the CEX does not affect prices and model trading costs through a quadratic trading penalty. This assumption is appropriate when  LPs are small relative to CEX liquidity.

However, when the DEX is large and LPs rely on the centralised venue to manage their risk, their trading activity affects CEX prices, even though the DEX itself remains a secondary market. In this appendix, we relax our model assumption and allow LPs' trades to generate transient price impact in the CEX. 

The stage-three execution problem of LTs in the DEX, described in Section~3 of the main paper, depends only on the pool's reserves, the fundamental price, and the fee structure. In particular, it does not depend on the LP's trading activity in the CEX. As a result, the analysis of stage~three is unchanged, and we take the solution from Section~3 of the main paper as given. Below, we solve stages~two and~one in the presence of transient price impact.

The risky asset's mid-price $S^\nu = (S_t^\nu)_{t \in [0, T]}$  in the CEX now has two components: the fundamental price $F$ and a transient market impact $I^\nu = (I_t^\nu)_{t \in [0, T]}$ induced by the LP's trades in the CEX. Formally,
\begin{equation}
S_t^\nu = F_t + I_t^\nu, \qquad t \in [0, T].
\end{equation}
The transient impact process $I^\nu$ satisfies
\begin{equation}\label{eq: dyn I}
I_t^\nu = \int_0^t \left(c\, \nu_s - \psi\,I^\nu_s\right) \,\diff s\,\quad \implies\quad
I^\nu_t=c\int_0^te^{\psi\,(s-t)}\,\nu_s\,\diff s\,.
\end{equation}

The transient impact term models the fact that LP trades in the CEX push prices temporarily in the direction of their trading activity: buying pressure increases prices, while selling pressure decreases them. When trading slows or stops, this price distortion decays over time as liquidity replenishes. Here, the parameter $c > 0$ measures the linear price pressure of the LP's trades, and the resilience parameter $\beta > 0$ governs the decay of transient impact.

In the presence of transient price impact, the LP's wealth in the pool $(L_t^{\nu})_{t\in[0,T]}$ is indirectly affected by trading activity $\nu$   through the cumulative effect of price impact. We 
write
\begin{equation}
L_t^{\nu}
:= \int_{0}^{t} \Pi(F_u,\kappa)\,du + X_t + Y_t\,S_t^{\nu}.
\end{equation}
As before, the first term represents cumulative fee revenue from noise LTs, while the second and third terms correspond to the marked-to-market value of the LP's liquidity position, now evaluated at the impacted CEX price.

\subsection{Stage two: trading in the CEX}

\paragraph{The performance criterion.} As in the main model, the LP takes the liquidity supply $\kappa$ of stage one as given. The LP holds reserves $\{X_t, Y_t\}$ in the DEX and inventory  $\left(Q_t^\nu\right)_{t \in [0, T]}$ in the CEX. The LP's terminal holdings in the CEX are $Q_T^\nu$, which she values at the terminal CEX price $S_T^\nu$. The LP maximises her terminal wealth subject to penalties for deviating from a perfect replication strategy, so the LP's performance criterion, when employing the admissible strategy $\nu$ is 
\begin{equation}\label{eqn:criterion2 transient}\tag{PT}
    \E\bigg[\underbrace{\left(Y_{T}+Q_{T}^{\nu}\right)\,S_{T}^{\nu}}_\text{combined CEX-DEX position}-\underbrace{\int_{0}^{T}\left(S_{t}^{\nu}+\eta\,\nu_{t}\right)\,\nu_{t}\,\dt}_\text{risk offsetting}-\underbrace{\dfrac{\phi}{2}\int_{0}^{T}\left(Q_{t}^{\nu}+Y_{t}\right)^{2}\,\dt\bigg]}_\text{deviation penalty}\,.
\end{equation}

Define the two bounded linear operators $\mathfrak{Q}\,,\mathfrak{I}:\A_2\to\A_2$ given by\footnote{The operators are well defined using the inequalities
\begin{align}
    \E\!\left[\int_0^T\left|\int_0^t\nu_s\,\ds\right|^2\,\dt\right]&\leq\E\!\left[\int_0^Tt\int_0^t|\nu_s|^2\,\ds\,\dt\right]\leq T^2\;\E\!\left[\int_0^T|\nu_t|^2\,\dt\right]
\end{align}
\begin{align}
    \text{and }\qquad \E\!\left[\int_0^T\left|\int_0^te^{\beta\,(s-t)}\nu_s\,\ds\right|^2\,\dt\right]&\leq\E\!\left[\int_0^Tt\int_0^t|\nu_s|^2\,\ds\,\dt\right]\leq T^2\;\E\!\left[\int_0^T|\nu_t|^2\,\dt\right]\,.
\end{align}}
\begin{equation}
    (\mathfrak{Q}\nu)_t=\int_0^t\nu_s\,\ds\qquad\text{and}\qquad(\mathfrak{I}\nu)_t=c\int_0^te^{\beta(s-t)}\,\nu_s\,\ds\,.
\end{equation}
Notice that $Q^\nu=Q_0+\mathfrak{Q}\nu$ and $I^\nu=\mathfrak{I}\nu$. The following result shows that the performance criterion \eqref{eqn:criterion2 transient} is a real-valued functional on $\A_2$. 

\begin{lemma}\label{lem:perfcrit transient}
Let $G$ be defined as $G_t := \partial_1 h(F_t,\kappa)\,A_t
+ \dfrac{\sigma^2}{2}\,\partial_{11}h(F_t,\kappa)\,F_t .$ The performance criterion \eqref{eqn:criterion2 transient} can be written as $J[\nu]+H$, where $H$ is the following well-defined real number that does not depend on $\nu$:
\begin{equation}
    H=(Y_0+Q_0)\,F_0+\left[\int_{0}^{T}\left\{\left(G_t+\sigma^2\,\partial_1h(F_t,\kappa)\right)\,F_t^2+(Y_t+Q_0)\,A_t\,F_t-\dfrac{\phi}{2}\,(Y_t+Q_0)^2\right\}\, \dt\right]
\end{equation}
and $J$ takes on the bounded linear-quadratic form:
\begin{equation}\label{eq:functional J tiS}
    J[\nu]=-\frac{1}{2}\,\langle\Lambda\nu,\nu\rangle+\langle b,\nu\rangle\,,
\end{equation}
where $\Lambda$ is a symmetric bounded linear operator on $\A_2$ and $b$ is an element of $\A_2$, given by
\begin{equation}
    \Lambda=2\,\eta+\beta\,(\mathfrak{I}^\top\mathfrak{Q}+\mathfrak{Q}^\top\mathfrak{I})-c\,(\mathfrak{Q}+\mathfrak{Q}^\top)+\phi\,\mathfrak{Q}^\top\mathfrak{Q}\,,\label{defeq:Lambda}
\end{equation}
\begin{equation}
    b=\mathfrak{I}^\top(G\,F)+(c-\beta\,\mathfrak{I}^\top)(Y+Q_0)+\mathfrak{Q}^\top(A\,F-\phi\,(Y+Q_0))\,.\label{defeq:b}
\end{equation}
\end{lemma}
\begin{proof}
    See Appendix \ref{proof:lem:perfcrit}. 
\end{proof}

\paragraph{The optimal risk-offseting strategy.} 
In this section, we make the following standing assumption.
\begin{assume}\label{assume:2}
    $c < \sqrt{2 \,\eta\, \phi}$.
\end{assume}
This assumption bounds the instantaneous impact of the LP's trades on CEX prices and ensures that such impacts are offset by sufficiently high trading costs and deviation penalty. This prevents degenerate strategies that would otherwise push prices to infinity. Assumption~\ref{assume:2} is not very restrictive, as the parameter $\phi$ is typically large to reflect the LP's preference for strategies that closely replicate the LP's position in the DEX. Moreover, trading costs $\eta$ associated with immediate execution costs in the CEX are typically of a larger order of magnitude than the impact parameter~$c$.

To solve the problem, it is convenient to write the objective in a quadratic form on the Hilbert space $\mathcal A_2.$  Next, we adapt our results to the case of transient price impact.

\begin{proposition}\label{proposition strict concavity} 
The objective $J$ defined in Lemma~\ref{lem:perfcrit transient} is strictly concave. Moreover, $J$ defined in Lemma~\ref{lem:perfcrit transient} is G\^ateaux differentiable, and its G\^ateaux derivative $\D J[\nu]$ at $\nu\in\A_2$ is an element of $\A_2$, given by
\begin{equation}\label{eq:gateauxder}
\begin{split}
    \D J[\nu]_t&=-2\,\eta\,\nu_{t}+c\,\left(Y_{t}+Q_{t}^{\nu}\right)+\E\left[\left.\int_{t}^{T}\left(A_{s}\,F_s+c\,\nu_{s}-\beta \,I_{s}^{\nu}-\phi\left(Y_{s}+Q_{s}^{\nu}\right)\right)\,\ds\,\right|\,\F_t\right]
    \\
    &\quad\ +c\,e^{t\,\beta}\,\E\left[\left.\int_{t}^{T}e^{-s\,\beta}\left(G_{s}\,F_s-\beta\,\left(Y_{s}+Q_{s}^{\nu}\right)\right)\,\ds\,\right|\,\F_t\right]\,.
\end{split}
\end{equation}
\end{proposition}

\begin{theorem} \label{theorem optimal speed}
The Gâteaux derivative \eqref{eq:gateauxder} vanishes at $\nu^\star \in \mathcal A_2$ if and only if $\nu^\star$ solves the FBSDE
\begin{align}\label{eq:FBSDE}
\begin{split}
\left\{
\begin{array}{rlrl}
   2\,\eta\,\diff\nu^\star_{t}  & =\left(-A_t\,F_t+\beta\,I_{t}+(\phi+c\,\beta)\,\left(Y_t+Q_t\right)+c\,\beta\,Z_t\right)\,\dt+\diff M_{t}, & 2\,\eta\,\nu^\star_{T} & =c\,\left(Y_T+Q_T\right)\,, 
   \\
   \diff Z_{t}  & =\left(\beta\,\left(Z_t+Y_t+Q_t\right)-G_t\,F_t\right)\,\dt+\diff N_t, & Z_{T} & =0\,, 
   \\
   \diff I_{t}  & =\left(c\,\nu^\star_{t}-\beta\,I_{t}\right)\,\dt, & I_0 & =0\,, 
   \\
   \diff Q_{t}  & =\nu^\star_t\,\dt\,, &
\end{array}
\right.
\end{split}
\end{align}
for some $\mathbb F$-martingales $M$ and $N$ such that $M_T,N_T\in L^2(\Omega)$.
\end{theorem}
\begin{proof}
See Appendix \ref{proof:proposition strict concavity} 
    and \ref{proof:theorem optimal speed}.
\end{proof}

The next result shows that the solution to the LP's problem reduces to the solution of a differential Riccati equation, whose solution exists, is unique, and can be obtained efficiently numerically. 
\begin{proposition}\label{prop:FBSDEtoDRE}
Let
\[
\begin{aligned}
B_{11}&=\begin{pmatrix}-\beta & 0 \\ 0 & 0\end{pmatrix}, \quad
B_{12}=\begin{pmatrix}c & 0 \\ 1 & 0\end{pmatrix}, \quad
B_{21}=\frac{1}{2\,\eta}\begin{pmatrix}{\beta} & {\phi+c\beta} \\ 0 & 2\,\eta\,\beta\end{pmatrix}, \quad
B_{22}=\frac{1}{2\,\eta}\begin{pmatrix}0 & c\,\beta \\ 0 & 2\,\eta\,\beta\end{pmatrix}, \\[0.4em]
b_t&=\frac{1}{2\,\eta}\begin{pmatrix}-A_tF_t+(\phi+c\beta)Y_t \\[0.2em] 2\,\eta\left(\beta \,Y_t-G_t\,F_t\right)\end{pmatrix}, \quad
G=\frac{1}{2\,\eta}\begin{pmatrix}0 & c \\ 0 & 0\end{pmatrix}, 
\quad
K=\begin{pmatrix}0 \\ Q_0\end{pmatrix},
\quad
L=\frac{1}{2\,\eta}\begin{pmatrix}c\,Y_T \\ 0\end{pmatrix}.
\end{aligned}
\]
Suppose there exists a solution $P$, which is an $\mathbb R^{2\times2}$-valued $C^1$ function, to the DRE
\begin{equation}\label{eq:dreP}
P^\prime(t)+P(t)\,B_{11}+P(t)\,B_{12}\,P(t)-B_{21}-B_{22}\,P(t)=0\,,
\end{equation}
with terminal condition $P(T)=G$. Define the $\mathbb R^2$-valued processes $\ell$, $\Psi$, and $\Phi$ as follows:
\begin{equation}
\begin{cases}
   \ell_t&=e^{-\int_0^t\left(P(u)\,B_{12}-B_{22}\right)\,\du}\,\E\!\left[\left.L-\int_t^Te^{\int_0^s\left(P(u)\,B_{12}-B_{22}\right)\,\du}\,b_s\,\ds\,\right|\,\F_t\right]\,,\\ 
   \Phi_t&=e^{\int_0^t\left(B_{12}\,P(u)+B_{11}\right)\,\du}\,\left(K+\int_0^te^{-\int_0^s\left(B_{12}\,P(u)+B_{11}\right)\,\du}\,B_{12}\,\ell_s\,\ds\right)\,,\\
   \Psi(t)&=P(t)\,\Phi_t+\ell_t\,.
\end{cases}  
\end{equation}
Then $(\Phi,\Psi)$ is a solution to the FBSDE~\eqref{eq:FBSDE} with
\[
\Psi_t=\begin{pmatrix}\nu_t^\star \\ Z_t\end{pmatrix}\,,\quad
\Phi_t=\begin{pmatrix}I_t \\ Q_t\end{pmatrix}\,.
\]
\noindent
Moreover, under Assumption \ref{assume:2}, the DRE \eqref{eq:dreP} admits a unique solution.
\end{proposition}
\begin{proof}
    See Appendix \ref{proof:prop:FBSDEtoDRE}.
\end{proof}

\subsection{Stage one: liquidity supply}\label{sec:liqprovuniswap transient}

Assume that the LP starts with a CEX position $Q_0 = -Y_0 = -h(F_0,\kappa)$. Let $S_t^\star$, $Q_t^\star$, and $L_t^\star$ be the CEX price, inventory, and DEX wealth when the LP executes the optimal strategy $\nu_t^\star$ in the CEX, where
$$
L_t^\star:=\int_0^t\Pi(F_u,\kappa)\,\du+X_t + Y_t\, S^\star_t\,,
$$
and $\Pi$ is defined as
\begin{equation}\label{eq:Pi2}
\Pi_t := \Pi(F_{t},\kappa)
:= \frac{\lambda\,\pi\,(v-\pi)\,F_{t}^{2}}
{\partial_{11}\varphi\left(h(F_{t},\kappa),\kappa\right)}\,.    
\end{equation}
In the general case, the optimisation problem of stage one is
\begin{align}\label{eqn:valuefn transient}\tag{KT}
\sup_{\kappa\in[\underline{\kappa},\overline\kappa]}\mathbb{E}\!\left[
-X_0+L_T^{\star}
+ Q_T^{\star}\,S_T^{\star}
- \int_0^T \left(S_t^{\star} + \eta\, \nu^\star_t\right)\, \nu^\star_t \, \diff t
\right]\,,
\end{align}
where $\overline{\kappa}$ denotes the maximum admissible liquidity depth implied by the LP's budget constraint.  The next results show that the LP's objective is well defined and establish mild conditions under which it is continuous and therefore attains its maximum over the compact set $[\underline{\kappa},\overline{\kappa}]$.

\begin{proposition}\label{prop:valuefndfd}
    The objective in  \eqref{eqn:valuefn transient} is well-defined for all $\kappa>0$. Moreover, suppose there exist $\mathfrak{p},\mathfrak{q}\in\mathbb{R}$ and a continuous function $\mathfrak{C}:(0,\infty)\to(0,\infty)$ such that
    \begin{equation}
        |h(x,\kappa)-h(x,\kappa^\prime)|+|\partial_1h(x,\kappa)-\partial_1h(x,\kappa^\prime)|+|\partial_{11}h(x,\kappa)-\partial_{11}h(x,\kappa^\prime)|\leq \left(x^{\mathfrak{p}}+x^{\mathfrak{q}}\right)\,|\mathfrak{C}(\kappa)-\mathfrak{C}(\kappa^\prime)|
    \end{equation}
    for all $x,\kappa,\kappa^\prime>0$. Then the LP's objective in \eqref{eqn:valuefn transient} is continuous in $\kappa$ and therefore attains its maximum over the compact set $[\underline{\kappa},\overline{\kappa}]$.
\end{proposition}
\begin{proof}
    See Appendix~\ref{proof:prop:valuefndfd}.
\end{proof}

\allowdisplaybreaks

\subsection{Proofs}\phantomsection\label{apx:proofs2}

\subsubsection{Proof of Lemma \ref{lem:perfcrit transient}}\phantomsection \label{proof:lem:perfcrit}

\noindent
Take $\nu\in\A_2$. Comparing to those in Lemma~\ref{lem:perfcrit no transient} of the main paper, in Lemma~\ref{lem:perfcrit transient}, \eqref{eqn:criterion2 transient} gets the extra term
\begin{equation}
    \mathcal{E}_1=\E\!\left[Y_T\,F_T+Y_T\,I^\nu_T+Q^\nu_T\,I^\nu_T-\int_0^TI^\nu_t\,\nu_t\,\dt\right]\,,
\end{equation}
$H$ gets the extra term
\begin{equation}
    \mathcal{E}_2=Y_0\,F_0+\E\!\left[\int_0^T\left\{\left(G_t+\sigma^2\,\partial_1h(F_t,\kappa)\right)\,F_t^2+Y_t\,A_t\,F_t\right\}\,\dt\right]\,,
\end{equation}
and
$J[\nu]$ gets the extra term
\begin{equation}
    \mathcal{E}_3=-\frac{1}{2}\,\left\langle\left(\beta\,\left(\mathfrak{I}^\top\mathfrak{Q}+\mathfrak{Q}^\top\mathfrak{I}\right)-c\,\left(\mathfrak{Q}+\mathfrak{Q}^\top\right)\right)\nu,\nu\right\rangle+\left\langle\mathfrak{I}^\top(G\,F)+\left(c-\beta\,\mathfrak{I}^\top\right)(Y+Q_0),\nu\right\rangle
\end{equation}
Therefore, it suffices to show that all three extra terms are well-defined, $\mathcal{E}_1$ and $\mathcal{E}_3$ are continuous in $\nu$, and $\mathcal{E}_1=\mathcal{E}_2+\mathcal{E}_3$.

\vspace{.3cm}\noindent
\textbf{Step 1.} First, we show that $\mathcal{E}_1$ is well-defined and continuous. We know from \eqref{ineq:Q_t} that $Q^\nu_T\in L^2(\Omega)$. We also have $I^\nu_T\in L^2(\Omega)$:
\begin{align}
    \E\!\left[\left|I^\nu_T\right|^2\right]=\E\!\left[\left|c\int_0^Te^{\beta\,(t-T)}\,\nu_t\,\dt\right|^2\right]\leq c^2\,T\,\E\!\left[\int_0^T|\nu_t|^2\,\dt\right]<\infty\,.\label{ineq:I_t}
\end{align}
These estimates together with Lemma~\ref{lem:FhA_2} and Cauchy-Schwarz inequality imply $\E\!\left[Y_T\,F_T\right]$, $\E\!\left[Y_T\,I^\nu_T\right]$, $\E\!\left[Q^\nu_T\,I^\nu_T\right]$ are well-defined. Take $\nu^{(n)}\to\nu$ in $\A_2$. Then
\begin{align}
    \left|\E\left[Y_T\,\left(I^{\nu^{(n)}}_T-I^\nu_T\right)\right]\right|&\leq\E\left[|Y_T|^2\right]^{1/2}\,\E\left[\left|I^{\nu^{(n)}}_T-I^\nu_T\right|^2\right]^{1/2}
    \leq c\,\sqrt{T}\,\E\left[|Y_T|\right]^{1/2}\,\E\left[\int_0^T\left|\nu^{(n)}_t-\nu_t\right|^2\,\dt\right]^{1/2}\,,
    \\
    \left|\E\left[F_T\,\left(Q^{\nu^{(n)}}_T-Q^\nu_T\right)\right]\right|&\leq\E\left[|F_T|^2\right]^{1/2}\,\E\left[\left|Q^{\nu^{(n)}}_T-Q^\nu_T\right|^2\right]^{1/2}
    \leq \sqrt{2\,T}\,\E\left[|F_T|\right]^{1/2}\,\E\left[\int_0^T\left|\nu^{(n)}_t-\nu_t\right|^2\,\dt\right]^{1/2}\,,
\end{align}
and, by Minkwoski's inequality
\begin{align}
    & \hspace*{-4em}\left|\E\left[Q^{\nu^{(n)}}_T\,I^{\nu^{(n)}}_T-Q^{\nu}_T\,I^{\nu}_T\right]\right|
    \\ =&\;\left|\E\left[Q^{\nu^{(n)}}_T\,\left(I^{\nu^{(n)}}_T-I^{\nu}_T\right)+\left(Q^{\nu^{(n)}}_T-Q^{\nu}_T\right)\,I^{\nu}_T\right]\right|
    \\
    \leq&\;\E\left[\left|Q^{\nu^{(n)}}_T\right|^2\right]^{1/2}\,\E\left[\left|I^{\nu^{(n)}}_T-I^{\nu}_T\right|^2\right]^{1/2}+\E\left[\left|Q^{\nu^{(n)}}_T-Q^{\nu}_T\right|^2\right]^{1/2}\,\E\left[\left|I^{\nu}_T\right|^2\right]^{1/2}
    \\
    \leq&\left(\E\left[\left|Q^{\nu^{(n)}}_T-Q^{\nu}_T\right|^2\right]^{1/2}+\E\left[\left|Q^{\nu}_T\right|^2\right]^{1/2}\right)\,\E\left[\left|I^{\nu^{(n)}}_T-I^{\nu}_T\right|^2\right]^{1/2}
    \\
    &\quad\ +\E\left[\left|Q^{\nu^{(n)}}_T-Q^{\nu}_T\right|^2\right]^{1/2}\,\E\left[\left|I^{\nu}_T\right|^2\right]^{1/2}
    \\
    \leq&\sqrt{2}\,c\,T\,\E\left[\int_0^T\left|\nu^{(n)}_t-\nu_t\right|^2\,\dt\right]+\,c\,\sqrt{T}\,\E\left[\left|Q^{\nu}_T\right|^2\right]^{1/2}\,\E\left[\int_0^T\left|\nu^{(n)}_t-\nu_t\right|^2\,\dt\right]^{1/2}
    \\
    &\quad\ +\sqrt{2\,T}\,\E\left[\int_0^T\left|\nu^{(n)}_t-\nu_t\right|^2\,\dt\right]^{1/2}\,\E\left[\left|I^{\nu}_T\right|^2\right]^{1/2}\,.
\end{align}
These estimates imply
\begin{equation}
    \mathcal{E}_1=\E\!\left[Y_T\,F_T+Y_T\,I^\nu_T+Q^\nu_T\,I^\nu_T\right]-\langle\mathfrak{I}\nu,\nu\rangle
\end{equation}
is well-defined and continuous in $\nu$, as desired.

\vspace{.3cm}\noindent
\textbf{Step 2.} Next, we show that $\mathcal{E}_2$ is well-defined. This follows immediately from the proof of Proposition~\ref{prop:valuefndfd no transient} in Appendix~\ref{proof:prop:valuefndfd no transient} and the fact that $G,F^2\in\A_2$.

\vspace{.3cm}\noindent
\textbf{Step 3.} Next, we show that $\mathcal{E}_3$ is well-defined and continuous. Since $A\,F$ is in $\A_2$, the only part that needs verification is $G\,F\in\A_2$:
\begin{align}
    \E\!\left[\int_0^T\left|G_t\,F_t\right|^2\,\dt\right]&\leq\E\!\left[\int_0^T\left|G_t\right|^q\,\dt\right]^{\dfrac{2}{q}}\,\E\!\left[\int_0^TF_t^{\dfrac{2\,q}{q-2}}\,\dt\right]^{\dfrac{q-2}{q}}<\infty
\end{align}
by Lemma~\ref{lem:FhA_2}.

\vspace{.3cm}\noindent
\textbf{Step 3.} Next, we show that $\mathcal{E}_1=\mathcal{E}_2+\mathcal{E}_3$. By density, we may assume that $|\nu|\leq N$ for some constant $N$. Then
\begin{equation}
    \left|I^\nu_t\right|=\left|c\,\int_0^te^{\beta(s-t)}\,\nu_s\,\ds\right|\leq c\,T\,N\,.
\end{equation}
By It\^o's formula, we have
\begin{align}
Y_T\,I^\nu_T&=\int_0^T\left(Y_t\,(c\,\nu_t-\beta\,I^\nu_t)+I^\nu_t\,G_t\,F_t\right)\,\dt+\sigma\int_0^TF_t\,\partial_1h(F_t,\kappa)\,I^\nu_t\,\diff W_t\,,
\end{align}
and
\begin{align}
Q^\nu_T\,I^\nu_T&=\int_0^T\left(Q^\nu_t\,(c\,\nu_t-\beta\,I^\nu_t)+I^\nu_t\,\nu_t\right)\,\dt\,,
\end{align}
These together with \eqref{ito YF} imply
\begin{align}\label{e123 no exp}
\begin{split}
Y_T\,F_T+Y_T\,I^\nu_T+Q^\nu_T\,I^\nu_T-\int_{0}^{T}I^{\nu}_{t}\,\nu_{t}\,\dt&=Y_{0}\,F_{0}+\int_{0}^{T}\left\{ \left(G_{t}+\sigma^{2}\,\partial_{1}h(F_{t},\kappa)\right)\,F_{t}^{2}+Y_{t}\,A_{t}\,F_{t}\right\} \,\dt
\\
&\quad+\int_{0}^{T}\left\{ I_{t}^{\nu}\,G_{t}\,F_{t}+\left(Y_{t}+Q_{t}^{\nu}\right)\,\left(c\,\nu_{t}-\beta\,I_{t}^{\nu}\right)\right\} \,\dt
\\
&\quad+\sigma\int_{0}^{T}F_{t}\,\left[Y_{t}+\partial_{1}h(F_{t},\kappa)\,(F_{t}+I_{t}^{\nu})\right]\,\diff W_{t}\,.
\end{split}
\end{align}
Since 
\begin{align}
    \E\left[\int_0^TF_t^2\,\left|Y_t+\partial_1h(F_t,\kappa)\,(F_t+I^\nu_t)\right|^2\,\dt\right]&\lesssim\E\left[\int_0^TF_t^2\,\left(Y_t^2+|\partial_1h(F_t,\kappa)|^2\,\left(F_t^2+|I^\nu_t|^2\right)\right)\,\dt\right]
    \\
    &\lesssim\E\left[\int_0^TF_t^2\,\left(Y_t^2+\left(F_t^{2\,q_\kappa}+F_t^{2\,p_\kappa}\right)\,\left(F_t^2+c^2\,T^2\,N^2\right)\right)\,\dt\right]
    <\infty\,,
\end{align}
the process
\begin{equation}
    \int_0^tF_t\,\left[Y_t+\partial_1h(F_t,\kappa)\,(F_t+I^\nu_t)\right]\,\diff W_t\,,\quad 0\leq t\leq T\,,
\end{equation}
is a martingale, so
\begin{equation}
    \E\left[\int_0^TF_t\,\left[Y_t+\partial_1h(F_t,\kappa)\,(F_t+I^\nu_t)\right]\,\diff W_t\right]=0\,.
\end{equation}
Taking expectation on \eqref{e123 no exp} gives
\begin{align}
\mathcal{E}_1&=\mathcal{E}_2+\langle G\,F,\mathfrak{I}\nu\rangle+\langle Y+Q_0+\mathfrak{Q}\nu,c\,\nu-\beta\,\mathfrak{I}\nu\rangle
\\
&=\mathcal{E}_2+\langle \mathfrak{I}^\top(G\,F)+(c-\beta\,\mathfrak{I}^\top)(Y+Q_0),\nu\rangle+c\,\langle \mathfrak{Q}\nu,\nu\rangle-\beta\,\langle \mathfrak{I}^\top\mathfrak{Q}\nu,\nu\rangle
\\
&=\mathcal{E}_2+\langle \mathfrak{I}^\top(G\,F)+(c-\beta\,\mathfrak{I}^\top)(Y+Q_0),\nu\rangle+\frac{1}{2}\left(c\,\langle \mathfrak{Q}\nu,\nu\rangle+c\,\langle \nu,\mathfrak{Q}\nu\rangle-\beta\,\langle \mathfrak{I}^\top\mathfrak{Q}\nu,\nu\rangle-\beta\,\langle \nu,\mathfrak{I}^\top\mathfrak{Q}\nu\rangle\right)
\\
&=\mathcal{E}_2+\langle \mathfrak{I}^\top(G\,F)+(c-\beta\,\mathfrak{I}^\top)(Y+Q_0),\nu\rangle+\frac{1}{2}\left(c\,\langle \mathfrak{Q}\nu,\nu\rangle+c\,\langle \mathfrak{Q}^\top\nu,\nu\rangle-\beta\,\langle \mathfrak{I}^\top\mathfrak{Q}\nu,\nu\rangle-\beta\,\langle \mathfrak{Q}^\top\mathfrak{I}\nu,\nu\rangle\right)
\\
&=\mathcal{E}_2+\mathcal{E}_3\,.
\end{align}

\subsubsection{Proof of Proposition \ref{proposition strict concavity}}\phantomsection\label{proof:proposition strict concavity}

\vspace{.3cm}\noindent
First, we show $J$ is strictly concave. Take $\nu,\zeta\in\A_2$ and $\rho\in(0,1)$. We may reuse the proof of Proposition~\ref{proposition strict concavity no transient} in Appendix~\ref{proof:proposition strict concavity no transient} to write
\begin{equation}\label{convex comb J}
    J[\rho\,\nu+(1-\rho)\,\zeta]=\frac{1}{2}\,\rho\,(1-\rho)\,\langle\Lambda(\nu-\zeta),\nu-\zeta\rangle+\rho\,J[\nu]+(1-\rho)\,J[\zeta]\,.
\end{equation}
Consider
\begin{align}
    \langle\Lambda\nu,\nu\rangle&=\left\langle \left(2\,\eta+\beta\,(\mathfrak{I}^\top\mathfrak{Q}+\mathfrak{Q}^\top\mathfrak{I})-c\,(\mathfrak{Q}+\mathfrak{Q}^\top)+\phi\,\mathfrak{Q}^\top\mathfrak{Q}\right)\nu,\nu\right\rangle
    \\
    &=2\,\eta\,\Vert\nu\Vert^2-2\,c\,\langle \mathfrak{Q}\nu,\nu\rangle+\phi\,\Vert \mathfrak{Q}\nu\Vert^2+2\,\beta\langle \mathfrak{Q}\nu,\mathfrak{I}\nu\rangle
\end{align}
By integration by parts, we have
\begin{align}
    \langle\mathfrak{Q}\nu,\nu\rangle=\E\!\left[\int_0^T\int_0^t\nu_s\,\ds\,\nu_t\,\dt\right]&=\E\!\left[\left(\int_0^T\nu_t\,\dt\right)^2-\int_0^T\int_0^t\nu_s\,\ds\,\nu_t\,\dt\right]\,,
\end{align}
so
\begin{equation}
    \langle\mathfrak{Q}\nu,\nu\rangle=\frac{1}{2}\,\E\!\left[\left(\int_0^T\nu_t\,\dt\right)^2\right]\geq 0\,.\label{ineq:Qnunu>=0}
\end{equation}
Let
\begin{equation}
    \tilde{\mathfrak{I}}_t\coloneqq\int_0^te^{\beta\,(s-t)}\,\nu_s\,\ds\,.
\end{equation}
The dynamics \eqref{eq: dyn I} implies
\begin{align*}
    c\,\tilde{\mathfrak{I}}_t=I^\nu_t=c\int_0^t\nu_s\,\diff s-\beta\int_0^tI^\nu_s\,\diff s=c\,(\mathfrak{Q}\nu)_t-\beta\int_0^tc\,\tilde{\mathfrak{I}}_s\,\diff s\,,
\end{align*}
so
\begin{align}
    c\,(\mathfrak{Q}\nu)_t=c\,\left(\tilde{\mathfrak{I}}_t+\beta\int_0^t\tilde{\mathfrak{I}}_s\,\ds\right)\,.\label{eqn: Q frakI}
\end{align}
Therefore,
\begin{align*}
    \langle\mathfrak{Q}\nu,\mathfrak{I}\nu\rangle=\E\!\left[\int_0^T(\mathfrak{Q}\nu)_t\,(\mathfrak{I}\nu)_t\,\dt\right]&=\E\!\left[\int_0^Tc\,(\mathfrak{Q}\nu)_t\,\tilde{\mathfrak{I}}_t\,\dt\right]
    \\
    &=c\,\E\!\left[\int_0^T\left(\tilde{\mathfrak{I}}_t+\beta\int_0^t\tilde{\mathfrak{I}}_s\,\ds\right)\,\tilde{\mathfrak{I}}_t\,\dt\right]
    \\&=c\,\E\!\left[\int_0^T\tilde{\mathfrak{I}}_t^2\,\dt+\beta\int_0^T\tilde{\mathfrak{I}}_t\int_0^t\tilde{\mathfrak{I}}_s\,\diff s\,\dt\right]
    \\
    &=c\,\left(\Vert\tilde{\mathfrak{I}}\Vert^2+\beta\,\left\langle\mathfrak{Q}\tilde{\mathfrak{I}},\tilde{\mathfrak{I}}\right\rangle\right)
    \;\geq 0
\end{align*}
due to \eqref{ineq:Qnunu>=0}. It follows that (recall Assumption~\ref{assume:2})
\begin{align}
    \langle\Lambda\nu,\nu\rangle&=2\,\eta\,\Vert\nu\Vert^2-2\,c\,\langle \mathfrak{Q}\nu,\nu\rangle+\phi\,\Vert \mathfrak{Q}\nu\Vert^2+2\,\beta\langle \mathfrak{Q}\nu,\mathfrak{I}\nu\rangle
    \\
    &\geq2\,\eta\,\Vert\nu\Vert^2-2\,\sqrt{2\,\eta\,\phi}\,\langle \mathfrak{Q}\nu,\nu\rangle+\phi\,\Vert \mathfrak{Q}\nu\Vert^2
    =\left\Vert\sqrt{2\,\eta}\,\nu-\sqrt{\phi}\,\mathfrak{Q}\nu\right\Vert^2\,.
\end{align}
Consider the bounded linear operator $\mathfrak{V}:L^2[0,T]\to L^2[0,T]$ defined by
\begin{align}
    (\mathfrak{V}f)(t)=\sqrt{2\,\eta}f(t)-\sqrt{\phi}\int_0^t f(s)\,\ds\,,
\end{align}
whose inverse is also a bounded linear operator on $L^2[0,T]$ and is given by
\begin{align}
    (\mathfrak{V}^{-1}f)(t)=\frac{1}{\sqrt{2\,\eta}}\,f(t)+\frac{\sqrt{\phi}}{2\,\eta}\,\int_0^te^{\sqrt{\dfrac{\phi}{2\,\eta}}\,(t-s)}\,f(s)\,\ds\,.
\end{align}
Since $\nu(\omega)\in L^2[0,T]$ for $\mathbb P$-a.e. $\omega$, we have
\begin{align}
    \Vert\nu\Vert^2=\int_\Omega\Vert\nu(\omega)\Vert_{L^2[0,T]}^2\,\diff\mathbb P(\omega)&=\int_\Omega\left\Vert \mathfrak{V}^{-1}\mathfrak{V}(\nu(\omega))\right\Vert_{L^2[0,T]}^2\,\diff\mathbb P(\omega)
    \\
    &\leq\int_\Omega\left\Vert \mathfrak{V}^{-1}\right\Vert_{\operatorname{op}}^2\,\left\Vert \mathfrak{V}(\nu(\omega))\right\Vert_{L^2[0,T]}^2\,\diff\mathbb P(\omega)
    =\left\Vert \mathfrak{V}^{-1}\right\Vert_{\operatorname{op}}^2\,\left\Vert\sqrt{2\,\eta}\,\nu-\sqrt{\phi}\,\mathfrak{Q}\nu\right\Vert^2\,.
\end{align}
Therefore,
\begin{equation}\label{coercive Lambda}
    \langle\Lambda\nu,\nu\rangle\geq \left\Vert \mathfrak{V}^{-1}\right\Vert_{\operatorname{op}}^{-2}\,\Vert\nu\Vert^2\,.
\end{equation}
\eqref{coercive Lambda} and \eqref{convex comb J} imply
\begin{align}
    J[\rho\,\nu+(1-\rho)\,\zeta]\geq\frac{1}{2}\,\rho\,(1-\rho)\,\left\Vert \mathfrak{V}^{-1}\right\Vert_{\operatorname{op}}^{-2}\,\Vert\nu-\zeta\Vert^2+\rho\,J[\nu]+(1-\rho)\,J[\zeta]
    \;\geq\rho\,J[\nu]+(1-\rho)\,J[\zeta]\,,
\end{align}
with equality if and only if $\nu=\zeta$. Hence, $J$ is strictly concave.

Next, we show differentiability of $J$. The same calculation in Appendix~\ref{proof:proposition strict concavity no transient} gives 
\begin{equation}\label{eqn:DJHilbertelement}
    \D J[\nu]=-\Lambda\nu+b\,.
\end{equation}
Similarly as \eqref{derivation Q transp}, $\mathfrak{I}^\top$ is given by
    \begin{equation}
        (\mathfrak{I}^\top\zeta)_t=c\,\E\!\left[\left.\int_t^Te^{\beta\,(t-s)}\,\zeta_s\,\ds\,\right|\,\F_t\right]\,.
    \end{equation}
    It follows from \eqref{eqn:DJHilbertelement} that
    \begin{align}
        \D J[\nu]_t&=-2\,\eta\,\nu+c\,(Y+Q_0+\mathfrak{Q}\nu)+\mathfrak{Q}^T(A\,F-\beta\,\mathfrak{I}\nu+c\,\nu-\phi\,(Y+Q_0+\mathfrak{Q}\nu))
        \\
        &\quad\ +\mathfrak{I}^\top(G\,F-\beta\,(Y+Q_0+\mathfrak{Q}\nu))
        \\&=-2\,\eta\,\nu_t+c\,\left(Y_t+Q^{\nu}_t\right)+\E\left[\left.\int_t^T\left(A_s\,F_s+c\,\nu_s-\beta\,I^\nu_s-\phi\,\left(Y_s+Q^{\nu}_s\right)\right)\,\ds\,\right|\,\F_t\right]
        \\
        &\quad+c\,e^{t\,\beta}\,\E\left[\left.\int_t^Te^{-s\,\beta}\,\left(G_s\,F_s-\beta\,\left(Y_s+Q^\nu_s\right)\right)\,\ds\,\right|\,\F_t\right]\,.
    \end{align}
\qed

\subsubsection{Proof of Theorem \ref{theorem optimal speed}}\phantomsection\label{proof:theorem optimal speed}

\vspace{.3cm}\noindent
Suppose $\D J[\nu^\star]=0$ for some $\nu^\star\in\A_2$. Then by Proposition~\ref{proposition strict concavity} we have
\begin{align}
    2\,\eta\,\nu^\star_{t}&=\E\!\left[\left.c\,\left(Y_{t}+Q_{t}^{\nu^\star}\right)+\int_{t}^{T}\left(A_{s}\,F_s+c\,\nu^\star_{s}-\beta \,I_{s}^{\nu^\star}-\phi\left(Y_{s}+Q_{s}^{\nu^\star}\right)\right)\,\ds\,\right|\,\F_t\right]
        \\
        &\quad\ +c\,e^{t\,\beta}\,\E\!\left[\left.\int_{t}^{T}e^{-s\,\beta}\left(G_{s}\,F_s-\beta\,\left(Y_{s}+Q_{s}^{\nu^\star}\right)\right)\,\ds\,\right|\,\F_t\right]\,.
\end{align}
Comparing with Appendix~\ref{proof:theorem optimal speed no transient}, the right-hand side of above identity gets extra term
\begin{align}
    &\E\!\left[\left.c\,\left(Y_{t}+Q_{t}^{\nu^\star}\right)+\int_{t}^{T}\left(c\,\nu^\star_{s}-\beta \,I_{s}^{\nu^\star}\right)\,\ds\,\right|\,\F_t\right]+c\,e^{t\,\beta}\,\E\!\left[\left.\int_{t}^{T}e^{-s\,\beta}\left(G_{s}\,F_s-\beta\,\left(Y_{s}+Q_{s}^{\nu^\star}\right)\right)\,\ds\,\right|\,\F_t\right]
    \\
    &=\E\!\left[\left.c\,\left(Y_{T}+Q_{T}^{\nu^\star}\right)+\int_{0}^{T}\left(-c\,G_s\,F_s-\beta \,I_{s}^{\nu^\star}\right)\,\ds\,\right|\,\F_t\right]-\int_{0}^{t}\left(-c\,G_s\,F_s-\beta \,I_{s}^{\nu^\star}\right)\,\ds
        \\
        & \quad \ +c\,e^{t\,\beta}\,\E\!\left[\left.\int_{0}^{T}e^{-s\,\beta}\left(G_{s}\,F_s-\beta\,\left(Y_{s}+Q_{s}^{\nu^\star}\right)\right)\,\ds\,\right|\,\F_t\right]-c\,e^{t\,\beta}\int_{0}^{t}e^{-s\,\beta}\left(G_{s}\,F_s-\beta\,\left(Y_{s}+Q_{s}^{\nu^\star}\right)\right)\,\ds
        \\ & 
        \quad \  -c\,\sigma\,\E\!\left[\left.\int_{t}^{T}\partial_1h(F_s,\kappa)\,F_s\,\diff W_s\,\right|\,\F_t\right]\,.
\end{align}
    Since
    \begin{align}
        \E\left[\int_0^T\left|\partial_1h(F_t,\kappa)\right|^2\,F_t^2\,\dt\right]\lesssim\E\left[\int_0^T\left(F_t^{2\,q_\kappa+2}+F_t^{2\,p_\kappa+2}\right)\,\dt\right]<\infty\,,
    \end{align}
    the process
    \begin{equation}\label{mart:dhFdW}
        \int_{0}^{t}\partial_1h(F_s,\kappa)\,F_s\,\diff W_s\,,\quad 0\leq t\leq T\,,
    \end{equation}
    is a martingale, so
    \begin{equation}
        \E\left[\left.\int_{t}^{T}\partial_1h(F_s,\kappa)\,F_s\,\diff W_s\,\right|\,\F_t\right]=0\,.
    \end{equation}
    Define process $\tilde{N}$ by
    \begin{align*}
        \tilde{N}_t=\E\left[\left.\int_{0}^{T}e^{-s\,\beta}\left(G_{s}\,F_s-\beta\,\left(Y_{s}+Q_{s}^{\nu^\star}\right)\right)\,\ds\,\right|\,\F_t\right]\,.
    \end{align*}
    Then $\tilde{N}$ is a martingale with
    \begin{align}
        \E\left[|\tilde{N}_T|^2\right]&\leq\E\left[\left|\int_{0}^{T}e^{-s\,\beta}\left(G_{s}\,F_s-\beta\,\left(Y_{s}+Q_{s}^{\nu^\star}\right)\right)\,\ds\right|^2\right]
        \\
        &\lesssim\E\left[\int_0^T\left|G_s\,F_s\right|^2\,\ds\right]+\E\left[\int_0^T\left|Y_s\right|^2\,\ds\right]+\E\left[\int_0^T\left|Q^{\nu^\star}_s\right|^2\,\ds\right]
        \\
        &\lesssim\E\left[\int_0^T\left|G_s\right|^q\,\ds\right]^{2/q}\,\E\left[\int_0^T\left|F_s\right|^r\,\ds\right]^{2/r}+\E\left[\int_0^T\left|Y_s\right|^2\,\ds\right]+\E\left[\int_0^T\left|Q^{\nu^\star}_s\right|^2\,\ds\right]<\infty
    \end{align}
    for some $q\in(2,p)$ and $r>2$ such that $\dfrac{1}{q}+\dfrac{1}{r}=\dfrac{1}{2}$ due to Lemma~\ref{lem:FhA_2}. Define process $Z$ by
    \begin{align*}
        Z_t=e^{t\,\beta}\,\left(\tilde{N}_t-\int_{0}^{t}e^{-s\,\beta}\left(G_{s}\,F_s-\beta\,\left(Y_{s}+Q_{s}^{\nu^\star}\right)\right)\,\ds\right)\,.
    \end{align*}
    Then $Z_T=0$, and generalised It\^o's formula (note $\tilde{N}$ is c\`adl\`ag but not necessarily continuous) implies
    \begin{align} 
        -Z_t&=\int_t^T\beta\,e^{s\,\beta}\left(\tilde{N}_{s-}-\int_{0}^{s}e^{-u\,\beta}\left(G_{u}\,F_u-\beta\,\left(Y_{u}+Q_{u}^{\nu^\star}\right)\right)\,\du\right)\,\diff s+\int_t^Te^{s\,\beta}\,\diff\tilde{N}_s
        \\
        &\quad\ -\int_t^T\left(G_{s}\,F_s-\beta\,\left(Y_{s}+Q_{s}^{\nu^\star}\right)\right)\,\ds+\sum_{t<s\leq T}\left[e^{s\,\beta}\,\tilde{N}_s-e^{s\,\beta}\,\tilde{N}_{s-}-e^{s\,\beta}\,\Delta\tilde{N}_s\right]
        \\
        &=\int_t^T\beta\,e^{s\,\beta}\left(\tilde{N}_{s}-\int_{0}^{s}e^{-u\,\beta}\left(G_{u}\,F_u-\beta\,\left(Y_{u}+Q_{u}^{\nu^\star}\right)\right)\,\du\right)\,\diff s+\int_t^Te^{s\,\beta}\,\diff\tilde{N}_s
        \\
        &\quad\ -\int_t^T\left(G_{s}\,F_s-\beta\,\left(Y_{s}+Q_{s}^{\nu^\star}\right)\right)\,\ds
        \\
        &=\int_t^T\left(\beta\,\left(Z_s+Y_s+Q^{\nu^\star}_s\right)-G_s\,F_s\right)\,\diff s+\int_t^Te^{s\,\beta}\,\diff\tilde{N}_s\,,\label{dyn Z}
    \end{align}
    where the second equality is because a c\`adl\`ag path has at most countably many jumps, which form a Lebesgue measure zero set. Define process $N$ by
    \begin{align*}
        N_t=\int_0^te^{s\,\beta}\,\diff\tilde{N}_s\,,\quad 0\leq t\leq T\,.
    \end{align*}
    Since
    \begin{align}
        \E\left[\int_0^Te^{2\,s\,\eta}\,\diff[\tilde{N}]_s\right]\leq e^{2\,T\,\eta}\,\E\left[[\tilde{N}]_T\right]\leq e^{2\,T\,\eta}\,\E\left[|\tilde{N}_T|^2\right]<\infty\,,
    \end{align}
    $N$ is a martingale with $N_T\in L^2(\Omega)$.     Moreover, the process $M$, obtained by adding
    \begin{align*}
        \E\!\left[\left.c\,\left(Y_T+Q^{\nu^\star}_T\right)+\int_0^T\left(-c\,G_s\,F_s-\beta\,I^{\nu^\star}_s\right)\,\diff s\,\right|\,\F_t\right]+c\,N_t
    \end{align*}
    to \eqref{def mart m}, is also a martingale with $M_T\in L^2(\Omega)$. Consequently, the right-hand side of \eqref{2 eta nu no transient} gets extra term
    \begin{equation}
        c\,\left(Y_T+Q^{\nu^\star}_T\right)-\int_t^T\left(\beta\,I^{\nu^\star}_s+c\,\beta\,\left(Y_s+Q^{\nu^\star}_s\right)+c\,\beta\,Z_s\right)\,\diff s\,.
    \end{equation}
Thus $\nu^\star$ satisfies FBSDE \eqref{eq:FBSDE}, which is obtained from FBSDE \eqref{eq:FBSDE no transient} by adding $\left(\beta\,I_t+c\,\beta\,\left(Y_s+Q_t+Z_t\right)\right)\,\dt$ to $2\,\eta\,\diff\nu^\star_t$, adding $c\,\left(Y_T+Q_T\right)$ to $2\,\eta\,\nu^\star_T$, and including the dynamics of $Z$ (which is \eqref{dyn Z}) and $I$.

\vspace{.3cm}\noindent
Conversely, assume $\nu^\star\in\A_2$ satisfies FBSDE \eqref{eq:FBSDE} for some martingales $M$ and $N$ such that $M_T,N_T\in L^2(\Omega)$. By integrating  $\nu^\star$ and $Z$ and using the terminal conditions, we may write
\begin{equation}
    2\,\eta\,\nu^\star_t=c\,\left(Y_T+Q^{\nu^\star}_T\right)+\int_t^T\left(A_s\,F_s-\beta\,I^{\nu^\star}_s-(\phi+c\,\beta)\,\left(Y_s+Q^{\nu^\star}_s\right)-c\,\beta\,Z_s\right)\,\ds-M_T+M_t\label{2 eta nu 2}
\end{equation}
and
\begin{align}
    Z_t=\int_t^T\left(-\beta\,\left(Z_s+Y_s+Q^{\nu^\star}_s\right)+G_s\,F_s\right)\,\ds-N_T+N_t
\end{align}
Comparing with \eqref{2 eta nu no transient 2}, the right-hand side of \eqref{2 eta nu 2} gets extra term
\begin{align}
    &c\,\left(Y_T+Q^{\nu^\star}_T\right)+\int_t^T\left(-\beta\,I^{\nu^\star}_s-c\,\beta\,\left(Y_s+Q^{\nu^\star}_s\right)-c\,\beta\,Z_s\right)\,\ds
    \\
    &=c\,\left(Y_t+Q^{\nu^\star}_t\right)+\int_t^Tc\,G_s\,F_s\,\ds+\int_t^Tc\,\sigma\,\partial_1h(F_s,\kappa)\,F_s\,\diff W_s+\int_t^Tc\,\nu^\star_s\,\ds
    \\
    &\quad\ +\int_t^T\left(-\beta\,I^{\nu^\star}_s-c\,\beta\,\left(Y_s+Q^{\nu^\star}_s\right)\right)\,\ds+c\,Z_t+\int_t^T\left(c\,\beta\,\left(Y_s+Q^{\nu^\star}_s\right)-c\,G_s\,F_s\right)\,\ds+c\,N_T-c\,N_t
    \\
    &=c\,\left(Y_t+Q^{\nu^\star}_t\right)+\int_t^T\left(c\,\nu^\star_s-\beta\,I^{\nu^\star}_s\right)\,\ds+c\,Z_t+\int_t^Tc\,\sigma\,\partial_1h(F_s,\kappa)\,F_s\,\diff W_s+c\,N_T-c\,N_t\,.\label{eqn:2knustar}
\end{align}
To solve for $Z$, we use generalised It\^o's formula and the dynamics and terminal condition of $Z$ to write
\begin{align}
    Z_t& =e^{t\,\beta}\,e^{-t\,\beta}\,Z_t
    \\ 
    &=e^{t\,\beta}\,\left(\int_t^T\beta\,e^{-s\,\beta}\,Z_s\,\ds-\int_t^Te^{-s\,\beta}\,\left(\beta\,\left(Z_s+Y_s+Q^{\nu^\star}_s\right)-G_s\,F_s\right)\,\ds-\int_t^Te^{-s\,\beta}\,\diff N_s\right)
    \\
    &=-e^{t\,\beta}\,\left(\int_t^Te^{-s\,\beta}\,\left(\beta\,\left(Y_s+Q^{\nu^\star}_s\right)-G_s\,F_s\right)\,\ds+\int_t^Te^{-s\,\beta}\,\diff N_s\right)\,.
\end{align}
Since
    \begin{align}
        \E\!\left[\int_0^Te^{-2\,s\,\eta}\,\diff[ N]_s\right]\leq\E\left[[ N]_T\right]\leq \E\left[|N_T|^2\right]<\infty\,,
    \end{align}
    the process
    \begin{equation}
        \int_0^te^{-s\,\beta}\,\diff N_s\,,\quad 0\leq t\leq T\,,
    \end{equation}
    is a martingale. Therefore,
    \begin{equation}
        \E[Z_t\,|\,\F_t]=e^{t\,\beta}\,\E\!\left[\left.\int_t^Te^{-s\,\beta}\,\left(G_s\,F_s-\beta\,\left(Y_s+Q^{\nu^\star}_s\right)\right)\,\ds\,\right|\,\F_t\right]\,.
    \end{equation}
    Plugging this into \eqref{eqn:2knustar} and taking conditional expectation gives (recall that the process in \eqref{mart:dhFdW} is a martingale)
\begin{equation}
    c\,\left(Y_t+Q^{\nu^\star}_t\right)+\E\!\left[\left.\int_t^T\left(c\,\nu^\star_s-\beta\,I^{\nu^\star}_s\right)\,\ds\,\right|\,\F_t\right]+c\,e^{t\,\beta}\,\E\!\left[\left.\int_t^Te^{-s\,\beta}\,\left(G_s\,F_s-\beta\,\left(Y_s+Q^{\nu^\star}_s\right)\right)\,\ds\,\right|\,\F_t\right]
\end{equation}
Adding above term to the right-hand side of \eqref{2 eta nu no transient 3} gives
    \begin{align}
        2\,\eta\,\nu^\star_t=c\,\left(Y_t+Q^{\nu^\star}_t\right)&+\E\!\left[\left.\int_t^T\left(A_s\,F_s+c\,\nu^\star_s-\beta\,I^{\nu^\star}_s-\phi\,\left(Y_s+Q^{\nu^\star}_s\right)\right)\,\ds\,\right|\,\F_t\right]
        \\
        & +c\,e^{t\,\beta}\,\E\!\left[\left.\int_t^Te^{-s\,\beta}\,\left(G_s\,F_s-\beta\,\left(Y_s+Q^{\nu^\star}_s\right)\right)\,\ds\,\right|\,\F_t\right]\,,
    \end{align}
    that is, $\D J[\nu^\star]_t=0$ by Proposition~\ref{proposition strict concavity}.

\subsubsection{Proof of Proposition \ref{prop:FBSDEtoDRE}}\phantomsection\label{proof:prop:FBSDEtoDRE}
First, we show we may construct a solution of the FBSDE from a solution of the DRE. Suppose $P$ is a solution to the DRE \eqref{eq:dreP} and the processes $\ell$, $\Phi$, $\Psi$ are defined as stated in the proposition. Let us differentiate these processes. For $\Phi$, we have
\begin{align}
    \diff\Phi_t&=\left(B_{12}\,P(t)+B_{11}\right)\,\Phi_t\,\dt+e^{\int_0^t\left(B_{12}\,P(u)+B_{11}\right)\,\du}\,e^{-\int_0^t\left(B_{12}\,P(u)+B_{11}\right)\,\du}\,B_{12}\,\ell_t\,\dt
    \\
    &=\left(B_{12}\,\left(P(t)\,\Phi_t+\ell_t\right)+B_{11}\,\Phi_t\right)\,\dt
    \\
    &=\left(B_{12}\,\Psi_t+B_{11}\,\Phi_t\right)\,\dt\,.
\end{align}
For $\ell$, we have
\begin{align}
    \ell_t&=e^{-\int_0^t\left(P(u)\,B_{12}-B_{22}\right)\,\du}\,\E\!\left[\left.L-\int_t^Te^{\int_0^s\left(P(u)\,B_{12}-B_{22}\right)\,\du}\,b_s\,\ds\,\right|\,\F_t\right]
    \\
    &=e^{-\int_0^t\left(P(u)\,B_{12}-B_{22}\right)\,\du}\,\left(\E\!\left[\left.L-\int_0^Te^{\int_0^s\left(P(u)\,B_{12}-B_{22}\right)\,\du}\,b_s\,\ds\,\right|\,\F_t\right]+\int_0^te^{\int_0^s\left(P(u)\,B_{12}-B_{22}\right)\,\du}\,b_s\,\ds\right)\,.
\end{align}
Let
\begin{equation}
    \tilde{\mathcal{M}}_t=\E\!\left[\left.L-\int_0^Te^{\int_0^s\left(P(u)\,B_{12}-B_{22}\right)\,\du}\,b_s\,\ds\,\right|\,\F_t\right]\,,
\end{equation}
then $\tilde{M}$ is an $\mathbb R^2$-valued martingale. By Lemma~\ref{lem:FhA_2}, we have
\begin{align}
\begin{split}
    \hspace*{-1em}\E\!\left[\int_0^T|b_t|^2\,\dt\right]
    &\lesssim\E\!\left[\int_0^T\left(|A_t\,F_t|^2+|Y_t|^2+|G_t\,F_t|^2\right)\,\dt\right]
    \\
    &\leq\E\!\left[\int_0^T|A_t|^p\,\dt\right]^{\dfrac{2}{p}}\,\E\!\left[\int_0^T|F_t|^{\dfrac{2\,p}{p-2}}\,\dt\right]^{\dfrac{p-2}{p}}
    \\
    &\quad +\E\!\left[\int_0^T|F_t|^2\,\dt\right]+\E\!\left[\int_0^T|G_t|^q\,\dt\right]^{\dfrac{2}{q}}\,\E\!\left[\int_0^T|F_t|^{\dfrac{2\,q}{q-2}}\,\dt\right]^{\dfrac{q-2}{q}}
    <\infty
\end{split}
\label{ineq:bbound}
\end{align}
for some $q\in(2,p)$, and thus
\begin{align}
    \E\!\left[\left|\tilde{\mathcal{M}}_T\right|^2\right]&\leq\E\!\left[\left|L-\int_0^Te^{\int_0^s\left(P(u)\,B_{12}-B_{22}\right)\,\du}\,b_s\,\ds\right|^2\right]\lesssim\E[Y_T^2]+\E\!\left[\int_0^T|b_s|^2\,\ds\right]<\infty\,.\label{ineq:tildecalMbound}
\end{align}
By generalised It\^o's formula,
\begin{align}
    \diff\ell_t&=-\left(P(t)\,B_{12}-B_{22}\right)\,\ell_t\,\dt+e^{-\int_0^t\left(P(u)\,B_{12}-B_{22}\right)\,\du}\,\left(\diff\tilde{\mathcal{M}}_t+e^{\int_0^t\left(P(u)\,B_{12}-B_{22}\right)\,\du}\,b_t\,\dt\right)
    \\
    &=\left(\left(-P(t)\,B_{12}+B_{22}\right)\,\ell_t+b_t\right)\,\dt+e^{-\int_0^t\left(P(u)\,B_{12}-B_{22}\right)\,\du}\,\diff\tilde{\mathcal{M}}_t\,.
\end{align}
Let
\begin{equation}
    \mathcal{M}_t=\int_0^te^{-\int_0^s\left(P(u)\,B_{12}-B_{22}\right)\,\du}\,\diff\tilde{\mathcal{M}}_s\,.
\end{equation}
Since the integrand is deterministic and differentiable and because of \eqref{ineq:tildecalMbound}, $\mathcal{M}$ is an $\mathbb R^2$-valued martingale with $\E[|\mathcal{M}_T|]^2<\infty$. For $\Psi$, we have
\begin{align*}
    \diff\Psi_t&=P^\prime(t)\,\Phi_t\,\dt+P(t)\,\diff\Phi_t+\diff\ell_t
    \\
    &=P^\prime(t)\,\Phi_t\,\dt+P(t)\,(B_{11}\,\Phi_t+B_{12}\,(P(t)\,\Phi_t+\ell_t))\,\dt+\diff\ell_t
    \\
    &=(P^\prime(t)+P(t)\,B_{11}+P(t)\,B_{12}\,P(t))\,\Phi_t\,\dt+P(t)\,B_{12}\,\ell_t\,\dt+\diff\ell_t
    \\
    &=(B_{21}+B_{22}\,P(t))\,\Phi_t\,\dt+P(t)\,B_{12}\,\ell_t\,\dt+((-P(t)\,B_{12}+B_{22})\,\ell_t+b_t)\,\dt+\diff\mathcal{M}_t
    \\
    &=(B_{21}\,\Phi_t+B_{22}\,(P(t)\,\Phi_t+\ell_t)+b_t)\,\dt+\diff\mathcal{M}_t
    \\
    &=(B_{21}\,\Phi_t+B_{22}\,\Psi_t+b_t)\,\dt+\diff\mathcal{M}_t\,.
\end{align*}
Thus we obtain the FBSDE
\begin{align*}
    \left\{
    \begin{array}{rlrl}
        \diff\Phi_t & = \left(B_{11}\,\Phi_t+B_{12}\,\Psi_t\right)\,\dt\,, & \Phi_0 & = K
        \\~\\
        \diff\Psi_t & = \left(B_{21}\,\Phi_t+B_{22}\,\Psi_t+b_t\right)\,\dt + \diff\mathcal{M}_t\,, & \Psi_T & = G\,\Phi_T + L
    \end{array}
    \right.\,,
\end{align*}
which is precisely FBSDE \eqref{eq:FBSDE} written in vectorial form, provided we identify
\[
\Psi_t=\begin{pmatrix}\nu_t^\star \\ Z_t\end{pmatrix}\,,\quad
\Phi_t=\begin{pmatrix}I_t \\ Q_t\end{pmatrix}\,,\quad\mathcal{M}_t=\begin{pmatrix}\dfrac{1}{2\,\eta}\,M_t \\ N_t\end{pmatrix}\,.
\]

\vspace{.3cm}\noindent
Moreover, due to \eqref{ineq:bbound} and \eqref{ineq:tildecalMbound}, we obtain the three inequalities
\begin{align}
    \E\!\left[\int_0^T|\ell_t|^2\,\dt\right]&=\E\!\left[\int_0^T\left|e^{-\int_0^t\left(P(u)\,B_{12}-B_{22}\right)\,\du}\,\left(\tilde{\mathcal{M}}_t+\int_0^te^{\int_0^s\left(P(u)\,B_{12}-B_{22}\right)\,\du}\,b_s\,\ds\right)\right|^2\,\dt\right]
    \\
    &\lesssim\E\!\left[\int_0^T\left(\left|\tilde{\mathcal{M}}_t\right|^2+\int_0^t|b_s|^2\,\ds\right)\,\dt\right]
    \\
    &\lesssim\E\!\left[\left|\tilde{\mathcal{M}_T}\right|^2\right]+\E\!\left[\int_0^T|b_t|^2\,\dt\right]
    <\infty\,,
\end{align}
\begin{align}
    \E\!\left[\int_0^T|\Phi_t|^2\,\dt\right]&=\E\!\left[\int_0^T\left|e^{\int_0^t\left(B_{12}\,P(u)+B_{11}\right)\,\du}\,\left(K+\int_0^te^{-\int_0^s\left(B_{12}\,P(u)+B_{11}\right)\,\du}\,B_{12}\,\ell_s\,\ds\right)\right|^2\,\dt\right]
    \\
    &\lesssim Q_0^2+\E\!\left[\int_0^T\int_0^t|\ell_s|^2\,\ds\,\dt\right]
    \\
    &\lesssim Q_0^2+\E\!\left[\int_0^T|\ell_t|^2\,\dt\right] <\infty\,,
\end{align}
and
\begin{equation}
    \E\!\left[\int_0^T|\Psi_t|^2\,\dt\right]=\E\!\left[\int_0^T\left|P(t)\,\Phi_t+\ell_t\right|^2\,\dt\right]
    \lesssim\E\!\left[\int_0^T|\Phi_t|^2\,\dt\right]+\E\!\left[\int_0^T|\ell_t|^2\,\dt\right]
    <\infty\,,
\end{equation}
which implies $\nu^\star\in\A_2$.

\vspace{.3cm}\noindent
Next, we show the DRE \eqref{eq:dreP} admits a unique solution under Assumption \ref{assume:2}, that is, $c^2<2\,\eta\,\phi$. Here we only consider the case where $c>0$. The $c=0$ case is addressed in Proposition \ref{prop:notrasient}, where we derive an explicit solution of \eqref{eq:dreP}. Let
\begin{equation}
    z=-\frac{1}{2}\left(\frac{c^2}{2\,\beta}+\sqrt{\frac{\phi\,c^2\,\eta}{2\,\beta^2}}\right)<0
\end{equation}
and
\begin{equation}
    w = \frac{2\,\beta\,z^2}{c\,\eta}\,.
\end{equation}
Since
\begin{equation}
    \sqrt{\frac{\phi\,c^2\,\eta}{2\,\beta^2}}>\sqrt{\frac{c^4}{4\,\beta^2}}=\frac{c^2}{2\,\beta}\,,
\end{equation}
we have
\begin{equation}
    -\sqrt{\frac{\phi\,c^2\,\eta}{2\,\beta^2}}<z<-\frac{c^2}{2\,\beta}\,.\label{range:z}
\end{equation}
Let
\begin{equation}
    C = \begin{pmatrix} 1 & 0 \\ 0 & w \end{pmatrix}\,,
    \quad
    D = \begin{pmatrix} 0 & \dfrac{z}{2\,\eta} \\ z & -\dfrac{c\,z}{2\,\eta} \end{pmatrix}\,,
\end{equation}
and
\begin{equation}
    \mathcal{L}=\begin{pmatrix}
        C\,B_{11}+D\,B_{21} & C\,B_{12}+B_{11}^\top\,D+D\,B_{22}
        \\
        0 & B_{12}^\top\,D
    \end{pmatrix}\,.
\end{equation}
Consider
\begin{equation}
    \mathcal{L}+\mathcal{L}^\top=\begin{pmatrix}
        \mathcal{K}_{11} & \mathcal{K}_{12}
        \\
        \mathcal{K}_{12}^\top & \mathcal{K}_{22}
    \end{pmatrix}\,,
\end{equation}
where
\begin{align}
    \mathcal{K}_{11}&=C\,B_{11}+(C\,B_{11})^\top+D\,B_{21}+(D\,B_{21})^\top
    =\begin{pmatrix}
            -2\,\beta & \dfrac{\beta\,z}{k} \\
            \dfrac{\beta\,z}{k} & \dfrac{\phi\,z}{k}
        \end{pmatrix}\,,
        \\
    \mathcal{K}_{12}&=C\,B_{12}+B_{11}^\top\,D+D\,B_{22}=\begin{pmatrix}
            c & 0 \\
            w & 0
        \end{pmatrix}\,, 
    \\
    \mathcal{K}_{22}&=B_{12}^\top\,D+D^\top\,B_{12}=\begin{pmatrix}
            2\,z & 0 \\
            0 & 0
        \end{pmatrix}\,.
\end{align}
We have
\begin{equation}
    \mathcal{K}_{22}\preceq 0\label{Lpsd:1}
\end{equation}
and
\begin{equation}
    (I-\mathcal{K}_{22}\,\mathcal{K}_{22}^\dagger)\,\mathcal{K}_{12}^\top=\begin{pmatrix}
        0 & 0 \\
        0 & 1
    \end{pmatrix}\,\begin{pmatrix}
        c & w \\
            0 & 0
    \end{pmatrix}=\begin{pmatrix}
        0 & 0 \\
        0 & 0
    \end{pmatrix}\,.\label{Lpsd:2}
\end{equation}
Also, consider
\begin{align}
    \mathcal{K}_{11}-\mathcal{K}_{12}\,\mathcal{K}_{22}^\dagger\,\mathcal{K}_{12}^\top &=\begin{pmatrix}
            -2\,\beta & \dfrac{\beta\,z}{k} \\
            \dfrac{\beta\,z}{k} & \dfrac{\phi\,z}{k}
        \end{pmatrix}-\begin{pmatrix}
            c & 0 \\
            w & 0
        \end{pmatrix}\,\begin{pmatrix}
            \dfrac{1}{2\,z} & 0 \\
            0 & 0
        \end{pmatrix}\,\begin{pmatrix}
            c & w \\
            0 & 0
        \end{pmatrix} 
        \\ & 
        =\begin{pmatrix}
            -2\,\beta-\dfrac{c^2}{2\,z} & 0 \\
            0 & \dfrac{\phi\,z}{k}-\dfrac{w^2}{2\,z}
        \end{pmatrix}\,.
\end{align}
Due to \eqref{range:z}, we have
\begin{equation}
    -2\,\beta-\frac{c^2}{2\,z}<-2\,\beta+\frac{c^2}{2}\cdot\frac{2\,\beta}{c^2}=-\beta<0
\end{equation}
and
\begin{align}
    \frac{\phi\,z}{k}-\frac{w^2}{2\,z}=\frac{\phi\,z}{k}-\frac{2\,\beta^2\,z^3}{c^2\,\eta^2}=\frac{z}{k}\left(\phi-\frac{2\,\beta^2\,z^2}{c^2\,\eta}\right)<\frac{z}{k}\left(\phi-\frac{2\,\beta^2}{c^2\,\eta}\cdot\frac{\phi\,c^2\,\eta}{2\,\beta^2}\right)=0,
\end{align}
so
\begin{equation}
    \mathcal{K}_{11}-\mathcal{K}_{12}\,\mathcal{K}_{22}^\dagger\,\mathcal{K}_{12}^\top\prec 0\,.\label{Lpsd:3}
\end{equation}
Combining \eqref{Lpsd:1}, \eqref{Lpsd:2}, and \eqref{Lpsd:3}, we conclude
\begin{equation}
    \mathcal{L}+\mathcal{L}^\top\preceq0\,.
\end{equation}
Moreover, 
\begin{equation}
    C+D\,G+G^\top\,D^\top=\begin{pmatrix}
       1  & 0 \\
        0 & w+\dfrac{c\,z}{k}
    \end{pmatrix}\succ0\,,
\end{equation}
since \eqref{range:z} implies
\begin{equation}
    w+\frac{c\,z}{k}=\frac{2\,\beta\,z^2}{c\,\eta}+\frac{c\,z}{k}=z\,\left(\frac{2\,\beta\,z}{c\,\eta}+\frac{c}{k}\right)>z\,\left(-\frac{2\,\beta}{c\,\eta}\cdot\frac{c^2}{2\,\beta}+\frac{c}{k}\right)=0\,.
\end{equation}
By Theorem 2.3 in \cite{freiling2000}, DRE \eqref{eq:dreP} has a unique solution.

\qed

\subsubsection{Proof of Proposition \ref{prop:valuefndfd}}\phantomsection\label{proof:prop:valuefndfd}

\vspace{.3cm}\noindent
This proof follows the same steps as in the proof of Proposition~\ref{prop:valuefndfd no transient} in Appendix~\ref{proof:prop:valuefndfd no transient}, except that we need to show the modified $J[\nu^\star]$ is continuous in $\kappa$.

To that end, fix $\kappa_n\to\kappa$. Due to \eqref{coercive Lambda}, $\Lambda$ is coercive, so Lax-Milgram lemma implies $\Lambda$ has an inverse $\Lambda^{-1}$, which is a bounded linear operator on $\A_2$, and $\nu^\star$ is given by $\nu^\star=\Lambda^{-1}b$, so
\begin{equation}
    J[\nu^\star]=-\frac{1}{2}\,\left\langle\Lambda\Lambda^{-1}b,\Lambda^{-1}b\right\rangle+\left\langle b,\Lambda^{-1}b\right\rangle=\frac{1}{2}\,\left\langle b,\Lambda^{-1}b\right\rangle\,,
\end{equation}
where
\begin{equation}
    b=\mathfrak{I}^\top(F\,G)+(c-\beta\,\mathfrak{I}^\top-\phi\,\mathfrak{Q}^\top)(Y-Y_0)+\mathfrak{Q}^\top(A\,F)\,.
\end{equation}
Since
\begin{align}
    |G_t(\kappa_n)-G_t(\kappa)|&\leq\;|A_t|\,|\partial_1 h(F_t,\kappa_n)-\partial_1 h(F_t,\kappa)|+\frac{\sigma^2}{2}F_t\,|\partial_{11}h(F_t,\kappa_n)-\partial_{11}h(F_t,\kappa)|
    \\
    &\leq\;\left(|A_t|+\frac{\sigma^2}{2}\,F_t\right)\,\left(F_t^{\mathfrak{p}}+F_t^{\mathfrak{q}}\right)\,|\mathfrak{C}(\kappa_n)-\mathfrak{C}(\kappa)|\,,
\end{align}
we have
\begin{align}
    &\hspace*{-2em}\Vert \mathfrak{I}^\top(F\,G(\kappa_n))-\mathfrak{I}^\top(F\,G(\kappa))\Vert
    \\
    &\leq \left\Vert\mathfrak{I}^\top\right\Vert_{\operatorname{op}}\,\E\!\left[\int_0^TF_t^2\,|G_t(\kappa_n)-G_t(\kappa)|^2\,\dt\right]^{\tfrac{1}{2}}
    \\
    &\leq \left\Vert\mathfrak{I}^\top\right\Vert_{\operatorname{op}}\,\E\!\left[\int_0^T\left(|A_t|+\frac{\sigma^2}{2}\,F_t\right)^2\,\left(F_t^{\mathfrak{p}+1}+F_t^{\mathfrak{q}+1}\right)^2\,|\mathfrak{C}(\kappa_n)-\mathfrak{C}(\kappa)|^2\,\dt\right]^{\tfrac{1}{2}}
    \\
    &\leq\left\Vert \mathfrak{I}^{\top}\right\Vert _{\operatorname{op}}\,|\mathfrak{C}(\kappa_{n})-\mathfrak{C}(\kappa)|\\&\quad\left\{ \E\,\left[\int_{0}^{T}|A_{t}|^{p}\,\dt\right]^{\tfrac{1}{p}}\,\left(\left\Vert F^{\tfrac{(\mathfrak{p}+1)\,p}{p-2}}\right\Vert ^{\tfrac{p-2}{p}}+\left\Vert F^{\tfrac{(\mathfrak{q}+1)\,p}{p-2}}\right\Vert ^{\tfrac{p-2}{p}}\right)+\frac{\sigma^{2}}{2}\,\left\Vert F^{\mathfrak{p}+2}+F^{\mathfrak{q}+2}\right\Vert \right\} \,,
\end{align}
so $\kappa\mapsto\mathfrak{I}^\top(F\,G(\kappa))$ is continuous. Since we've already known $\kappa\mapsto Y(\kappa)-Y_0(\kappa)$ from $(0,\infty)$ to $\A_2$ is continuous,
\begin{equation}
    \kappa\mapsto b(\kappa)=\mathfrak{I}^\top(F\,G(\kappa))+(c-\beta\,\mathfrak{I}^\top-\phi\,\mathfrak{Q}^\top)(Y(\kappa)-Y_0(\kappa))+\mathfrak{Q}^\top(A\,F)
\end{equation}
is continuous. It follows that $\kappa\mapsto J[\nu^\star](\kappa)$ is continuous.
\qed

\clearpage

\bibliographystyle{plainnat}
\bibliography{references}

\end{document}